\documentclass[12pt]{book}

\usepackage{amsfonts}
\usepackage{amssymb}
\usepackage{amsthm}
\usepackage{CoRR_format}
\usepackage{epsfig}
\usepackage{url}

\DeclareSymbolFontAlphabet{\mathbb}{AMSb}

\long\def\ignore#1{}

\newcommand{\ie}{\emph{i.e.}}
\newcommand{\eg}{\emph{e.g.}}

\newcommand{\lp}{\lambda{\it Prolog}}
\newcommand{\Ll}{L_\lambda}
\newcommand{\dg}[2]{\langle #1 , #2 \rangle}
\newcommand{\tuple}[1]{\langle #1 \rangle}

\newcommand{\app}{{\ }}
\newcommand{\lambdadb}{\lambda \,}
\newcommand{\lambdax}[1]{\lambda #1 \,}

\newcommand{\dum}[1]{@ #1}
\newcommand{\lenv}{{\lbrack\!\lbrack}}
\newcommand{\renv}{{\rbrack\!\rbrack}}
\newcommand{\env}[1]{{\lenv #1 \renv}}

\newcommand{\allx}{\forall}
\newcommand{\somex}{\exists}
\newcommand{\andxy}[2]{#1 \wedge #2}
\newcommand{\orxy}[2]{#1 \vee #2}
\newcommand{\conj}{\wedge}
\newcommand{\disj}{\vee}
\newcommand{\imp}{\supset}
\newcommand{\impxy}[2]{#1 \supset #2}

\newcommand{\arrxy}[2]{#1\, \rightarrow \, #2}
\newcommand{\ra}{\rightarrow}
\newcommand{\lra}{\longrightarrow}
\newcommand{\lras}{\stackrel{*}{\lra}}

\newcommand{\pif}{{\tt :\!-}}
\newcommand{\pimp}{=>}
\newcommand{\plam}{\backslash}

\newcommand{\qf}{{\it quantifier}\_{\it free}}
\newcommand{\ia}{{\it is}\_{\it atomic}}
\newcommand{\itm}{{\it is}\_{\it term}}

\newcommand{\rnf}[1]{{\vert #1 \vert}}

\newcommand{\dquad}{\quad\quad}
\newcommand{\squad}{\quad\quad\quad\quad\quad\quad}

\newcommand{\st}{\mathcal}

\newtheorem{defn}{Definition}[section]
\newtheorem{prop}{Proposition}[section]
\newtheorem{theorem}{Theorem}[section]

\begin{document}

\setcounter{secnumdepth}{5}

\Author{Xiaochu Qi}

\Title{An Implementation of the Language Lambda Prolog Organized
around Higher-Order Pattern Unification}

\Month{October}

\Year{2009}

\Adviser{Dr. Gopalan Nadathur}

\titlepage
\copyrightpage

\pagenumbering{roman}

\setcounter{tocdepth}{2}

\pagenumbering{roman}
\abstract{

The automation of meta-theoretic aspects of formal systems typically
requires the treatment of syntactically complex objects. Thus,
programs must be represented and manipulated by
program development systems, mathematical
expressions by computer-based algebraic systems, and logic formulas
and proofs by automatic proof systems and proof assistants.
The notion of bound variables plays an important role in the
structures of such syntactic objects, and should therefore
be reflected in their representations and properly accounted for in
their manipulation.
The $\lambda$-calculus was designed specifically to treat binding in a
logically precise way and the terms of such a calculus turn out to be
an especially suitable representational
device for the application tasks of interest.
Moreover, the equality relation associated with these terms and the
accompanying notion of higher-order unification leads to a convenient
means for analyzing and decomposing these representations in a way
that respects the binding structure inherent in the formal objects.

This thesis concerns the language $\lp$ that has
been designed to provide support for the kinds of meta-programming
tasks discussed above. In its essence, $\lp$ is a logic
programming language that builds on a conventional language like
Prolog by using typed $\lambda$-terms instead of first-order terms as
data structures, by using higher-order unification rather than
first-order unification to manipulate these data structures and by
including new devices for restricting the scopes of names and of code
and thereby providing the basis for realizing recursion over binding
constructs. These features make $\lp$ a convenient programming vehicle
in the domain of interest. However, they also raise significant
implementation questions that must be addressed adequately if the
language is to be an {\it effective} tool in these contexts. It is
this task that is undertaken in this thesis.

An efficient implementation of $\lp$ can potentially exploit the
processing structure that has been previously designed for realizing
Prolog. In this context, the main new issue to be treated becomes that
of higher order unification. This computation has
characteristics that make it difficult to embed it effectively within a
low-level implementation: higher-order unification is in general
undecidable, it does not admit a notion of most general unifiers and a
branching search is involved in the task of looking for
unifiers. However, a sub-class of this computation that is referred to
as $\Ll$ or higher-order pattern unification has been discovered
that is substantially better behaved: in particular, for this class,
unification is decidable, most general unifiers exist and a
deterministic unification procedure can be provided. This class is
also interesting from a
programming point-of-view: most natural computations carried out
using $\lp$ fall within it. Finally, a
treatment of full higher-order unification within the context of
$\lp$ can be realized by solving only higher-order pattern
unification problems at intermediate stages, delaying any branching
and possibly costly search to the end of the computation.

This thesis examines the use of the strategy described above in
providing an implementation of $\lp$. In particular, it develops a new
virtual machine and compilation based scheme for the language by
embedding a higher-order pattern unification algorithm due to Nadathur
and Linnell within the well-known Warren Abstract Machine model for
Prolog. In executing this idea, it exposes and treats various
auxiliary issues such as the low-level representation of
$\lambda$-terms, the implementation of reduction on such terms, the
optimized processing of types in computation and the representation of
unification problems whose solution must be deferred till a later
point in computation.  Another important component of this thesis is
the development of an actual implementation of $\lp$---called {\it
Teyjus Version 2}---that is based on the conceptual design that is
presented. This system contains an emulator for the virtual machine
that is written in the C language for efficiency and a compiler that
is written in the OCaml language so as to enhance readability and
extensibility. This mix of languages within one system raises
interesting software issues that are handled. Portability across
architectures for the emulator is also treated by developing a modular
mapping from term representation to actual machine structures. A final
contribution of the thesis is an assessment of the efficacy of the
various design ideas through experiments carried out with the
assistance of the system.
}
\acknowledgments{I take this opportunity to express my gratitude to my advisor Dr.
Gopalan Nadathur for his guidance, support, valuable instruction and
encouragement throughout the entire project. Thanks are due to
Zach Snow, Steven Holte and Andrew Gacek for their help in the development
and maintenance of the system {\it Teyjus version 2}. I am also grateful to
Dr. Dale Miller, Dr. Eric Van Wyk, Dr. Mats Heimdahl and Dr. Williams
Messing for their interest on this thesis.

Work on this thesis has been partially supported by the NSF Grants
CCR-0096322 and CCR-0429572. I have also received support in the
course of my graduate studies from  a Grant-in-Aid of Research
provided by the Graduate School at the University of Minnesota
and by funds provided by the Institute of Technology and the
Department of Computer Science and Engineering at the University of
Minnesota. Opinions, findings, and conclusions or recommendations
expressed in this thesis should be understood as mine. In particular,
they do not necessarily reflect the views of the National Science
Foundation.
}
\tableofcontents \clearpage \listoffigures \pagenumbering{arabic}

\textpages
\singlespace
\chapter{Introduction}\label{chp:intro}
This thesis is concerned with the implementation of a higher order
logic programming language called $\lp$. This language is of interest
because it provides perspicuous and effective ways for
realizing computations over formal objects such as programs,
mathematical expressions,
logical formulas, and proofs. Computations of this kind are frequently
needed in meta-level application tasks such as those involved in
building
program development systems~\cite{HuetLang78},
automated algebraic systems~\cite{deBruijn80, Cons86},
automatic reasoning systems~\cite{CH88,HHP93}, and
proof assistants~\cite{BGMNT07,CoqMan97,gacek08ijcar,Paulson94}. In this
chapter we motivate the $\lp$ language from the perspective of such
applications, explain what is involved in implementing it well and
then characterize the contributions of this thesis.

\section{Using $\lambda$-terms as Data Structures}
An important first step in building systems that manipulate formal
objects is the design of a convenient representation for such
objects. When we examine the specific programming tasks, it turns out
that in many of them there is a need to deal with syntactic constructs
that involve a notion of binding. As an example, consider a theorem
proving system that manipulates quantificational formulas. When
representing a formula such as $\forall x P(x)$, where $P(x)$ denotes
an arbitrary formula in which $x$ may appear free, it is necessary to
capture the scoping aspect of the quantifier as well as the fact that
the particular choice of name for the quantified variable is not
significant. These properties will be necessary, for example, in
correctly instantiating the quantifier when needed---we have to be
careful not to substitute terms for $x$ which contain variables that
get captured by quantifiers appearing further inside $P(x)$---and in
recognizing that the formula $\forall x P(x)$ is really the same as
$\forall y P(y)$. Similar observations can be made with respect to the
representation of programs in a program manipulation system. Here,
it is necessary to encode functions in such a way that the
binding aspects of arguments and issues of scope are clearly
recognized in the course of analyzing and transforming their
structures. Some of these aspects can be illustrated by considering the
simple setting of the $\lambda$-calculus that underlies the idea of
functions in programming languages. Suppose, for example, that we want
to write an evaluator for the $\lambda$-calculus. In this setting, we
have to be able to transform an expression of the form
$((\lambdax{x} M)\ N)$ into one that is obtained by replacing the free
occurrences of $x$ in $M$ by $N$. In carrying out this operation, we
have to be able to distinguish free occurrences of $x$ from the bound
ones and we also have to be careful to not allow any free variable
in $N$ to be captured by an abstraction within $M$. Moreover, a
prerequisite for applying such a transformation is that we have to be
able to recognize that a term has the form $((\lambdax{x} M)\ N)$ even if
the abstracted variable in the ``function part'' is not exactly named
$x$.

A careful examination of the examples discussed above shows that even
though the application domains are quite different, there is a common
part to what needs to be treated with regard to binding in both
cases. The important aspects of binding can in fact be uniformly
captured by using the terms of the $\lambda$-calculus as a
representational mechanism. For example, the concept of the scope of a
binding is explicitly reflected in the structure of a
$\lambda$-abstraction. Similarly, the recognition of the irrelevance of
names for bound variables and the preservation of their scopes during
substitution are manifest though the usual $\lambda$-conversion rules.
Thus,
the representation of formal objects relevant to different contexts
can be accomplished by using $\lambda$-abstractions to capture the
underlying binding structures and using constructors like in
first-order abstract syntax representations to encode whatever context
specific semantics is relevant to the analysis.

As an illustration of this idea, consider the formula $\allx{x} P(x)$
mentioned in the theorem proving example. This formula can be
represented by the $\lambda$-term $(all\app(\lambdax{x} \overline{P(x)}))$,
where $all$ is a constructor chosen to denote the universal quantifier, and
$\overline{P(x)}$ represents, recursively, the formula $P(x)$. This
representation separates out the two different roles of a universal
quantifier, one of which corresponds to imposing the ``for all"
semantics and the other that indicates the scope of the
quantification, and it captures the latter explicitly through a
$\lambda$-abstraction. Using this representation, the instantiation of
the universal quantifier of the given formula can be simply denoted as
an $\lambda$-application of form $((\lambdax{x} \overline{P(x)})\app t)$,
where $t$ is the representation of the object-level term that the
quantifier is to be instantiated
with. This ``application term'' is equivalent under the rules of
$\lambda$-conversion to a term that results from replacing each
occurrence of $x$ in $\overline{P(x)}$ with $t$ being careful, of
course, to avoid any inadvertent capture of free variables in $t$.
Similarly, the object-level $\lambda$-term $((\lambdax{x} M)\ N)$ can be
represented by the expression $(app\ (abs\ (\lambdax{x} \overline{M}))\
\overline{N})$; notice that $app$ and $abs$ are constructors
chosen to encode object language application and abstraction in this
representation, and the binding effect of an object-level
abstraction is captured by an abstraction of the meta-language.
With this kind of representation, we can describe the evaluation rule
that was of interest earlier as
simply that of rewriting an expression of the form $(app\ (abs\ T)\
R)$ to the form $(T\ R)$; the meta-language understanding of
$\lambda$-terms ensures then that the required substitution operation
will be carried out in a logically correct manner.

Our interest in this thesis is in a language for carrying out
computations over formal objects. From this perspective, what we
desire is a language that allows us to use $\lambda$-terms as a means
for representing objects and that provides primitives for
manipulating these in a logically meaningful way.
The logic programming language $\lp$~\cite{NM88} is one of this sort.
It is based on a higher-order logic built around a typed version of the
$\lambda$-calculus. The presence of $\lambda$-terms as basic data
structures in this language provides the convenience discussed earlier
in this section in representing formal objects,
and therefore renders the language an especially
suitable tool to describe formal systems.
This language attributes operational semantics to logical
connectives and quantifiers, so that these logical symbols can also be
viewed as programming primitives. As a result, the language allows for
the construction of descriptions of formal systems that can be viewed
as specifications but that are also executable as programs. In
comparison with usual logic programming languages, $\lp$
provides two new logical devices for specifying the scopes of names
and of clauses defining predicates. From the programming perspective,
these devices turn out to be helpful in describing recursive
computations over binding structure.
Many uses have been made of these various features of
$\lp$ in describing interesting computations over formal
objects; see, for example, \cite{AppFel99, Felty93jar,  Hannan90,
  NM88, Pareschi89phd, Pfe88}.
These kinds of applications motivate the development of an efficient
implementation of this language, a topic that is the focus of this
thesis.

\section{Using Higher-Order Unification for Computation}

An important part of the computational machinery underlying $\lp$ is a
realization of unification over $\lambda$-terms. This form of
unification, known as higher-order unification, differs from the one
used in a language like Prolog in that equality between terms is based
not just on identity but also on the conversion rules of
the $\lambda$-calculus. Pragmatically, this operation is the basis for
analyzing the shapes of syntactic structures that involve binding: for
example, it is this form of unification that allows $\forall x(P(x)
\land Q(x))$ to be used as a template for matching with formulas that have a
particular form and for decomposing them into the parts corresponding
to the conjuncts embedded inside the universal quantification if they
do have this
form. Most existing implementations of $\lp$ realize higher-order unification
based on a procedure described by Huet. While higher-order unification
seems a necessary operation within $\lp$, it unfortunately also turns
out to be one that has poor theoretical properties. For example, it
does not admit most general unifiers, a possibly redundant search may be
involved in calculating unifiers and unifiability is, in the limit,
undecidable. These kinds of properties manifest themselves in Huet's
procedure by giving it a non-deterministic branching structure, by
restricting it to calculating only pre-unifiers so as to avoid
redundancy and by making
it a possibly non-terminating computation. Embedding such a procedure
within a larger language implementation is difficult and can also make
it difficult to realize other associated operations in an efficient
manner.

While the situation with employing higher-order unification in a practical
way seems difficult at first sight, the signs from looking at actual
attempts to employ it is much more hopeful. In particular,
from using a system realizing $\lp$ based on Huet's
procedure~\cite{NM99cade}, and also from using other logical
frameworks and proof assistants  such as
{\it Twelf}~\cite{Pfenning99systemdescription} and {\it
  Isabelle}~\cite{nipkow02book} 
that employ higher-order unification, it
becomes evident that there is a large collection of practically
relevant meta-programming tasks in which the relevant higher-order
unification problems actually have unique solutions that can be
completely revealed even by using Huet's pre-unification procedure.
Based on a study of the usage of higher-order unification in
these examples, Dale Miller has identified a subset
of the general problems, known as the $\Ll$
or the higher-order pattern class~\cite{Miller91jlc, Nipkow93}, which covers
the major cases of the unification problems occurring in practice~\cite{MP92}.
Unifiability on this subset is known to be decidable and it is also
known that a single most general unifier can be provided in any of the
cases where a unifier exists.
In fact, Miller has described a (non-deterministic) algorithm for
solving higher-order pattern unification problems that has the
characteristic of either determining non-unifiability
or producing a most general unifier at the end. The idea underlying
this procedure have been extended to dependently typed
$\lambda$-calculi~\cite{PF91, Pfenning91unificationand} and
higher-order rewrite systems~\cite{NT91}. 

It turns out that Huet's procedure is also effective when applied to
higher-order pattern unification problems in that it is guaranteed to
terminate and will do so with a unique successful branch. One may
wonder therefore if there is any purpose to describing specialized
unification procedures for this subclass and it is important to
address this question to put the work in this thesis in
perspective. There are, in fact, particular pragmatically significant
ways in which the behavior of Huet's procedure can be improved by
taking the restriction seriously. First, even though the
(pre)-solution found is unique, Huet's procedure conducts a branching
search to find it; it must do this since it needs to also address
more general higher-order unification problems. It turns out that if
one is not concerned about covering the larger class then the
intermediate steps can also be made deterministic. Second, even when
restricted to the higher-order pattern fragment, Huet's procedure
is guaranteed only to find pre-unifiers; in some instances, it will
return with a substitution and a remaining solvable problem but one that
it chooses not to solve. When focusing only on the higher-order
pattern unification class, however, it is possible to provide a
different unification algorithm that will solve the problem
entirely. Finally, the structures of solutions to general higher-order
unification problems depend on the types of the the terms being
unified, and consequently Huet's procedure examines these types during
computation. However, for the higher-order pattern fragment, it is
possible to structure computation so that it does not depend on type
information. This has a practical significance since it is, in
general, an expensive proposition to compute and carry around type
annotations with terms during execution.

Several implementations have been described of $\lambda$Prolog prior
to this thesis and one of them, {\em Teyjus Version 1}
\cite{NM99cade}, even considers
a compilation-based realization that is borrowed from heavily in this
thesis. All these implementations embed
within them Huet's treatment of higher-order unification. The
distinguishing feature of the work described here is that it analyses the
implementation of $\lambda$Prolog based on a model that treats only
higher-order pattern unification. The observations in this section
indicate a merit to considering this question: the higher-order
pattern fragment is practically relevant and restricting to only this
class can have an impact on the computational model that is important
to understand.

\section{Contributions of the Thesis}
This thesis explores the idea of orienting an implementation of
$\lp$ around a particular higher-order pattern
unification algorithm---the one proposed by Nadathur and Linnell
\cite{NL05HOP}. More specifically, it considers the full $\lp$
language, \ie, it does not restrict the syntax of this language in
any way. However, when unification problems are encountered, they are
solved completely only if they fall within the $\Ll$ fragment; more
general problems are deferred and later solved only if
instantiations convert them into ones in this subset.

The implementation that is developed is based on using a special
abstract machine for $\lp$ and on compiling programs in the
language into instructions for this machine.
The basic framework for the machine is provided by the WAM, the
abstract machine that D.~H.~D.~Warren designed for
Prolog~\cite{Warren83WAM}. The main new challenge in this work is to
embed pattern unification into the WAM that was originally designed to
treat only first-order terms.\footnote{In comparison with Prolog,
$\lp$ has additional search primitives
and also permits a quantification over predicates. However, we add nothing new to the
treatment of these aspects, simply inheriting them from {\it Teyjus
  Version 1} that is discussed later.}
There are several issues that must be considered in realizing
such an embedding in a practically acceptable fashion.
One class of such issues arises from the fact that a richer class of
terms---the terms of a $\lambda$-calculus instead of just first-order
terms---have to be represented and manipulated. The machine
representation that is chosen for such terms should, at the outset,
facilitate an efficient equality examination between terms based on
$\lambda$-conversions; in particular, it should support well the
recognition of equality between terms that differ only in the names of
bound variables and should also provide an efficient realization of
$\beta$-reduction or function evaluation. Beyond this, it is important
to treat efficiently the typical decomposition of terms that is needed in the
course of pattern unification.
For example, it is often necessary to get quickly to the head of a
term and this is best realized in a scheme that represents nested
applications in  a form that collects the successive arguments into a
vector form and directly exposes the embedded head.
A similar argument can be made for collecting a sequence of
abstractions into a single abstraction over several variables.
A second issue that needs to be treated is the seamless integration of
the richer higher-order pattern unification into the compiled
treatment of first-order unification that is the hallmark of the
WAM. In this context, we note that the treatment must also contain
within it a suitable mechanism for delaying higher-order unification
problems that do not fall within the higher-order pattern class. A
final issue that
we mention here is that of treating the polymorphic typing regime that
is part of the $\lp$ language. A consequence of this
polymorphism is that the particular type instances must be known when
comparing two constants that otherwise have the same name; the
ultimate identity of these constants must, in this case, be based on
an equality of their types.
Although the pattern unification does not need types in deciding
the structures of unifiers, the role of types mentioned above makes
it necessary to sometimes examine these dynamically to decide
unifiability of terms.
An efficient runtime type processing scheme should then be provided,
in which types are maintained and examined only for the
identity checking of constants.

This thesis addresses these various issues and proposes solutions to
them.
Towards providing an efficient realization of reduction over
$\lambda$-terms, it exploits the idea of an explicit substitution
notation for such terms \cite{ACCL91, NW98tcs}. It further considers
particular reduction procedures that can be used with such a
representation towards getting the best time and space performance.
To treat equivalence under bound variable renaming, it uses a
nameless representation for such variables in the style of de Bruijn
\cite{debruijn72}.
The ability to treat substitutions explicitly is exploited in
distributing this operation over the steps that need to be performed
in realizing unification towards minimizing redundant computations.
The low-level representation of terms in the explicit substitution
form pays attention to how applications and abstractions are
encoded so as to obtain fast access to the components that need
to be examined often in the course of unification.
The instruction set of the WAM is enhanced towards integrating the
treatment of higher-order pattern unification into the standard
compilation model. The particular approach that is used here is to
develop these instructions so that first-order unification is still
treated via compilation whereas the new components in higher-order
pattern unification lead to the invocation of an interpretive
phase. When parts of the unification problem falls outside the
higher-order pattern class, these are carried into subsequent
computations in the form of constraints that may be addressed later.
A practical representation is proposed for such residual problems and
the addition to these problems as well as their re-examination is
integrated into the instructions for the abstract machine. Finally,
the issue of runtime type processing is treated by first developing a
static analysis process that reduces the footprint of such types
considerably and then including instructions in the abstract machine
to treat the remaining aspects as much as possible through compiled
code.

In addition to proposing an implementation scheme for the
$\lp$ language, this thesis also develops an actual
implementation of the conceptual design that it produces. A
characteristic of this part of the work is a careful attention to the
issue of portability across varied architectures and operating
systems. Towards this end, a modular method is developed for mapping
the abstract machine onto the low-level hardware on which it is
emulated. Another aspect to which close consideration has been given is
that of enhancing the flexibility and expandability of the
implementation. To realize this goal, an attempt is made to use as
much as possible a high-level language---here the language {\em
  OCaml}---in the implementation, employing the language {\em C} only
in realizing those parts
whose efficiency depends critically on the closeness to the underlying
hardware. This mix of implementation languages raises interesting
problems of its own that we discuss later in the thesis. We note
finally that having an implemented system gives us the ability to test
the efficacy of our various design ideas, a topic that we also consider
on in this thesis.

In summary, the contributions of this thesis are threefold:
\begin{enumerate}
\item The design of an abstract machine and associated compilation
  methods for treating the $\lp$ language. A key characteristic of
  the abstract machine that is developed is that it attempts to
  exploit the efficiencies that arise out of focusing on higher-order
  pattern unification rather than treating more general forms of
  unification for $\lambda$-terms.

\item An actual implementation of $\lp$---{\it Version 2} of the {\it Teyjus}
  system---based on the virtual machine and compilation scheme
  developed. This implementation has proven to be extremely portable
  and also combines components written in the {\em C} and the {\em
    OCaml} languages towards enhancing openness and expandability in
  its structure.

\item A study of the performance impact of using higher-order pattern
  unification, optimized runtime types processing and other related
  design ideas. This study is based on experiments conducted with {\it
  Teyjus Version 2} using practical $\lp$ programs that exploit the
  meta-programming capabilities of the language.
\end{enumerate}

Prior to the work of this thesis, another abstract machine that is
organized around Huet's unification procedure has been designed for
$\lp$~\cite{KNW94cl, N03treatment, NJK95lp}.
This abstract machine has in fact provided the basis
for {\em Version 1} of the {\it Teyjus} system that we have mentioned
earlier.
Many challenges faced in realizing the new search primitives and
higher-order features present in $\lp$ were considered for the first
time in the context of that work and the design presented in this
thesis has been influenced by the ideas developed there.
However, the work undertaken in this thesis differs significantly
from the previous design and implementation in that it takes seriously
the idea of realizing a higher-order logic programming language through
the narrower mechanism of treating higher-order pattern
unification. In particular, it examines carefully the impact of this
decision on various aspects of the structure of the abstract machine
and of the efficiency of implementation. An auxiliary aspect of this
work is that it has resulted in a system that is far more portable and
expandable because of the particular approaches that have been used in
its implementation.

A central idea underlying this thesis is that of approaching
higher-order unification through higher-order pattern unification. It
is important to stress that the use of this idea is by itself not
novel to our work: in particular, this idea has been employed
previously in the proof assistant {\it Isabelle} \cite{nipkow02book}
and in the logical framework {\em Twelf}
\cite{Pfenning99systemdescription}. The particular deployment of this
idea in the 
{\it Isabelle} system is, in our understanding, quite different from
the method we use in this thesis: {\em Isabelle} first 
tries to solve unification problems by means of a 
higher-order pattern unification procedure and, if this does not
succeed, it then falls back to full higher-order unification. By
contrast, the method we use is quite similar to that employed within
{\em Twelf}: in both cases, a higher-order pattern unification
procedure is all that is used and problems that do not fall within the
class that this procedure is capable of handling are deferred till a
later point in the computation. The distinguishing characteristic of
our work in this context is that it explores the impact of
this idea on the design of an abstract machine and compilation
model for the underlying logic programming language. Another aspect of
our work is that it attempts to quantify the 
benefits of using this approach through a head-to-head comparison with
an implementation that uses Huet's unification procedure directly in
implementation.

\section{Organization of the Thesis}

The rest of the thesis is organized as the follows.
Chapter~\ref{chp:language} provides an overview of the $\lp$ language.
The discussion here illustrates the usefulness of the higher-order
features of the language
in describing formal systems and provides an intuitive understanding
on the underlying computation model.
Chapter~\ref{chp:termRep} describes the notion of equality of $\lambda$-terms
that is based on $\lambda$-conversion. This chapter also introduces an explicit
substitution based representation of such terms, which facilitates
efficient term comparison based on the relevant notion of equality.
An abstract interpreter for the $\lp$ language is presented in
Chapter~\ref{chp:interpreter} for the purpose of formally defining
the model of computation underlying this language. The role unification
plays in this computational model is discussed and a practical higher-order
pattern unification algorithm is introduced.
The low-level term representation scheme used in the implementation
developed by this thesis is discussed in Chapter~\ref{chp:machineTermRep}.
This discussion includes the presentation of an algorithm that
efficiently realizes
$\beta$-reduction based on the explicit substitution representation discussed
in Chapter~\ref{chp:termRep}. Also presented in the chapter is a
refinement to the term representation geared towards providing
fast access to the subcomponents that are needed by the term
decomposition operations used within pattern unification.
Chapter~\ref{chp:compile} describes a compilation based implementation
of the $\lp$ language. A detailed discussion is included of the way in
which pattern unification can be integrated into the WAM-based
computation model.
In Chapter~\ref{chp:types}, an efficient runtime type processing
scheme is proposed together with accompanying static optimization
processes and their integration into the compilation based processing model.
The actual software system, {\em Teyjus Version 2}, that realizes the
conceptual design ideas of this thesis is the focus of
Chapter~\ref{chp:system}. This
chapter also discusses  the practical issues faced in the
implementation of this software system, such as the realization of
the properties of portability and openness of code.
An assessment of the design and of the performance of {\em Teyjus Version 2}
is the topic of Chapter~\ref{chp:expr}. Experimental
data is presented and analyzed here towards providing a quantitative
understanding of the impact of our conceptual design ideas.
Finally, Chapter~\ref{chp:conclusion} concludes the thesis with a
discussion of some future directions.

\chapter{The $\lambda$Prolog Language}\label{chp:language}

In this chapter we provide an overview of the higher-order logic
programming language $\lp$ whose implementation will be the subject of
the rest of the thesis. The foundation for this language is provided
by a subclass of formulas in an intuitionistic version of Church's
higher-order logic \cite{Church40}. This class of formulas, known as
{\em higher-order hereditary Harrop formulas}, enhances the
collection of first-order Horn clauses that underlie conventional logic
programming languages like {\em Prolog} in several significant
ways. In particular, the enriched formulas allow the arguments of
predicates to be $\lambda$-terms rather than just first-order terms,
they permit implications and universal quantifiers to be used in
queries thereby giving rise to new search primitives and they support
higher-order programming by including quantification over function and
(limited occurrences of) predicate symbols. By exploiting these
additions, $\lp$ provides strong support for what has come to be called
the {\em higher-order abstract syntax} approach to representing formal
syntactic objects~\cite{Pfenning88pldi2}.

Our presentation of $\lp$ below mixes a description of its theoretical
basis with a feeling for programming in the language. In
Section~\ref{sec:lambda_terms} we recall the simply typed
$\lambda$-calculus upon which the higher-order logic of interest is
based, presenting these terms in a way a $\lp$ user would encounter
them. In Section~\ref{sec:hohh} we introduce the higher-order
hereditary Harrop formulas and also describe at a high level the
computational interpretation that $\lp$ associates with these
formulas. In~Section~\ref{sec:language_example} we illustrate the idea
of higher-order abstract syntax and the support that $\lp$ provides
for this approach by considering an extended example. The discussion
in the first three sections assumes a simple monomorphic typing system
with $\lp$. In reality, the language allows for polymorphic typing. We
discuss this aspect in the last section of this chapter.

\section{The Simply Typed $\lambda$-Calculus}\label{sec:lambda_terms}
The logic underlying $\lp$ is based on a polymorphically typed version of the
simply-typed $\lambda$-calculus.
The types used in this calculus are constructed from sorts and type variables
by recursive applications of type constructors.
For simplicity, we initially restrict our attention to a simple,
monomorphically typed version of this calculus by leaving out the
usage of type variables. We eventually add these type variables to the
language in Section~\ref{sec:types_in_computation}.

In the interpretation used here, we assume given a set of sorts and
another set of type constructors each element of which is specified
with an arity. The types in the language are then described
through the following rules:
\begin{enumerate}
\item Each sort $s$ is a type;
\item $(c\app\tau_1\ ...\ \tau_n)$ is a type provided $c$ is a type
constructor of arity $n$, and $\tau_1$, ..., $\tau_n$ are types;
\item If $\tau_1$ and $\tau_2$ are types, then $\arrxy{\tau_1}{\tau_2}$ is a type.
\end{enumerate}
The type defined by the last rule is viewed as a function type, where $\ra$
is called the function type constructor. Types other than function types are
called atomic.
In the following discussions, the usage of parentheses is minimized by
assuming that $\ra$ is right associative and has a lower priority than
other type constructors. Under these assumptions, any function type can
be elaborated as $\alpha_1\ra...\ra\alpha_n\ra\beta$, where $\beta$ is
atomic. We call $\alpha_1$, ..., $\alpha_n$ the argument types of such
a type and we refer to $\beta$ as its target type.

From the programming perspective, the $\lp$ language starts out with a
set of ``built-in'' sorts and type constructors. This set contains $o$,
the type of propositions, and other primary types like {\it int}, {\it real},
{\it string} with obvious meanings.
It also includes a unary type constructor {\it list} which is used to
form types of homogeneous lists. These sets of sorts and type
constructors can be added to by the programmer by using declarations
that have the following form:
\begin{tabbing}
\quad\quad {\it kind} \quad {\it c} \quad\quad {\it type} $\ra$ ... $\ra$
{\it type}.
\end{tabbing}
Such a declaration associates with the symbol $c$ an arity that is one
less than the number of the occurrences of the keyword {\it type} in
it, and $c$ is considered a sort when its arity is zero.
As a concrete example, the following declarations define a binary type
constructor {\it pair} and a sort {\it i}.
\begin{tabbing}
\quad\quad {\it kind} \quad \={\it pair} \quad\quad\= {\it type}
$\ra$ {\it type} $\ra$ {\it type}.\\
\quad\quad {\it kind} \> {\it i} \> {\it type}.
\end{tabbing}
Based on the enhanced sets of sorts and type constructors, {\it (pair int i)},
{\it (list int $\ra$ $o$)} and {\it (pair i (i $\ra$ i))} are all legal
types.

Assuming sets of typed constants
and variables, the terms of the simply typed $\lambda$-calculus are
identified together with their types through the following rules:
\begin{enumerate}
\item a constant or a variable of type $\tau$ is a term of type $\tau$;
\item  the expression $(\lambdax{x} t)$ is a term of type
  $\arrxy{\tau_2}{\tau}$ provided $x$ is a variable of type $\tau_1$
  and $t$ is a term of type $\tau_2$;
\item  the expression $(t_1\ t_2)$ is a term of type $\tau$ provided
  $t_1$ and $t_2$ are terms of type $\arrxy{\tau_1}{\tau}$ and
  $\tau_1$ respectively.
\end{enumerate}
Terms defined by the second and third rules are called {\it applications} and
{\it abstractions} respectively. We minimize the usage of parentheses by assuming
that applications are left associative and that abstractions have higher precedence
than applications.

Abstractions are of special interest among the categories of
terms, because it is they that endow the language the ability
to explicitly represent binding. From a scoping perspective, an
abstraction term of the form $\lambdax{x}t$ captures the concept that
$x$ is a variable that ranges over $t$.
From the perspective of meaning, such a term can be understood as a
function definition in which $x$ is the formal parameter and $t$ is the function body,
\ie, supplied with an actual parameter, say $t_2$, the evaluation result of this function
should be a variation of $t$ in whose structure the occurrences of $x$ are replaced
by $t_2$. Such an evaluation process is encompassed by an application
term $(t_1\ t_2)$ where $t_1$ denotes a function definition and $t_2$ an actual parameter.

The intended meanings of $\lambda$-terms are made formal by defining a
notion of equality between them that takes into account the binding
and functional character of abstractions discussed above. The
formation rules for these terms gives rise to a natural notion of
subcomponents or subterms. Further, let us say that an occurrence of a
variable $y$ is bound or free in a term $t$ depending on whether or
not it appears within a subterm of the form $\lambdax{y} t'$ and that a
variable is free or bound in $t$ if it has a free or bound occurrence
in it. Finally, let $t[x:=s]$ denote the result of replacing
all the free occurrences of $x$ by $s$ in $t$, where $t$ and $s$ are
terms and $x$ is a variable of the same type as that of $s$. In this
context the rules of $\lambda$-conversion that identify the desired
equality notion are defined as follows:
\begin{description}
\item[($\alpha$-conversion)]
 Replacing a subterm of form $\lambdax{x}t$ of
a given term with $\lambdax{y}(t[x:=y])$, provided $y$ is a variable with
the same type as that of $x$ and not occur in $t$.
\item[($\beta$-conversion)]
Replacing a subterm of form $(\lambdax{x}t)\app s$ of a given term
with $t[x:=s]$ or vice versa,
provided for every free variable $y$ of $s$, $y$ does not have a bound
occurrence in $t$. The subterm $(\lambdax{x}t)\app s$ is known as a $\beta$-redex,
\item[($\eta$-conversion)] Replacing a subterm of form $\lambdax{x}(t\app x)$
of a given term with $t$ or vice versa, provided $t$ is of type
$\arrxy{\alpha}{\beta}$, $x$ is a variable of type $\alpha$ and
not appear free in $t$.
\end{description}
The rule of $\alpha$-conversion recognizes the irrelevance of the names
of bound variables in an abstraction. For example, the terms $(\lambdax{x}x)$
and $(\lambdax{y}y)$ encode the same identification function despite
the different names given to its formal parameter. The
$\beta$-conversion rule formalizes the notion of function evaluation
discussed earlier. This rule initially seems limited because its
application requires $s$ not to have free variables that are bound in
$t$. However, if this condition is not satisfied at the beginning then
a sequence of $\alpha$-conversions can be used to rename the bound
variables in $t$ to avoid the name collisions.
The $\eta$-conversion rule encompasses the common assumption in
mathematics that the functions $f$ and $g$ are equal if
for every term $t$ of a suitable type, the function applications $(f\app t)$
and $(g\app t)$ are equal.

A pair of $\lambda$-terms are considered equal if they can be obtained from each
other by a sequence of applications of $\alpha$-, $\beta$- or $\eta$- rules.
The computation underlying $\lp$ is in fact organized around a process of comparing
$\lambda$-terms based on this notion of equality, and this process is known as
{\it unification}. The concept of unification will be discussed in details in
Chapter~\ref{chp:interpreter}. For now, we can simply understand it as a
matching process during which variables that are free at the top-level of
the terms can be replaced by some other term structures in attempting to make
the terms equal. A key requirement in such a replacement, however, is
that we cannot introduce variable occurrences that get captured by
abstractions occurring in the term into which the replacement
is done.

The last issue with regard to understanding the data structure of $\lp$ is
about the usage of constants from the programming perspective.
The set of constants of this language can be partitioned into two sub-categories
as {\it logical} and {\it non-logical} ones.
The language has internal interpretations to logical constants,
and they can be used to construct high level computation control.
This set of constants consists of
the symbols $\top$ of type $o$, denoting the tautological proposition, the symbols
$\conj$, $\disj$, $\imp$, of type $\arrxy{o}{\arrxy{o}{o}}$, corresponding to
logical conjunction, disjunction and implication respectively, and sets of
symbols $\Pi_{\alpha}$ and $\Sigma_{\alpha}$ of type
$\arrxy{(\arrxy{\alpha}{o})}{o}$ for each type $\alpha$.
The last two (families of) logical constants are used to construct
universal and existential quantifications:
formulas usually written as $\allx{x}\app t$ and
$\exists{x}\app t$ are encoded as
$\Pi_{\alpha}\lambdax{x} t$ and $\Sigma_{\alpha}\lambdax{x} t$,
where $x$ is a variable of type $\alpha$. The type subscripts associated with these
constants will be left out when they are not essential to our discussion.
Further, when the context is clear, we will still use the conventional
$\allx{x}\app t$ and $\exists{x}\app t$ representations for quantifications, and use $\conj$,
$\disj$ and $\imp$ as infix operators for better readability.

Constants other than the logical ones belong to the non-logical set.
Built-in support is provided to a primary collection of it,
and user can increment this set by defining their own in the course of programming.
The initial set of non-logical constants consists of the sets of integers, real numbers,
strings (character sequence enclosed by double quotes),
{\it nil$_\alpha$} of type {\it (list $\alpha$)}
and the right-associative binary infix operator {\it ::$_\alpha$} of type
{\it ($\alpha\ra$ list $\alpha\ra$ list $\alpha$)}.
The last two (families of) constants are used for encoding homogeneous lists
of element type $\alpha$, $\eg$ an integer list can be denoted as
{(1 ::$_{int}$ 2 ::$_{int}$ nil$_{int}$)}. Again, the type annotations of {\it list}
and {\it nil} will be omitted when the context is clear.

Users can define new non-logical constants together with their types
through declarations of the following kind
\begin{tabbing}
\quad\quad {\it type} \quad {\it const} \quad\quad {\it $<$type$>$}.
\end{tabbing}
where {\it $<$type$>$} should be replaced by the actual type of the
constant. Such declarations will typically be used when a new set of
constants is needed for encoding objects that need to be computed
over. As a concrete example, suppose that our computational task
requires us to represent the collection of of closed untyped
$\lambda$-terms built from the sole constant symbol {\it a}. The
following declarations then identify the required symbols
within $\lp$ to realize an encoding of such terms:
\begin{tabbing}
\quad\quad\= {\it kind} \quad\= {\it tm} \quad\quad\= {\it type}. \\
\> {\it type} \> {\it a}   \> {\it tm}.\\
\> {\it type} \> {\it app} \> $\arrxy{tm}{\arrxy{tm}{tm}}$.\\
\> {\it type} \> {\it abs} \> $\arrxy{(\arrxy{tm}{tm})}{tm}$.
\end{tabbing}
A sort {\it tm} is first declared as the type of the set of object-level terms, \ie, the set
of terms to be represented.
The second line above declares a constant {\it a} as the only object-level constant term.
Constants {\it app} and {\it abs} are the selected constructors for denoting object-level
applications and abstractions respectively: an object-level application can be formed by
applying {\it app} to two arguments of type {\it tm}, whereas an object-level abstraction
is denoted by applying {\it abs} to a meta-level abstraction of type $\arrxy{tm}{tm}$.
Within such a setup, an object-level term $(\lambdax{x}(a\ x))\ (\lambdax{y}y)$ can be
represented
as $app\ (abs\ (\lambdax{x}(app\ a\ x)))\ (abs\ (\lambdax{y}y))$.
Based on the above representations, now we can think of realizing operations over
the object-level terms. For example, suppose a copy operation, whose functionality is to duplicate a
given object-level term, is of interest. We can declare a predicate constant, \ie,
constant with proposition target type, named {\it copy} for this purpose.
\begin{tabbing}
\quad\quad{\it type} \quad {\it copy} \quad\quad $\arrxy{tm}{\arrxy{tm}{o}}$.
\end{tabbing}
We expect that this predicate evaluates to {\it true} if and only if its first and
second arguments are identical to each other. Such functionality can be specified through
definitions of predicates constructed by {\it formulas} in our language, which are
discussed in the next section.

\section{Higher-Order Hereditary Harrop Formulas}\label{sec:hohh}

The language of {\it higher-order hereditary Harrop} or {\it hohh
  formulas} is determined by two special classes of expressions: the
$G$-formulas that function as {\it goals} or {\it queries} in a logic
programming setting and the $D$-formulas that function as {\it program
  clauses} or {\it definition clauses} in this context. These formulas
are essentially subsets of $\lambda$-terms of type $o$ that are
constructed from recursive applications of logical constants with
certain restrictions.

Using symbol $P$ to denote a non-logical constant or a variable, we define
an {\it atomic formula} as a term of type $o$ with the structure
$(P\app t_1\ ...\ t_n)$, where, for $1 \leq i \leq n$, the only
logical constants appearing in each $t_i$ are $\conj$, $\disj$,
$\Sigma$, or $\Pi$; a term satisfying such a restriction is referred
to as a {\em positive} term.
If the head $P$ of an atomic formula is a variable, the
formula is said to be {\it flexible} and otherwise it is said to be
{\it rigid}. Using the symbol $A$ to denote
atomic formulas and $A_r$ to denote rigid atomic formulas, the sets of goals $G$ and
program clauses $D$ is identified by the following syntactical rules:
\begin{tabbing}
\quad\quad\= $G ::=\top$  $\vert$  $ A$  $\vert$  $\andxy{G}{G}$  $\vert$
$\orxy{G}{G}$  $\vert$  $\exists{x}G$ $\vert$ $\forall{x}G$ $\vert$ $\impxy{D}{G}$.\\
\>           $D ::= A_r$  $\vert$  $\impxy{G}{A_r}$  $\vert$  $\andxy{D}{D}$  $\vert$  $\forall{x} D$.
\end{tabbing}
In a program clause of form $A_r$ or $\impxy{G}{A_r}$, $A_r$ is called the {\it head} of the clause,
and for a clause of the latter form, $G$ is said to be its {\it body}.
The goals in the forms of $\forall{x}G$ and $\impxy{D}{G}$ are called
{\it generic} and {\it augment} goals respectively.

A program in $\lp$
is a set of closed clauses, \ie, a set of clauses that do not contain
any free variables.
Computation in $\lp$ corresponds to solving a top-level closed query
against a given program and relative to a given signature that
identifies the set of available constants. The program at the
beginning consists of all the clauses that the user of $\lp$ has
provided at the top-level and the signature consists of all the
built-in constants as well as those identified through user
declarations. The manner in which the computation proceeds is dictated
by the top-level structure of the query as indicated by the rules
below.
\begin{enumerate}
\item The goal $\top$ leads to an immediate solution regardless of the
  program and the signature.
\item The goal $\andxy{G_{1}}{G_{2}}$ is
  solved against any program and signature by solving both $G_{1}$ and
  $G_{2}$ using the same program and signature.
\item The goal $\orxy{G_{1}}{G_{2}}$ is solved against any program by
  solving one of $G_{1}$ or $G_{2}$ using the same program and
  signature.
\item The goal $\exists{x}G$ is solved against a program and a
  signature by picking a closed term $t$ of the same type as $x$
  that is constructed using only the constants in the given
  signature and then solving  $G[x := t]$ from the same
  program and signature; notice that the correctness of the
  replacement of $x$ by $t$ here is guaranteed by the fact that $t$ is
  closed.
\item The goal $\forall{x}G$ is solved against a given program $\cal
  P$ and signature $\Sigma$ by selecting a constant $c$ of the same
  type as $x$ that does not
  belongs to $\Sigma$ and then solving $G[x := c]$ against
  the program $\cal P$ and the signature $\Sigma \cup \{c\}$.
\item The goal $\impxy{D}{G}$ is solved against a program $\cal P$ and
  a signature $\Sigma$ by solving $G$ against the program ${\cal P}
    \cup \{D\}$ and the signature $\Sigma$.
\item The rigid atomic goal $A_r$ is solved from a program $\cal P$
  and a signature $\Sigma$ by picking a clause from $\cal P$,
  instantiating all the top-level universally quantified variables in
  it with closed terms constructed using only constants in $\Sigma$ to
  get a formula that $\lambda$-converts to the form $A_r$ or
  $\impxy{G}{A_r}$ and, in the latter case, solving the goal $G$
  from the program $\cal P$ and signature $\Sigma$.
\end{enumerate}
An important point to note with regard to the rules presented above is
that they can lead, in particular instances, to changes in the program
and the signature against which a query is to be solved. In
particular, a generic goal can extend the signature and an augment goal
can lead to additions to the program. These kinds of goals thus have
the ability to give names and clauses a scope over particular
computations. This situation is to be contrasted with the usual Horn
clauses that underlie {\em Prolog}; generic and augment goals are not
permitted in that setting and consequently the scoping ability in
question is absent there.

The above description of the operational semantics for $\lp$ is not
yet suitable to be used as a basis for implementation. First, we are
assuming an oracle in picking a proper instance for existentially
quantified variables in queries and universally quantified variables
in clauses.  Second, we have not specified how to select clauses for
solving rigid atomic goals when multiple possibility exists and nor
have said how to select the disjunct to solve when processing
disjunctive goals. Finally, a practical means is needed for
controlling the visibility of constants and clauses
introduced in generic and augment goals. We defer a discussion of
these issues till Chapter~\ref{chp:interpreter}, hoping that
enough details have been provided here to make clear when a particular
computation has been correctly carried out.

The new scoping mechanisms present in {\it hohh} formulas provides the
ability to realize recursion over abstractions in $\lambda$-terms and,
thus, over binding structures present in object languages over which
we are interested in describing computations.
To illustrate this capability, we consider the {\it copy} predicate
introduced in the previous section and show how it can be defined
in the $\lp$ language.

Assuming the representation for $\lambda$-terms that we have already
presented, it is very natural to define the copying computation for constant
and applications with the following two clauses:
\begin{tabbing}
\quad\quad\=({\it copy} {\it a} {\it a})\\
\>          ($\allx{t_1}\allx{t_2}\allx{t_3}\allx{t_4}\ (\impxy{(\andxy{copy\ t_1\ t_3}{\it copy\ t_2\ t_4})}{copy\ (app\ t_1\ t_2)\ (app\ t_3\ t_4)})$)
\end{tabbing}
These clauses simply state that a copy of the constant {\it a} is the
constant itself, and copying an application can be
carried out by constructing a new application over the copies of its arguments.
Now we need to consider how to recursively copy an abstraction of form {\it (abs ($\lambdax{x}$t))}.
Intuitively, we would like to have a clause of form
\begin{tabbing}
\dquad$\allx{t_1}\allx{t_2}\ (\impxy{copy\ t_1\ t_2}{copy\ (abs\ t_1)\ (abs\ t_2)})$
\end{tabbing}
to descend into the argument of {\it abs}. However, this clause is illegal because the arguments of {\it copy} should have type {\it tm}
whereas the argument of {\it abs} has type $\arrxy{tm}{tm}$ which essentially corresponds to an
abstraction $\lambdax{x}t$. A more careful consideration reveals that the copy of
$\lambdax{x}t$ in fact can be realized by first constructing a copy
for $t[x:=c]$ where $c$ is a new constant, and then
constructing an abstraction over the structure that results from
extracting $c$ out of this
copy. These operations can be easily expressed by using a generic
goal. In particular, consider the clause
\begin{tabbing}
\quad\quad$\allx{t_1}\allx{t_2}\ (\impxy{(\allx{c}\ copy\ (t_1\ c)\ (t_2\ c))}{copy\ (abs\ t_1)\ (abs\ t_2)})$.
\end{tabbing}
The generic goal that appears in this clause will lead to the
introduction of a new constant $c$. By applying $t_1$, which
essentially corresponds to an abstraction $\lambdax{x}t$ to $c$, the
substitution $t[x:=c]$ is automatically taken care of. Once this
structure has been copied, the control of the scope of $c$ embodied in
the generic goal ensures that the only correct instantiation of $t_2$
would be one that extracts $c$ out of $t[x:=c]$ and
constructs an abstraction over it. Thus the recursion over
abstractions in defining
{\it copy} is accomplished by the use of a generic goal.
However, our program is still not
entirely correct because there is no clause so far specifying how to copy the constant
$c$ introduced by the generic goal. The computation itself is very simple and
can be specified by a clause of form {\it copy c c}, but this clause cannot be simply added into
our program at the top-level because the constant $c$ is only visible inside the generic goal we discussed
above. The solution is to enhance the clause for copying abstractions by the use of an augment goal
\begin{tabbing}
\quad\quad$\allx{t_1}\allx{t_2}\ (\impxy{(\allx{c}\ (\impxy{copy\ c\ c}{copy\ (t_1\ c)\ (t_2\ c)}))}{copy\ (abs\ t_1)\ (abs\ t_2)})$.
\end{tabbing}
Now the clause {\it copy c c} has its scope inside that of $c$, so that it is effective only when
computation descends into the body of an abstraction.

\begin{figure}
\begin{tabbing}
\quad\quad\={\it copy a a}. \\
\> {\it copy (app T1 T2) (app T3 T4)} \= $\pif$ {\it copy T1 T3, copy T2 T4}. \\
\>{\it copy (abs T1) (abs T2)}\> $\pif$ $Pi\ c\plam \ (copy\ c\ c\ \pimp \ copy\ (T1\ c)\ (T2\ c))$.
\end{tabbing}
\caption{A program defining the predicate {\it copy}.}\label{fig:copy}
\end{figure}

We shall find it convenient to use in the rest of this thesis a {\em
  Prolog}-like syntax in presenting program clauses that are meant to
constitute $\lp$ programs. In particular, we always omit top-level
conjunctions in a program and use a period to terminate
top-level clauses.
Second, we use capitalized names for universally quantified variables over top-level clauses and for existentially
quantified variables over top-level goals and leave the quantifiers implicit. Third, we use the syntax {\it $A_r$ $\pif$ G.}
to denote top-level clauses of form $\impxy{G}{A_r}$. Finally, comma
and semicolon will be used to denote $\conj$ and $\disj$ respectively.
Based on these conventions and using the concrete syntax $\pimp$ for $\imp$, {\it Pi} for $\forall$, and the
infix operator $\plam$ for $\lambda$, the {\it copy} program that we
have just described would be presented concretely as in Figure~\ref{fig:copy}.

\section{An Extended Example }\label{sec:language_example}
We now provide a closer look at the power of $\lp$ and a better
feeling for programming in it by considering a extended example of its
use in a meta-programming task. The particular task we consider is
that of encoding formulas from a first-order logic and realizing a
syntactic transformation on them to produce their prenex normal forms,
\ie, a form in which all the quantifiers appear at the head of the
formula.

\begin{figure}[t]
\begin{tabbing}
\quad\quad\={\it kind}\quad\={\it form}\quad\quad\={\it type}. \\[5pt]
\>{\it type}\>{\it truth}\>{\it form}.\\
\>{\it type}\>{\it false}\>{\it form}.\\
\>{\it type}\>{\it and}\>$\arrxy{form}{\arrxy{form}{form}}$. \\
\>{\it type}\>{\it or}\>$\arrxy{form}{\arrxy{form}{form}}$.  \\
\>{\it type}\>{\it imp}\>$\arrxy{form}{\arrxy{form}{form}}$. \\
\>{\it type}\>{\it all}\>$\arrxy{(\arrxy{term}{form})}{form}$. \\
\>{\it type}\>{\it some}\>$\arrxy{(\arrxy{term}{form})}{form}$.
\end{tabbing}
\caption{Encoding the logical symbols in an object
  logic.}\label{fig:prenex_formula}
\end{figure}

The formulas that we want to encode will be from a logic that, as
usual, is characterized by logical and non-logical symbols.
The logical symbols that we assume here are $\top$, $\bot$,
$\conj$, $\disj$, $\imp$, $\forall$ and $\exists$. We shall encode
these by using the constants {\it truth}, {\it false}, $and$, $or$,
$imp$, $all$ and $some$, respectively. In encoding the quantifiers,
we, once again, separate a treatment of their meanings from a treatment
of their binding
effects. Figure~\ref{fig:prenex_formula} contains a set of declarations
that identify these constants; the type {\it form} is used in the
encoding to correspond to the category of formulas. For the
non-logical vocabulary, we shall assume that the object logic has
three constants {\it a}, {\it   b} and {\it c} and a single function
symbol {\it f} with arity 1. Beyond this, we assume two binary
predicate symbols {\it adj} and {\it path}; intuitively, these
symbols serve to describe graphs, with the first being used to
describe an adjacency relation and the second the relation
corresponding to the existence of a path between two nodes. Using the
type {\it term} to represent object logic terms, the declarations in
Figure~\ref{fig:prenex_term} provide an encoding of this non-logical
vocabulary.

\begin{figure}[t]
\begin{tabbing}
\quad\quad\={\it kind}\quad\={\it term}\quad\quad\={\it type}. \\[5pt]
\>{\it type}\>{\it a}\>{\it term}.\\
\>{\it type}\>{\it b}\>{\it term}.\\
\>{\it type}\>{\it c}\>{\it term}.\\
\>{\it type}\>{\it f}\>$\arrxy{term}{term}$.\\[5pt]
\>{\it type}\>{\it adj}\>$\arrxy{term}{\arrxy{term}{form}}$. \\
\>{\it type}\>{\it path}\>$\arrxy{term}{\arrxy{term}{form}}$.
\end{tabbing}
\caption{Encoding the non-logical symbols in an object logic.}\label{fig:prenex_term}
\end{figure}

We illustrate our encoding of formulas by considering the
representation of the following object-level formula that describes a
graph with four nodes and that describes the {\it path} relation in
terms of the {\it adj} relation:
\begin{tabbing}
\quad\quad\= $adj(a,\ b)\ \conj$ \\
\> $adj(b,\ c)\ \conj$ \\
\> $adj(c,\ f(c))\ \conj$ \\
\> $(\allx{x}\allx{y}\ (\impxy{adj(x,\ y)}{path(x,\ y)}))\ \conj$\\
\> $(\allx{x}\allx{y}\allx{z}\ (\impxy{(\andxy{adj(x,\ y)}{path(y,\ z)})}{path(x,\ z)}) $
\end{tabbing}
The $\lp$ term that represents this formula is the following:
\begin{tabbing}
\quad\quad\= $(and\ $\=$(adj\ a\ b)$ \\
\>\>                   $(and\ $\=$(adj\ b\ c)$ \\
\>\>\>                           $(and\ $\=$(adj\ c\ (f\ c))$ \\
\>\>\>\>                                   $(and\ $\=$(all\ x\plam\ (all\ y\plam\ (imp\ (adj x\ y)\ (path\ x\ y))))$ \\
\>\>\>\>\>                                           $(all\ x\plam\ (all\ y\plam\ (all\ z\plam\ (imp\ $\=$(and\ (adj\ x\ y)\ (path\ y\ z))$ \\
\>\>\>\>\>\>                                                                                             $(path\ x\ z))))))))).$
\end{tabbing}

\begin{figure}
\begin{tabbing}
\quad\quad{\it type}\quad$\itm$\quad\quad$\arrxy{term}{o}$. \\
\quad\quad\=$\itm\ (f\ X)\ $\= \kill
\>$\itm\ a$. \\
\>$\itm\ b$. \\
\>$\itm\ c$. \\
\>$\itm\ (f\ X)$ \>$\pif\ \itm\ X$.\\[5pt]
\quad\quad{\it type}\quad$\ia$\quad\quad$\arrxy{form}{o}$. \\
\quad\quad\=$\ia\ (path\ A\ B)\ $\= \kill
\>$\ia\ (adj\ X\ Y)$ \>$\pif\ \itm\ X,\ \itm\ Y$. \\
\>$\ia\ (path\ X\ Y)$ \>$\pif\ \itm\ X,\ \itm\ Y$. \\[5pt]
\quad\quad{\it type}\quad$\qf$\quad\quad$\arrxy{form}{o}$. \\
\quad\quad\=$\qf\ (and\ A\ B)\ $\= \kill
\>$\qf\ truth$. \\
\>$\qf\ false$. \\
\>$\qf\ A$\>$\pif\ \ia\ A$.\\
\>$\qf\ (and\ A\ B)$\>$\pif\ \qf\ A,\ \qf\ B$. \\
\>$\qf\ (or\ A\ B)$\>$\pif\ \qf\ A,\ \qf\ B$. \\
\>$\qf\ (imp\ A\ B)$\>$\pif\ \qf\ A,\ \qf\ B$.
\end{tabbing}
\caption{Some recognizers for encodings of object logic categories.}\label{fig:prenex_quantifier_free}
\end{figure}

With this representation in place, we consider the specifications of
the simple properties of being (the encodings of) a term, an atomic
predicate and a quantifier
free formula. Predicates recognizing these attributes of $\lp$ terms
are presented in Figure~\ref{fig:prenex_quantifier_free}; the
names $\itm$, $\ia$ and $\qf$ are used for recognizers for each of
these categories, respectively.

\begin{figure}\small
\begin{tabbing}
{\it type}\quad{\it prenex}\quad\quad$\arrxy{form}{\arrxy{form}{o}}$.\space\\ \space
\={\it prenex (some X) (some Y) }\= \kill
\>{\it prenex truth truth}. \qquad\qquad{\it prenex false false}. \\
\>{\it prenex B B}\>$\pif\ia\ B.$ \\
\>{\it prenex $($and B C$)$ D}\>$\pif${\it prenex B U, prenex C V, mrg $($and U V$)$ D.} \\
\>{\it prenex $($or B C$)$ D}\>$\pif${\it prenex B U, prenex C V, mrg $($or U V$)$ D.} \\
\>{\it prenex $($imp B C$)$ D}\>$\pif${\it prenex B U, prenex C V, mrg $($imp U V$)$ D.} \\
\>{\it prenex $($all B$)$ $($all D$)$}\>$\pif${\it Pi x$\plam($term x $\pimp$ prenex $($B x$)$ $($D x$))$}.\\
\>{\it prenex $($some B$)$ $($some D$)$}\>$\pif${\it Pi x$\plam($term x $\pimp$ prenex $($B x$)$ $($D x$))$}.\\[2.9pt]
{\it type}\quad{\it mrg}\quad\quad$\arrxy{form}{\arrxy{form}{o}}$.\\
\={\it mrg (and (all X) (all Y)) (all Z) }\= \kill
\>{\it mrg $($and $($all B$)$ $($all C$))$ $($all D$)$ }\>$\pif$
{\it Pi x$\plam($term x $\pimp$ mrg $($and $($B x$)$ $($C x$))$ $($D x$))$}. \\
\>{\it mrg $($and $($all B$)$ C$)$ $($all D$)$ }\>$\pif$
{\it Pi x$\plam($term x $\pimp$ mrg $($and $($B x$)$ C$)$ $($D x$))$}. \\
\>{\it mrg $($and $($some B$)$ C$)$ $($some D$)$ }\>$\pif$
{\it Pi x$\plam($term x $\pimp$ mrg $($and $($B x$)$ C$)$ $($D x$))$}. \\
\>{\it mrg $($and B $($all C$))$ $($all D$)$ }\>$\pif$
{\it Pi x$\plam($term x $\pimp$ mrg $($and B $($C x$))$ $($D x$))$}. \\
\>{\it mrg $($and B $($some C$))$ $($some D$)$ }\>$\pif$
{\it Pi x$\plam($term x $\pimp$ mrg $($and B $($C x$))$ $($D x$))$}. \\

\>{\it mrg $($or $($some B$)$ $($some C$))$ $($some D$)$ }$\pif$\\
\squad\squad{\it Pi x$\plam($term x $\pimp$ mrg $($or $($B x$)$ $($C x$))$ $($D x$))$}. \\

\>{\it mrg $($or $($all B$)$ C$)$ $($all D$)$ }\>$\pif$
{\it Pi x$\plam($term x $\pimp$ mrg $($or $($B x$)$ C$)$ $($D x$))$}. \\
\>{\it mrg $($or $($some B$)$ C$)$ $($some D$)$ }\>$\pif$
{\it Pi x$\plam($term x $\pimp$ mrg $($or $($B x$)$ C$)$ $($D x$))$}. \\
\>{\it mrg $($or B $($all C$))$ $($all D$)$ }\>$\pif$
{\it Pi x$\plam($term x $\pimp$ mrg $($or B $($C x$))$ $($D x$))$}. \\
\>{\it mrg $($or B $($some C$))$ $($some D$)$ }\>$\pif$
{\it Pi x$\plam($term x $\pimp$ mrg $($or B $($C x$))$ $($D x$))$}. \\

\>{\it mrg $($imp $($all B$)$ $($some C$))$ $($some D$)$ }$\pif$ \\
\squad\squad{\it Pi x$\plam($term x $\pimp$ mrg $($and $($B x$)$ $($C x$))$ $($D x$))$}. \\
\>{\it mrg $($imp $($all B$)$ C$)$ $($some D$)$ }\>$\pif$
{\it Pi x$\plam($term x $\pimp$ mrg $($imp $($B x$)$ C$)$ $($D x$))$}. \\
\>{\it mrg $($and $($some B$)$ C$)$ $($all D$)$ }\>$\pif$
{\it Pi x$\plam($term x $\pimp$ mrg $($imp $($B x$)$ C$)$ $($D x$))$}. \\
\>{\it mrg $($and B $($all C$))$ $($all D$)$ }\>$\pif$
{\it Pi x$\plam($term x $\pimp$ mrg $($imp B $($C x$))$ $($D x$))$}. \\
\>{\it mrg $($and B $($some C$))$ $($some D$)$ }\>$\pif$
{\it Pi x$\plam($term x $\pimp$ mrg $($imp B $($C x$))$ $($D x$))$}. \\
\>{\it mrg B B}\>$\pif \ $\qf$\ B$.
\end{tabbing}
\vspace{-0.8cm}
\caption{A specification of the prenex-normal form relation.}\label{fig:prenex}
\end{figure}

We now consider the encoding of the prenex normal form
relation. Specifically we are interested in writing down a set of
program clauses that define a predicate {\it prenex} such
that a goal of the form {\it prenex A B} is solvable from them just in
the case that {\it A} and {\it B} are both encodings of formulas and,
further, {\it B} represents a prenex normal form of the formula
represented by {\it A}. The definition of this predicate is presented
in Figure~\ref{fig:prenex}. Use is made in this definition of an
auxiliary predicate {\it mrg} for the purpose of raising quantifiers
over binary connectives. The definitions of
of both {\it prenex} and {\it mrg} use generic and augment goals in
a fashion already illustrated with the definition of the {\it copy}
predicate  to analyze and synthesize abstraction structures so as to
realize a recursion over the representation of quantified formulas.

Given program in Figure~\ref{fig:prenex}, the query
\begin{tabbing}
\quad\quad{\it ?- prenex} \={\it $($or }\={\it$($all x$\plam($and $($adj x x$)$ $($and }\={\it$($all y$\plam($path x y$))$} \\
\>\>\>{\it $($adj $($f x$)$ c$))))$}\\
\>\>{\it $($adj a b$)$} \\
\>{\it Pnf}.
\end{tabbing}
should succeed by instantiating the top-level existentially quantified
variable {\it Pnf} to the term
\begin{tabbing}
\dquad{\it $($all x$\plam($all y$\plam($or $($and $($adj x x$)$ $($and $($path x y$)$ $($adj $($f x$)$ c$)))$ $($adj a b$))))$}.
\end{tabbing}
For another example, the query
\begin{tabbing}
\dquad{\it ?- prenex $($and $($all x$\plam($adj x x$))$ $($all z$\plam($all y$\plam($adj z y$))))$ Pnf}.
\end{tabbing}
is also solvable with any one of the following five instantiations for
the variable {\it Pnf}:
\begin{tabbing}
\dquad{\it all z$\plam($all y$\plam($and $($adj z z$)$ $($adj z y$)))$},\\
\dquad{\it all z$\plam($all x$\plam($and $($adj x x$)$ $($adj z x$)))$},\\
\dquad{\it all x$\plam($all z$\plam($all y$\plam($and $($adj x x$)$ $($adj z y$))))$},\\
\dquad{\it all z$\plam($all x$\plam($all y$\plam($and $($adj x x$)$ $($adj z y$))))$}, and\\
\dquad{\it all z$\plam($all y$\plam($all x$\plam($and $($adj x x$)$ $($adj z y$))))$}.
\end{tabbing}
The multiple solutions listed above are a result of the existence of multiple matching clauses when solving
atomic goals in the course of computation.

We have only considered one example of the use of $\lp$ in encoding
computations over binding structures but, hopefully, this example will
provide the background for understanding our later discussions about
implementation. An interested reader can find several other examples
in the literature; such examples and a discussion of the
meta-programming capabilities of the language may be found, for
instance, in \cite{NM98Handbook}. In realizing such computations we
will have to find an effective way for treating varied aspects such as
search and the selection of instantiation terms, issues that we have
ignored in the presentation here as noted already. We will take these
issues up seriously in Chapter~\ref{chp:interpreter}. Anticipating
that discussion we note that the examples of {\it prenex} and {\it
  copy} both belong to the $\Ll$ fragment of $\lp$ programming, a class
for which the unification computation is decidable and admits unique
solutions and for which we are interested in developing a good
treatment in this thesis.

\section{Polymorphism and the Role of Types in Computation}\label{sec:types_in_computation}
Our presentation of $\lp$ up to now has treated it as if it is
monomorphically typed. In reality, the type system of $\lp$ allows for
polymorphism to provide flexibility and convenience in
programming. Such polymorphism is obtained by admitting the use of
type variables.  In particular, in addition to the sets of sorts and
type constructors, an infinite supply of type variables is also
assumed. Sorts and type variables are basic types, starting from which
constructed types, including function types, are built by recursive
applications of type constructors.  In the subsequent discussion, we
use capital letters to denote type variables.

Intuitively, a type variable can be viewed as an abbreviation of an
infinite set of types in the monomorphic type system. Thus a type with
type variables occurring inside in fact provides a schema for a family
of types: sets of more specific types can be generated by replacing
the contained type variables with other types. For instance, the
polymorphic type {\it list A} denotes a family of list types such as
{\it list int} for integer lists, {\it list $($list int$)$} for list
of integer lists, and {\it list $($list B$)$} for list of lists whose
element is of type {\it B}, where $B$ can again be instantiated by
arbitrary types. Consequently, a constant declared with a polymorphic
type can be viewed as an abbreviation of an infinite set of constants,
each element of which has a monomorphic type as an instance of the
type schema. For example, previously we have families of empty list
$nil_{\alpha}$ and list constructor $::_{\alpha}$ for each monomorphic
type $\alpha$. Now these sets can be abbreviated into two constants
{\it nil} of type {\it list A} and {\it ::} of type
$\arrxy{A}{\arrxy{\it list\ A}{\it list\ A}}$. By instantiating the
type variable {\it A} to {\it int}, an integer list can be denoted by
{\it 1 :: 2 :: nil}. Note that the instantiation of the type variable
has to be performed in a uniform manner across the entire polymorphic
type.  For example, a structure of form {\it 1 :: "a" :: nil} is not
well-typed since the integer argument of the first {\it ::} requires
its type variable being replaced by {\it int}, whereas the second
string list argument demands it being replaced by {\it string}
instead.

The idea of using polymorphic types to abbreviate sets of constants
can also be applied to clause definitions. An example for such a usage
is the predicate {\it append} which concatenates the lists in its
first two arguments into the third one.
\begin{tabbing}
\dquad\={\it type}\quad{\it append}\quad\quad\=$\arrxy{\it list\ A}{\arrxy{\it list\ A}{\arrxy{\it list\ A}{o}}}$. \\
\>{\it append nil L L.}\\
\>{\it append $($X :: L1 $)$ L2 $($X :: L3$)$ }$\pif\ ${\it append L1 L2 L3}.
\end{tabbing}
The two clauses defining {\it append} are shared by the append
operation of an infinite set of list types.  From the programming
perspective, this sort of polymorphism is known as {\it parametric},
where a function (predicate) works uniformly over a range of types.

In addition to parametric polymorphism, another sort of polymorphism
is embodied by $\lp$, which is obtained when a function (predicate)
works in unrelated ways on several different types and is known as
{\it ad hoc} polymorphism. An example is provided by the following
definitions of predicate {\it print}, in which we assume predicates
{\it write$\_$int}, {\it write$\_$string} and {\it write$\_$list} for
printing out given arguments of type {\it int}, {\it string} and {\it
  list A} respectively to the standard IO.
\begin{tabbing}
\quad\quad\= $kind$ \quad \=$string$\quad\quad\quad\quad\=$type$. \kill
\>           $type$ \>  $print$\> $A$ $\ra$ $o$. \\
\>           $type$ \>  $write\_int$ \> $int$ $\ra$ $o$.\\
\>           $type$ \>  $write\_string$ \> $string$ $\ra$ $o$.\\
\>           $type$ \>  $write\_list$\> $(list\app A)$ $\ra$ $o$.\\
\\
\> $print\app N\app :-\app write\_int\app N$.\\
\> $print\app L\app :-\app write\_list\app L$. \\
\> $print\app S\app :-\app write\_string\app S$.
\end{tabbing}
In the execution of the above program, the dispatching to different
{\it write} methods depends on the type of the first argument of {\it
  print}, which is inherited from the top-level {\it print} query. To
achieve this effect, types should participate in the computation of
solving goals. In particular, they are necessary for deciding the
equality of constants, $\eg$, the dispatching in the {\it print}
example is based on the fact that a constant {\it print} of integer
type is different from those of string type or list types.

Based on the above discussions, it is clear that the roles that types
play in our language are two-fold.  First, they are used to ensure the
correctness of programs, and second they participate in computation
for deciding the solvability of queries. When playing the first role,
types are used to identify legitimate terms by restricting the
applicability of specific operations, thereby providing a control over
the computations that can be attempted. From this perspective, the
usage of our type system is very similar to that of the functional
language SML~\cite{DM82POPL, Harper86SML}. Naturally, it can be
expected that this usage should be discharged at compilation time. In
actual compilation based implementations of $\lp$, a type checking
procedure, which encompasses a process inferring types for every term
in program from those declared with constants, is commonly used by the
compilers for this purpose.  The second role of types is more peculiar
to logic programming languages, where types are actually employed
during runtime computation and have an influence on the
solutions~\cite{NP92type}. The specific way in which the computations
are determined depends on the particular algorithm used to realize the
underlying unification operations, and will be clarified in the
discussions of Chapter~\ref{chp:interpreter}.
It should be noted here that this usage requires
types to be manipulated during the execution of programs, which
consequently poses a challenge on efficient implementations of $\lp$
with regard to minimizing runtime overhead of this sort.
An optimized runtime type processing scheme is provided by this
thesis in the particular context where computation is organized around
the higher-order pattern unification, and the discussions about it
appear in Chapter~\ref{chp:types}.

\chapter{Comparison of $\lambda$-Terms}\label{chp:termRep}



The computational model described in the previous chapter requires the
comparison of atomic goals: in solving a goal of the form $A_r$, we
have to find an instance of a clause that is equal to $A_r$ or to
$G\supset A_r$. Observe, however, that the notion of equality that is
involved here is richer than that used in first-order logic
programming. In particular, we are allowed to use the conversion rules
of the $\lambda$-calculus in determining if the instance of a clause
has the required form. A question that must be addressed in an implementation
of our language, therefore, is how to effectively carry out such a
determination. As we discuss in this chapter, comparisons of this kind
between terms can be realized by first reducing them into a normal
form. The process of transforming a $\lambda$-term into a normal form
is not trivial and must be given careful attention from an efficiency
perspective. An aspect that must be given special consideration in
this context is the treatment of substitutions that are generated in the
course of reductions. We discuss the various issues involved in the
overall comparison process in this
chapter, leading eventually to what is known as an explicit
substitution notation for $\lambda$-terms. This notation eventually
serves as a high-level representation for such terms that we later
refine into an actual machine-level implementation.

This chapter is structured as follows. In the first section we provide
an overview of the comparison of $\lambda$-terms, introducing in the
process the idea of using normal forms as the basis for such
comparisons. Section~\ref{sec:beta_reduction} then discusses at a high
level the issue of carrying out $\beta$-reductions on terms in the
course of producing normal forms. This discussion highlights the
importance of treating substitutions carefully in the course of
reduction. The next section presents an explicit substitution
notation for $\lambda$-terms that is known as the {\it suspension
  calculus} \cite{gacek07suspension}; this notation provides the basis
for realizing
the normalization of terms in a finely controlled way and is what
underlies the term representation we use in the implementation scheme
developed in this thesis. Section~\ref{sec:hnorm} contains some formal
properties of the suspension calculus and it also lifts the idea of
normal forms and of rewriting sequences to produce normal forms to the
suspension calculus. This discussion underlies the reduction procedure
that is eventually used in the implementation to realize the comparison
operation.
We conclude this
chapter with a discussion of how the $\eta$-conversion rule can be taken
into account in the context of the suspension calculus.

\section{Normal Forms and Term Comparison}\label{sec:overview}

Normal forms usually play an important role in the comparison of terms
in a situation where equality encompasses a
richer notion than a simple check for syntactic identity. In the context of the
$\lambda$-calculus, a useful such form is what is known as a {\it head
  normal form}. A term is said to be in such a form if, for some $n, m
\geq 0$, it has the structure
\begin{tabbing}
\dquad$(\lambdax{x_1}\ \ldots( \lambdax{x_n}\ (\ldots (h \ t_1)\ \ldots t_m))\ldots)$,
\end{tabbing}
where $h$ is a constant or a variable,
possibly in the set $\{x_1,\ldots,x_n\}$. We call the abstractions
at the front of such a term its {\it binder}, the atom $h$ its {\it
  head} of the term and $t_i, \ldots, t_m$ its {\it arguments}. Notice
that, in particular instances, the binder might be empty and the term
may also not have any arguments.
A special case of a head normal form is one where each of its arguments
recursively have this structure. A term that satisfies this structure is
said to be in {\it $\beta$-normal form}.

An alternative characterization of a $\beta$-normal form---that is
easily seen to be equivalent to the one provided above---is that it is
a term that does not have any $\beta$-redexes as subterms. We can
think of trying to convert an arbitrary $\lambda$-term to a
$\beta$-normal form by orienting the $\beta$-conversion rule. In
particular, given a term that has a subterm of the form
$(\lambdax{x}t)\app s$, we can first use $\alpha$-conversions to
rename the bound variables in $t$ so that they are distinct from the
free variables of $s$. If we obtain the term $(\lambdax{x}t')\app s$
from this process, we can then replace this subterm by the form
$t'[x:=s]$. We shall refer to such a sequence of applications of the
$\alpha$-conversion rule followed by the oriented application of the
$\beta$-conversion rule as a $\beta$-contraction and we call a
sequence of $\beta$-contraction rule applications a
$\beta$-reduction. An important property of the
simply typed $\lambda$-calculus, that carries over also to the
polymorphic version of it that is used in $\lambda$Prolog, is that any
$\beta$-reduction sequence that starts from a given term must
terminate \cite{Girard89}. It follows from this that every term in our
language can be converted to a $\beta$-normal form and hence also a
headnormal form. We shall refer to such a form as a $\beta$-normal
(head normal) form {\it for} the term.

Two terms that have identical $\beta$-normal forms are obviously equal
under the $\lambda$-conversion rules. Ignoring for the moment the
$\eta$-conversion rule, a converse of this observation is also
available by virtue of the Church-Rosser Theorem for the
$\lambda$-calculus~\cite{Bar81}: two terms that are equal must have
$\beta$-normal forms that differ only in the names used for bound
variables. We can use this observation to describe an algorithm for
comparing two terms that have the same types; it is only such terms
that we ever need to compare in the execution model for
$\lambda$Prolog. First, we take the two terms and convert them into
head normal forms. At this stage, we compare their binder lengths. If
these are not equal, then the terms are not equal.  Otherwise, using a
sequence of $\alpha$-conversions, we can ensure that the names of the
variables in the two binders are identical; later we shall consider a
nameless representation of bound variables in the style of de Bruijn
\cite{debruijn72} that shall make this renaming step redundant. Now,
if the heads of the two terms are distinct then the terms are once
again unequal. If, on the other hand, the heads are identical, then
the typing assumption ensures that they must have an equal number of
arguments. The comparison of the two terms now reduces to a pairwise
comparison of their arguments.

The comparison algorithm that we have just described is, of course,
inadequate in the situation when the $\eta$-conversion rule is also
included. However, a simple change to it suffices in this richer
context. After we have converted the two terms to head normal forms,
it may be the case that one of them has a shorter binder than the
other. In this case, our first task is to extend the length of the
shorter binder. Suppose that this term is of the form
$\lambdax{x_1}\ldots\lambdax{x_n}t$. Clearly $t$ must have a function
type, \ie, its type must be of the form $\alpha \ra \beta$. But then
the term under consideration is equal by virtue of the
$\eta$-conversion rule to the term
$\lambdax{x_1}\ldots\lambdax{x_n}\lambdax{x_{n+1}}(t\app x_{n+1})$
where $x_{n+1}$ is some variable that does not appear free in $t$. By
a repeated use of transformation, we can make the binders of the two
terms of equal length. The comparison algorithm now proceeds as
before. The correctness of this algorithm follows from a version of
the Church-Rosser Theorem that applies to the situation where the
$\eta$-rule is included.

\section{Issues in the Realization of $\beta$-reduction}\label{sec:beta_reduction}
From the discussion in the previous section, it is clear that the
reduction of a $\lambda$-term
to a head normal form is an important component of the term comparison
operation. However, the realization of this transformation is not
trivial.
Theoretical presentations of the $\lambda$-calculus typically treat the
substitution required in rewriting a $\beta$-redex as an atomic
operation.
In particular, given a term of the form $(\lambdax{x}t)\app s$, the
sequence of $\alpha$-conversions that produces the term
$(\lambdax{x}t')\app s$ that are intended to avoid the capture of free
variables in $s$ and the subsequent rewriting to the form  $t'[x :=
  s]$ is assumed to be achieved magically in a single step.
However, from an implementation perspective,
this is a task too complicated to be accomplished in one step.
The actual realization of this operation usually combines the renaming
of the bound variables in $t$ and the replacement of the free
occurrences of $x$ by $s$ into one combined operation. It then
breaks this operation into smaller
steps: an environment is maintained to explicitly record the needed variable
replacements and each rewriting step focuses on propagating
the environment over a specific sort of term structure.
Specifically, at the beginning of the performance of $t[x:=s]$,
$[x:=s]$ is first registered into an environment $e$.
Then the rewriting task becomes that of propagating $e$ over $t$.
The interesting case arises when $t$ is of form $\lambdax{y}t'$.
Now if $y$ does not occur free in $s$, the same environment $e$ can
be simply pushed inside the abstraction. Otherwise, the occurrences of
the variable $y$ in $\lambdax{y}t'$ should be renamed to $z$, such
that $z$ does not appear free in $s$. The renaming action $[y:=z]$
is then accumulated into the environment. As a result, we have
an environment propagation step in this situation that is given by a
rewrite rule of the form
\begin{tabbing}
\dquad $(\lambdax{y}t')[x:=s]  \longrightarrow\ \lambdax{z}(t'[y:=z, x:=s])$,
\end{tabbing}
assuming $y$ appears free in $s$.
When variable occurrences are finally encountered in the substitution
performance process, replacement can be actually carried out according
to information recorded in the environment.

In the discussion above, we have thought of using an environment to
encode multiple simultaneous substitutions. Although the environment
in the example we have considered has exactly one substitution
generated from rewriting a $\beta$-redex, it is possible to imagine
environments that have more than one such substitution.
By allowing for such environments, we obtain an ability to combine
the term traversal needed in effecting substitutions with the
traversal needed for finding and reducing $\beta$-redexes. As an
example, consider the term $(\lambdax{x}\lambdax{y} t_1)\app t_2\app
t_3$. This term can be transformed through a sequence of
$\beta$-contractions to the form $t_1[x:= t_2, y := t_3]$. The
replacement in $t_1$ of $x$ by $t_2$ and $y$ by $t_3$ can now be done
at the same time and can also be combined with the identification and
rewriting of further $\beta$-redexes within $t_1$.

We have treated an environment or substitution up to now as
an auxiliary device, outside the term structure, to be used
essentially in implementing reduction. However, it is also possible to
include substitutions explicitly in terms, treating a
term with a substitution also as a term; such a term is similar to the
idea of a closure used in implementing functional programming
languages except that closures are now also treated as first-class
terms. If we allow substitutions to be used in this manner, so that we
permit the term $t$ in an environment of the form $[x:=t]$ to carry its
own environment, then we also obtain the ability to delay the
performance of substitutions so as to carry them out in a demand
driven fashion, thereby further enhancing the capability to combine
reduction and substitution traversals of terms. As an example,
consider the term
$(\lambdax{x}\ ((\lambdax{y}\ t_1)\ t_2))\ t_3$. This term can be
rewritten to the form $t_1[y:=t_2[x:=t_3],x:=t_3]$. Notice that in
this term we have delayed the substitution of $t_3$ for $x$ in
$t_2$. We may eventually need to reduce the term $t_2$ and the
mentioned substitution can then be carried out in the same traversal
as is needed for this reduction.

Implementations of functional programming languages typically use the
idea of environments to encode substitutions. A simplifying assumption
that is used in these contexts is that it is never necessary to look
at term structure embedded within abstractions. As a result of this
assumption, there is never any need to rename bound variables: the
terms that are being substituted are never carried into a context
where their free variables may get bound. The assumption of looking
within abstractions is, however, no acceptable in a situation where we
have to compare arbitrary $\lambda$-terms.
For instance, to decide the inequality of the terms
\begin{tabbing}
\dquad$(\lambdax{y}\ ((\lambdax{x}\ (\lambdax{y}\ x))\ y))$\quad and\quad
      $(\lambdax{y}\ ((\lambdax{x}\ (\lambdax{y}\ y))\ y))$,
\end{tabbing}
$\beta$-redexes inside abstractions have to be rewritten. Combining renaming
substitutions with $\beta$-contraction substitutions seems not to
be a problem when we use explicit names for bound variables. However,
the need to also consider $\alpha$-convertibility in comparing terms
usually dictates that an nameless representation be used for bound
variables. In such a situation, the descent into abstraction contexts
requires a lot more care. This issue is specifically dealt with in
explicit substitution calculi like the suspension calculus that we
discuss next.

\section{The Suspension Calculus}\label{sec:suspension}
Before the actual discussion on the suspension calculus, we first introduce
a notation of $\lambda$-terms proposed by de Bruijn~\cite{debruijn72}
that simplifies the task of checking for equality under $\alpha$-conversion.
In this notation, an occurrence of a variable is denoted by a
positive number, called a de Bruijn index, which counts the number
of abstractions between this occurrence and the abstraction binding
the variable.
For example, the term represented as
$(\lambdax{x}\app (\lambdax{y}\app (x \app y)) \app x)$ in a
name-based setting is denoted in
the de Bruijn notation by
$(\lambdadb(\lambdadb (\# 2 \app \# 1))\app \# 1)$.
It can in fact, be easily seen that any pair of $\alpha$-convertible
terms in the name-based notation
have the same de Bruijn representation.
It should be noted that bound variable renaming needed for substitution
propagation discussed in the previous section is not really eliminated
by the de Bruijn notation, but is, rather, transformed into
a form as the renumbering of de Bruijn indexes.
For example, upon pushing substitutions into an abstraction in
the context of the de Bruijn notation, it has to be properly reflected that
first the index corresponding to the variable that will be substituted should
become one greater than that is recorded in the current substitutions, and
second, the indexes corresponding to the variables occurring free in the
term that is to be substituted with should be increased by one.
Moreover, when an environment based reduction approach is under
consideration, a problem similar to what has been discussed in the
previous section also exists in combining a substitution corresponding
to a redex embedded in an abstraction, $\eg$,
$\lambdadb((\lambdadb t)\ s)$, to the enclosing environment:
in addition to the substitution of $s$
for the first free variable in $t$, the decreasing of the indexes
corresponding to the variables occurring free in $t$ should also be
properly reflected into the environment.
The details on how the required renumbering tasks are accomplished
in the context of the suspension calculus, which is based on the
de Bruijn representation, will become clear in the
discussions that follow.

It has been illustrated in the previous section
that the explicit maintenance of substitution environment
could be beneficial to the efficiency of the $\beta$-reduction
process. The explicit encoding of substitutions in term
representations provides a stronger control on the reduction and
substitution steps and thereby the flexibility of ordering
them towards further efficiency improvement to the overall term
comparison operation.
One such benefit is the ability to avoid unnecessary
performance of substitutions.
For example, consider the comparison of the pair
\begin{tabbing}
\dquad$(\lambdax{x}\ (x\ t))\ a$\quad and\quad  $(\lambdax{x}\ (x\ s))\ b$
\end{tabbing}
where $a$ and $b$ are different constants and $t$ and $s$ are some
complicated term structures. By reducing the redexes, substitutions
$[x:=a]$ and $[x:=b]$ are generated over $(x\ t)$ and $(x\ s)$ respectively.
It is obvious that the inequality of the terms is in fact entirely decided
by the results of applying the substitutions over the leading $x$'s,
and is irrelevant to those of $t$ and $s$. With the capability to record
substitutions along with other term structures, the generation and performance
of substitutions can be completely separated in an explicit substitution
calculus. This provides the chance to delay
the performance of the substitutions on $t$ and $s$, and consequently
to carry out the comparison on the structures
$(a\ (t[x:=a]))$ and $(b\ (s[x:= b]))$, which eventually avoids
the effort of effecting the delayed substitutions over $t$ and $s$.

Various explicit substitution calculi have been proposed for
reflecting substitutions into term structures, such as the
{\it suspension calculus}~\cite{NW98tcs}, the
{\it $\lambda\sigma$-calculus}~\cite{ACCL91},
the {\it $\lambda\upsilon$-calculus}~\cite{BBLR96},
the {\it $\lambda\xi$-calculus}~\cite{M96expsub} and
{\it $\lambda s_e$-calculus}~\cite{KR97}.
Among those calculi, the suspension calculus and the
$\lambda\sigma$-calculus are especially useful because besides the
lazy performance of substitutions, these notations also provide
support to combine substitutions generated from different
$\beta$-redexes; such a capability is essential for realizing the
sharing of structure traversal discussed in the previous section.
In this thesis we choose to use the suspension calculus because it
more closely attuned to practical applications in comparison with the
$\lambda\sigma$-calculus.

The terms of the suspension calculus are obtained from de Bruijn terms
essentially by adding a new form that is capable of representing a
term with a suspended substitution.
The full collection of terms is described formally by the
syntax rules in Figure~\ref{fig:susp_terms}.
\begin{figure}
\begin{tabbing}
\dquad\=$EnvTerm\ $\=\kill
\>{\it Term}    \> $::=\ C\ |\ \#I\ |\ ({\it Term}\ {\it Term})\ |\ (\lambdadb\ {\it Term})\ |\ \env{{\it Term}, N, N, {\it Env}}$ \\
\>{\it Env}     \> $::=\ {\it nil}\ |\ {\it EnvTerm :: Env}$\\
\>{\it EnvTerm} \> $::=\ \dum N\ |\ ({\it Term}, N)$
\end{tabbing}
\caption{The syntax of terms in the suspension calculus.}\label{fig:susp_terms}
\end{figure}
In these rules, $C$ represents constants, $N$ denotes the category of natural
numbers and $I$ represents the category of positive numbers.
Expressions of the form
$\env{t, ol, nl, e}$, referred to as {\it suspensions}, constitute the
new category of terms. Intuitively, such a suspension represents a
term $t$ whose first $ol$ free variables, \ie, those given by
de Bruijn indices ranging from $1$ to $ol$, should be
substituted for in a way determined by $e$ and whose other variables
should be renumbered to reflect the fact that $t$ originally appeared
inside $ol$ number of abstractions, but now appears within $nl$ of
them; $nl$ may be different from $ol$ either because some abstractions
enclosing $t$ have disappeared because of $\beta$-contractions or
because $t$ is being substituted into a context embedded within some
additional abstractions.
The environment $e$, that has the structure of a list, explicitly records
substitutions to be performed for the first $ol$ free variables
in $t$---the $i$th entry in this environment is intended to be the
substitution for the $i$th free variable. Consequently, $e$ should
have a length equal to $ol$ for the
term to be well-formed.
Two sorts of substitutions can be recorded in an environment.
One kind of substitution corresponds to abstractions
that persist even after some abstractions within whose scope they
appear disappear because of $\beta$-contractions.
Such substitutions are recorded in an environment by means of
expressions of the form $\dum l$, where $l$ is the count of the number
of abstractions within whose scope the one binding the variable in
question occurs; the difference between $l$ and the count of the
abstractions that persist at the point of substitution---given by $nl$
in  a term of the form $\env{t, ol, nl, e}$---determines the new index
for the variable being substituted for. Notice that from
this discussion it follows that, for any $\dum l$ that appears in the
environment $e$ in a well-formed
suspension $\env{t, i, j, e}$, it must  be the case that $l< j$.
The other sort of environment entry corresponds to the substitution for the
variable bound by an abstraction that disappears because of a
$\beta$-contraction. Such a
substitution is recorded by an expression of the form $(s, l)$. The natural
number $l$ records the number of abstractions within which the
$\beta$-redex whose contraction generated the substitution is
embedded; when the variable replacement is actually carried out, $l$
is used together with the embedding level at the point of replacement
to determine an adjustment for indexes
of free variables in $s$. From this it follows easily that a
suspension $\env{t, i, j, e}$ is well-formed only if it is the case
that $l \leq j$ for any $(s, l)$ contained $e$.

\begin{figure}
\begin{tabbing}
\quad\=(r11)\ \=\dquad\dquad\dquad\=\kill
\> ($\beta_s$)\> $((\lambdadb t_1)\app t_2) \ra \env{ t_1, 1,
0, (t_2,0) :: nil}$\\
\> ($\beta'_s$)\> $((\lambdadb
\env{t_1,ol+1,nl+1,\dum{nl}::e})\app t_2) \ra  \env{ t_1,ol+1,nl, (t_2,nl) :: e }$\\
\>(r1)\> $\env{c,ol,nl,e} \ra c$\\
\>\> provided $c$ is a constant\\
\>(r2)\>$\env{\#i,ol,nl,e} \ra \#j$\\
\>\>provided $i > ol$ and $j = i - ol + nl$.\\
\>(r3)\>$\env{\#i,ol,nl, e} \ra \#j$\\
\>\>provided $i \leq ol$ and $e[i] = \dum{l}$ and $j = nl - l$.\\
\>(r4)\>$\env{\#i,ol,nl,e} \ra \env{t,0,j,nil}$\\
\>\> provided $i \leq ol$ and $e[i] = (t,l)$ and $j = nl - l$.\\
\>(r5)\> $\env{(t_1\app t_2),ol,nl,e} \ra
(\env{t_1,ol,nl,e}\app \env{t_2,ol,nl,e})$.\\
\>(r6)\>$\env{(\lambdadb t), ol, nl, e} \ra (\lambdadb
\env{t, ol+1, nl+1, \dum{nl} :: e})$.\\
\>(r7)\>$\env{\env{t,ol,nl,e},0,nl',nil} \ra \env{t,ol,nl+nl',e}.$ \\
\>(r8)\>$\env{t,0,0,nil} \ra t$
\end{tabbing}
\caption{The rewriting rules for the suspension calculus.}\label{fig:susp_rules}
\end{figure}
The collection of terms is complemented in the suspension calculus
by a set of rewriting rules for simulating $\beta$-reduction.
The rules are present in Figure~\ref{fig:susp_rules}.
We use $e[i]$ to refer to the $i$th item in an environment. Among these
rules, ($\beta _s$) and ($\beta' _s$) generate the suspended substitutions
corresponding to the reduction of $\beta$-redexes;
rules (r1)-(r8), referred to as \emph{reading rules}, are used to
actually carry out those substitutions.

Now we use a concrete example to illustrate how $\beta$-reductions
can be performed in the suspension calculus.
Consider the term
\begin{tabbing}
\dquad $((\lambdadb ((\lambdadb (\lambdadb ((\#1\ \#2)\ \#3)))\ t_2))\ t_3)$,
\end{tabbing}
where $t_2$ and $t_3$ are arbitrary de Bruijn terms. Using rule
($\beta_s$) to reduce the outermost redex, the term is rewritten to
\begin{tabbing}
\dquad $\env{((\lambdadb (\lambdadb ((\#1\ \#2)\ \#3)))\ t_2),1,0,(t_3,0)::nil}$.
\end{tabbing}
Now the suspended substitution needs to be propagated into the top-level
application, which is accomplished by applying rule (r5).
\begin{tabbing}
\dquad $\env{(\lambdadb (\lambdadb ((\#1\ \#2)\ \#3))),1,0, (t_3,0)::nil}\ \env{t_2,1,0,(t_3,0)::nil}$.
\end{tabbing}
Using rule (r6) to push the substitution into the abstraction in the
suspension term on the left, the whole term is rewritten to
\begin{tabbing}
\dquad $(\lambdadb \env{(\lambdadb ((\#1\ \#2)\ \#3)),2,1, \dum{0}::(t_3,0)::nil})\ \env{t_2,1,0,(t_3,0)::nil}$.
\end{tabbing}
Now a new $\beta$-redex is revealed in the top-level term structure,
and the reduction of this redex can be simulated by rule ($\beta'_s$), which
directly combines the newly generated substitutions into the existing
environment.
\begin{tabbing}
\dquad $\env{(\lambdadb ((\#1\ \#2)\ \#3)),2,0, (\env{t_2,1,0,(t_3,0)::nil},0)::(t_3,0)::nil}$.
\end{tabbing}
By applying rules (r5)-(r8) several times, we finally get a term of form
\begin{tabbing}
\dquad $(\lambdadb ((\#1\ \env{t_2,1,1,(t_3,0)::nil})\ \env{t_3,0,1,nil}))$.
\end{tabbing}
Depending on the particular structures of $t_2$ and $t_3$, the
rewrite rules can be applied to finally produce a $\beta$-normal form of
the original term.

It can be observed that the rule $(\beta'_s)$ is in fact redundant if
our only purpose is to simulate $\beta$-reduction: whenever this rule is
applied, rule ($\beta_s$) is applicable too. However, this rule plays an
important role in our rewriting system because it serves to combine
the substitution generated from an redex with those already recorded in the
environment and thus shares the term traversals for reducing nested redexes.
A particular pattern is required by ($\beta'_s$) on the redex
to be reduced:
\begin{tabbing}
\dquad $((\lambdadb\env{t_1,ol+1,nl+1,@nl::e})\ t_2)$.
\end{tabbing}
This pattern matches the result of propagating the suspension
$\env{\lambdadb t_1, ol, nl, e}$ inside the abstraction, and
arises frequently in the presence of nested redexes when the
reduction process follows an outermost and leftmost order.

\section{Head Normalization and Head Reduction Sequences}\label{sec:hnorm}
The capability of the suspension calculus to simulate
$\beta$-reductions
in the conventional $\lambda$-calculus is justified
in~\cite{nadathur99finegrained}
in two steps. First, it is shown that each well-formed term in the suspension
calculus can be transformed into a de Bruijn term by applying a finite
sequence of reading rules for carrying out the
suspended substitutions. Second, it can be shown that a de Bruijn term $t$
$\beta$-reduces to $s$ if and only if $t$ can be transformed to $s$ by
applying a finite sequence of rules in Figure~\ref{fig:susp_rules}.

As noted already, it is beneficial to interleave the performance of
substitutions also with the process of comparing terms. To justify
this at a formal level, it is necessary to lift the notion of head
normal forms to the suspension calculus. The following definition does
this after restating the definition for such forms in the de Bruijn
setting.

\begin{defn}
A de Bruijn term is in head normal form if it has the structure
\begin{tabbing}
\squad $(\lambdadb\ \ldots (\lambdadb\ (\ldots (h\ t_1)\ \ldots\ t_m))\ldots)$,
\end{tabbing}
where $h$ is a constant or a de Bruijn index. As before, we call $t_1,
\ldots, t_m$ the arguments of such a term, we call $h$ its head, we
call the abstractions in the front its binder and we refer to the
number of such abstractions as the binder length. By a harmless abuse
of notation, we permit the number of arguments and the binder
length to be $0$ in such a form. The notion of a head normal form is
extended to the suspension calculus setting by allowing the arguments
of such a form to be arbitrary suspension terms.
\end{defn}

The algorithm that we have previously described for comparing two
terms in the named calculus has an obvious adaptation to the de Bruijn
setting; the essential difference is, in fact, that the adjustment to
names of bound variables using $\alpha$-conversions is obviated. The
following proposition, proved in \cite{nadathur99finegrained}, allows
this algorithm to be adapted to the suspension calculus context.

\begin{prop}
Let $t$ be a de Bruijn term and suppose that the rules in
Figure~\ref{fig:susp_rules} allow $t$ to be rewritten to a head normal
form in the suspension calculus with $h$ being the head, $n$ being the binder
length and $t_1,\ldots ,t_m$ being the arguments. Let $\rnf{t_i}$ be
the de Bruijn term obtained from $t_i$ by a series (maybe empty) of
applications of the reading rules.
Then $t$ has the term
\begin{tabbing}
\squad$(\lambdadb\ \ldots( \lambdadb\ (\ldots (h \ \rnf{t_1})\ \ldots\rnf{t_m}))\ldots)$
\end{tabbing}
with a binder length of $n$ as a head normal form in the context of
the de Bruijn notation.
\end{prop}

A critical part of using the comparison algorithm is that of
generating a head normal form for a term. Such a form is best
generated by rewriting a {\it head redex} of the term at each stage;
a sequence of such rewritings is what is referred to as a {\it head
  reduction sequence}. In the de Bruijn setting, a term that is not in
head normal form has a unique head redex that is identified as follows:
\begin{enumerate}
\item If the term is a $\beta$-redex, then the term itself is its head redex;
\item Otherwise, if the term is of form $(\lambdadb\ t)$ or $(t\ s)$, then
its head redex is that of $t$.
\end{enumerate}
In this setting it is also a fact that a head reduction sequence
will always succeed in producing a head normal form for a term
whenever it has such a form \cite{Bar81}.

In the suspension calculus, there is one more kind of term and
there is also a larger set of rewriting rules. Moreover, the use of an
environment to record substitutions also leads to the possibility of
sharing subparts of terms. Taking these aspects into account, we can
generalize the notion of head redex and defines the head reduction
sequence in  the context of the suspension calculus as the following.
\begin{defn}
Let $t$ be a suspension term that is not in head normal form.
\begin{enumerate}
\item Suppose that $t$ has the form $(t_1\ t_2)$.
If $t_1$ is an abstraction, then $t$ is its sole head
redex. Otherwise the head redexes of $t$ are the head redexes of $t_1$.

\item If $t$ is of the form $(\lambdadb\ t_1)$,
its head redexes are identical to those of $t_1$.

\item If $t$ is of the form $\env{t_1,ol,nl,e}$,
then its head redexes are all the head redexes of $t_1$ and $t$ itself
provided $t_1$ is not a suspension.
\end{enumerate}

Let two subterms of a term be considered non-overlapping just in
case neither is contained in the other. Then a {\it head
reduction sequence} of a suspension term $t$ is a sequence $t =
r_0,r_1,r_2,\ldots,r_n,\ldots,$ in which, for $i \geq 0$, there is
a term succeeding $r_i$ if $r_i$ is not in head normal form
and, in this case, $r_{i+1}$ is obtained from $r_i$ by
simultaneously rewriting a finite set of non-overlapping subterms
that includes a head redex using the rule schemata in
Figure~\ref{fig:susp_rules}. Obviously, such a sequence terminates
if for some $m \geq 0$ it is the case that $r_m$ is in head normal form.
\end{defn}

The usefulness of this definition is based on the proposition below:
\begin{prop}
A term $t$ in the suspension calculus has a head normal
form if and only if every head reduction sequence of $t$ terminates.
\end{prop}

A detailed proof of this proposition can be found
in~\cite{nadathur99finegrained}, which essentially
maps the head reduction sequences of suspension terms to the corresponding
ones in the context of the de Bruijn notation. By virtue of this
proposition, we can base the comparison of terms on a procedure that
exploits the suspension form to delay substitutions and that
essentially picks a head reduction sequence to try and reduce a given
term to a head normal form. Notice that, unlike in the case of de
Bruijn terms, there can actually be a choice in the head redex to
rewrite at each stage. This non-determinism
provides a flexibility that can be exploited by practical
reduction procedures, a topic that we elaborate on
in Chapter~\ref{chp:machineTermRep}.

\section{The Suspension Calculus and $\eta$-conversions}\label{sec:suspeta}

\begin{figure}
\begin{tabbing}
\quad\=(r11)\ \=\dquad\dquad\dquad\=\kill
\> ($\eta_s$)\> $t \ra {\displaystyle \underbrace{\lambdadb\ldots
  \lambdadb}_n} (\env{t, 0, n, nil}\app \#n\app \ldots\app \#1)$\\
\>\> provided $n > 0$.
\end{tabbing}
\caption{The $\eta$-rule in the suspension calculus.}\label{fig:suspeta}
\end{figure}

In comparing terms, we have also to take into account that our
equality notion includes $\eta$-conversions. In the conventional
setting, this fact is accommodated by allowing the comparison
procedure to use $\eta$-conversions to adjust binder lengths in case
the reduction process yielded two head normal forms for which these
were unequal. A similar adjustment can be carried out also when the
suspension calculus is used. The basis for such an adjustment is a
special form of the $\eta$-rule for this setting. The relevant rule is
presented in Figure~\ref{fig:suspeta}. This rule has an additional proviso when
types are associated with terms: $t$ must have a function
type that has at least $n$ argument types. Notice also that some of
the reading rules can also be compiled into the application of this
rule when it is used to adjust the binder length in a head normal
form. Thus, the head normal form
\begin{tabbing}
\qquad\=\kill
\>${\displaystyle \underbrace{\lambdadb \ldots \lambdadb}_k} (h\app
t_1 \app t_m)$
\end{tabbing}
can be rewritten to the form
\begin{tabbing}
\qquad\=\kill
\>${\displaystyle \underbrace{\lambdadb \ldots \lambdadb}_{k+n}}
(h\app \env{t_1,0,nl,nil} \app \env{t_m,0,nl,nil}\app \#n \app
\ldots\app \#1)$
\end{tabbing}
if $h$ is a constant and to the form
\begin{tabbing}
\qquad\=\kill
\>${\displaystyle \underbrace{\lambdadb \ldots \lambdadb}_{k+n}}
(\#j\app \env{t_1,0,nl,nil} \app \env{t_m,0,nl,nil}\app \#n \app
\ldots\app \#1)$
\end{tabbing}
where $j$ is $i+n$ if $h$ is the de Bruijn index $\#i$.

\chapter{An Abstract Interpreter for $\lambda$Prolog}\label{chp:interpreter}
A high-level description of the computation model of the
$\lp$ language has been provided in Chapter~\ref{chp:language}. This
description is helpful for understanding $\lp$ programs,
but is not quite suitable as a basis for implementation.
For the latter purpose, concrete mechanisms have to be provided
first for deciding proper instances for existentially
quantified variables in solving goals of form $\exists{x}G$ and for universally
quantified variables in clauses for solving atomic goals,
second for selecting clauses for solving atomic goals in the presence
of multiple candidates as well as for picking the disjunct to solve when
processing disjunctive goals and third for controlling the scopes of
constants and program clauses with respect to
generic and augment goals.

In this chapter, we refine the computation model appearing in
Section~\ref{sec:hohh} into an abstract
interpreter for $\lp$ that includes solutions to all the issues mentioned
above.
We begin in Section~\ref{sec:general_unif} with the issue of finding
instances for variables existentially quantified in goals and
universally quantified in clauses. Towards this end, we introduce a
new category of variables, the {\it logic} variables, into the term
representation and we generalize term comparison into an equation
solving operation called {\it unification} that is based on the new
representation.
Section~\ref{sec:interpreter} presents an abstract
interpreter for $\lp$ that uses this operation.
In Section~\ref{sec:pattern_unif}, a particular form of the
unification problem for $\lambda$-terms is described and a practical
algorithm is presented for solving such problems. Problems in this
class are what are referred to as higher-order pattern
unification problems. This thesis is concerned only with solving
such problems completely and we assume a refinement of the abstract
interpreter that uses only the algorithm that we present for solving
such problems in the rest of the thesis.

\section{Logic Variables and Unification}\label{sec:general_unif}
The problem of deciding suitable instances for existentially
quantified variables in goals and universally quantified variables
in clauses is one that is also faced in the implementation of
{\it Prolog}. It is solved in that setting by delaying
the selection of an instance till a later point in computation
when enough information is available for making the ``right''
choices.
We adopt this solution also in our context.
Specifically, when a goal $\exists{x}G$ is encountered, a new
variable $X$ that can be instantiated
in the course of computation is introduced to replace
$x$ in $G$; this variable, that is different from traditional
variables in logic in that it can actually be instantiated
in the search for a proof, is what is known as a {\it logic
  variable}. Note that in a setting where types are present, $X$
should have the same type as the quantified variable it replaces. Once
this variable is introduced, computation proceeds to solve the goal
$G[x:=X]$.
The actual instantiation of $X$ is determined
at the point of solving the atomic goals contained by $G$ through
{\it unification}. This is a process or computation that allows us to
pick instantiations for logic variables so as to make two terms
equal. Thus, suppose that we have reached a point where the atomic
goal $A'$ has to be solved.
We then look for a clause of the form
$\forall{x_1}\ldots\forall{x_n}A$ or
$\forall{x_1}\ldots\forall{x_n}(\impxy{G}{A})$ such that by
replacing the universally quantified variables in the front of
this clause with new logic variables $X_1,\ldots,X_n$, we get
an expression of the form $A''$ or $\impxy{G''}{A''}$ that has the
characteristic that $A''$ and $A$ can be unified; in the second case,
this leads to a subsequent attempt to solve the corresponding instance
of $G''$.

The unification operation generalizes the usual term comparison
in the sense that we are also allows to compute substitutions for
logic variables to make the terms under consideration equal.
There is a proviso in our context that the substitution computed for a
variable by this process should be of the same type as the
variable. Further, in the context
of unifying $\lambda$-terms, substitutions for such variables should
also make sure that the free variables in the terms being introduced
do not get accidentally bound.
A correct characterization of such substitution can be provided
by using the equality of $\lambda$-terms.
A substitution is typically given by a set of pairs of the form
$\{\dg{X_i}{t_i}|1\leq i\leq n\}$ where the first element of the pair
is the variable being substituted for and the second element is the
term that it should be replaced with. The application of such a
substitution to the term $t$ can be given by the term
$(\lambdax{X_1}\ldots\lambdax{X_n}\ t)\ t_1\ \ldots\ t_n$.

Since we have to eventually deal with logic variables in an
implementation, we extend the suspension calculus to accommodate
them. As we have already noted, these variables have a different
character from the usual variables in $\lambda$-terms and so we
include a new category for them in the syntax. We shall write such
variables with a starting uppercase letter. We also add a special
rewrite rule pertaining to such variables that is shown in
Figure~\ref{fig:susp_r9}.
\begin{figure}
\begin{tabbing}
\quad\=(r11)\ \=\dquad\dquad\dquad\=\kill
\>(r9)\>$\env{X,ol,nl,e} \ra X$, provided $X$ is a logic variable.
\end{tabbing}
\caption{The rewriting rule in the suspension calculus for logic
variables.}\label{fig:susp_r9}
\end{figure}
The rule is justified by the fact that substitutions for logic
variables cannot be captured by enclosing abstractions and hence
cannot be affected by any reduction or renumbering substitutions.
With the addition of logic variables, we have also to extend our
definition of head normal forms to include the case where the head is
also such a variable. We shall say now that a head normal form is {\it
  flexible} if it has a logic variable as its head and that it is {\it
  rigid}   if the head is a constant or de Bruijn index.

A unification problem in the context of the typed $\lambda$-calculus
is known as a {\it higher-order unification problem}. Such a
problem can be represented by a {\it disagreement set} that is a
finite collection of pairs of $\lambda$-terms, known as the {\it
  disagreement pairs}, in which the two terms in each pair have equal
types. A solution to, or a {\it unifier} for, the problem is
substitution for logic variables---also represented as a set
of pairs of terms as discussed earlier---that is such that it makes
the two terms in each pair in the disagreement set equal when it is
applied to them.
A useful notion in the context of unification is that of a {\it most
general unifier}. This is a unifier for a disagreement set from which
any other unifier for the set can be obtained through further
substitutions for logic variables. Unfortunately higher-order
unification does not admit of most general unifiers. Particular
problems may, in fact, have an infinite set of unifiers none of which
can be obtained from others in the set through further substitutions.
A further observation is that no procedure can be provided that
computes a covering set of unifiers in a non-redundant way. However, a
non-redundant search can be carried out to determine
unifiability. Huet has in fact described a procedure that carries out
such a search \cite{Huet75TCS}. This procedure computes initial
portions of unifiers that are known as {\it pre-unifiers}. In several
instances, the pre-unifiers that it computes turn out actually to be
complete unifiers for the problem under consideration.

Huet's procedure consists of two phases, which are
repetitively invoked on a given disagreement set to transform it into
a form from which it can be decided that no unifier exists or for
which unifiability
is evident. Since equality is based on the rules of
$\lambda$-conversion, we can assume that the two terms in each
disagreement pair in a unification problem are in head normal form and
that their binders have been adjusted to have the same length. Now,
the first phase of Huet's unification procedure handles
pairs in which both terms are rigid, \ie,
{\it rigid-rigid} pairs, in a way similar to
term simplification in first-order unification: depending on whether
or not
the two heads are equal, the unification problem is simplified to one
consisting of pairs formed out of the arguments or non-unifiability is
determined. The second phase of Huet's algorithm
considers {\it flexible-rigid} pairs, and attempts to bind logic
variables as the heads of the flexible terms. In particular, assuming
the logic variable $X$ is the head of the flexible term, a substitution
of form
$\dg{X}{\lambdadb\ldots\lambdadb(r\ (H_1\ \#n\ \ldots\#1)\ \ldots\ (H_m\ \#n\ \ldots\ \#1))}$
is produced, where $H_1\ldots H_m$ are new logic variables of proper types.
The substituted term has the binder length $n$ that is decided by the number
of arguments in the type of $X$.
The head $r$ can be a de Bruijn index $\#i$, for $1\leq i\leq n$,
when the ith argument of the type of $X$ has $m$ arguments and has target
being the same as the that of the type of $X$,
or a constant $c$ when the rigid term has $c$ as its head,
and the type of $c$ has $m$ arguments.
Observations that are important to our discussions should be made on the
following two issues.
First, multiple bindings can be found for the same logic variable during the
binding phase, and they cannot be obtained from each other by performing
further substitutions. Second, the types of logic variables
play an important role in determining the structures of the bindings;
in particular, it is used to decide whether a de Bruijn index can be made
the head of the binding term. The details on unifying flexible-rigid pairs
in Huet's algorithm is beyond the scope of this thesis
and we refer interested readers to~\cite{Huet75TCS}.
In addition to the rigid-rigid and flexible-rigid cases,
pairs containing {\it flexible-flexible} terms may also occur during
unification. A pair of this sort is known as always unifiable, but a
complete search for the unifiers can be
unconstrained~\cite{Huet75TCS}.
Huet's algorithm treats a set consisting of only such pairs as a success
without further exploring the underlying unifiers.

The iterative use of the term simplification and binding phases in
Huet's algorithm naturally forms a branching search. If the searching
process terminates, either non-unifiability is determined or a
finite complete set of unifiers up to flexible-flexible pairs
for the given disagreement set is produced.

The undecidability property of higher-order unification manifests
itself in the fact that the search conducted by Huet's procedure may
not find a success at any finite depth, \ie, the search may go on for
ever. Even when successes are found at finite depth, the search still
not terminate because the number of successes to be found may be
infinite.

The theoretical properties of higher-order unification and the
branching search that must be conducted make it seem as if such
unification cannot be used effectively in a practical setting.
However, the actual
utilization of Huet's procedure in several programming systems,
including an implementation of $\lp$ that is known as {\it Teyjus
  Version 1} and that is based on attempting to solve the complete
set of higher-order unification problems, has demonstrated a
practical usefulness for this kind of computation. In particular, it
has been revealed that there is a wide collection of
application tasks in which the
unification problems that need to be solved in fact have unique
solutions.
Based on a study of the usage of higher-order unification in
these examples, Dale Miller has identified a subset
of the general problem that is known as the $\Ll$ or the higher-order
pattern class~\cite{Miller91jlc, Nipkow93}.
The problems in this subset occur when existential variables in
queries and universal variables in program clauses are used in a
restricted way. Unifiability for this subset is known to be decidable
and it is also known that a single most general unifier can be provided
in any of the cases where a unifier exists.
An empirical study conducted by Michaylov and Pfenning~\cite{MP92}
shows that even if we do not restrict the syntax of programs at the
outset to ensure that unification problems outside the $\Ll$ class are
not generated,
95\% of the unification problems occurring in the
computations underlying practical $\lp$ applications are first-order, and the
remaining evolve into problems belonging to the $\Ll$ subset once
substitutions determined by looking at other disagreement pairs are
made for logic variables.

When applied to higher-order pattern problems, Huet's procedure
is guaranteed to terminate and will do so with a unique successful
branch.
However, by fully taking advantages of the $\Ll$ restriction,
Huet's procedure can be further improved.
First, the unique solution to a problem in the $\Ll$ subset may be
found by Huet's procedure through a branching search, which is
known to be expensive in performance.
Second, Huet's procedure only partially computes the solutions for
flexible-flexible pairs, whereas the complete solution for such
pairs in the $\Ll$ subset can be found in a controlled way.
Third, it is known that the types of logic variables have no impact
on the structures of the unifiers of $\Ll$ problems, and consequently
the maintenance and examination effort required by Huet's algorithm
for such information becomes completely redundant.
Improvements of this sort have already been proposed by Dale Miller,
which lead to a simpler and more efficient approach for solving
pattern unifications. This approach is, however, described at
a high level in a non-deterministic manner. Research
conducted by Nadathur and Linnell in~\cite{NL05HOP} further
refines Miller's algorithm into one that is suitable to be used
as the basis of actual implementations by seriously taking the efficiency
of the algorithm into account. This approach is adopted in the
implementation scheme for $\lp$ underlying this thesis.

A critical part of defining higher-order pattern unification problems
and the algorithm for solving them is paying attention to the scopes
of quantifiers that give rise to logic variables and constants. The
logic underlying $\lp$ has the capability to mix such scopes
richly---for example, existential and universal quantifiers can be
used in arbitrary order over goals. However, to develop the discussion
in a way that leads naturally into an implementation of $\lp$, it is
useful to have available a particular approach to encoding and treating
quantifier scopes. We include this mechanism in an abstract
interpreter in the next section before explicitly taking on the
discussion of higher-order pattern unification in
Section~\ref{sec:pattern_unif}.

\section{An Abstract Interpreter}\label{sec:interpreter}

The model of computation presented in Section~\ref{sec:hohh} can be
refined into a state transition system whose purpose is essentially to
simplify a set of goals till they all are completely solved. The
reason for considering a set of goals in a state as opposed to a
single goal is that we allow for conjunctions in $G$ formulas: to
solve such a goal, we have to solve both goals. Another thing to note
is that in the presence of augment goals it is necessary also to
include the available program as a component of a state. However,
programs must parameterize the solution of particular goals and not
the entire set: in trying to solve the goal $D \supset G$, we get to
use the clause $D$ in solving $G$ but not in solving all the other
goals present in the set. We would also like to include in the
treatment a realistic model for finding substitutions for
existentially quantified variables in goals. For this reason, we
add to the state a disagreement set representing a unification
problem that still has to be solved and a substitution $\theta$ that
is proposed as a solution to parts of the overall unification
problem. Finally, we need to keep track of the constants and logic
variables that have already been introduced in the search up to this
point so as to make sure we do not reuse them.

Universal goals will be treated in the abstract interpreter in a
similar manner to that in the high-level description of computation:
they will be instantiated with new constants. For existential
quantifiers, we will use the idea of instantiating with logic
variables as already indicated. However, we
have to be careful to take into account the order in which these
quantifiers are encountered for correctness. As a concrete example,
consider an attempt to solve the goal $\somex{y}\allx{z}(p\ y\ z)$
given a program that contains the
clause $\allx{x}(p\ x\ x)$. Following the expected approach leads to
the disagreement set $\{\dg{p\ Y\ c}{p\ X\ X}\}$ being given to the
unification procedure; $Y$ and $X$ are logic variables here that have
been introduced for the purpose of instantiating the existential
quantifier in the goal and the universal quantifier in the program
clause and $c$ is a new constant introduced when the universal
quantifier in the goal is processed. Now, if we proceed naively with
unification, this disagreement set can be solved by instantiating $Y$
and $X$ to $c$. Unfortunately, this solution is incorrect because
instantiating $X$ with $c$ corresponds to producing a computation
sequence according to the high-level description in which the constant
introduced for the universal quantifier in the goal is not new.

The particular point that we have to pay attention to in order to
avoid bad solutions like that discussed above is that logic variables
can only be instantiated with terms from a signature that is in
existence at the time when these variables are introduced. A practical
way to realize this constraint in unification is to think of the term
universe as growing in stages, with each universal quantifier
introducing a new stage \cite{N93JAR}. Calling each stage a universe
level, we can think of labeling each constant with the universe level
at which it enters the signature. We can then also label logic
variables with universe levels to indicate the maximum level that can
be attached to a constant that appears in a term instantiating the
variable.

To realize this scheme within our abstract interpreter, we shall
include with each state a labeling function that assigns universe
levels to (finite sets of) constants and variables associated with the
state. Since the domain of this function is finite, we will sometimes
depict it by its graph, \ie, we will show it as a set of ordered
pairs. We further associate with each
goal the value of the universe level at the start of the processing of
that goal; this universe level will be manipulated by embedded
universal goals and will be used to label logic variables that are
generated in processing. We shall also need to make sure that
substitutions for logic variables are consistent with labeling
functions. Unification must produce substitutions that respect
labelings and they may lead to modifications to labelings needed to
ensure that subsequent substitutions will not violate the dependencies
generated by earlier ones. The following definition introduces the
notions needed to formalize these requirements:

\begin{defn}
A labeling function $\st{L}$ is a mapping from a finite collection of
logic variables and constants to natural numbers.
Let $\theta = \{\dg{X_i}{t_i}|1\leq i \leq n\}$ be a substitution, and let
$\st{L}$ be a labeling function. Then $\theta$ is proper with respect
to $\st{L}$ if for $1 \leq i \leq n$ it is the case that $\st{L}(c) \leq
\st{L}(X_i)$ for any constant $c$ appearing in $t_i$.
The labeling induced by $\theta$ and $\st{L}$ in this case is a
labeling function that is written as $\st{L}_\theta$. This function
behaves identically to $\st{L}$ on constants and on logic variables it
is such that
\begin{tabbing}
\dquad$\st{L_\theta}(X) =
min(\{\st{L}(X_i)|\dg{X_i}{t_i}\in\theta$ {\it and
  $X$ appears in $t_i$}$\})$
\end{tabbing}
if the variable is new, \ie, does not have a universe index already
assigned to it and is
\begin{tabbing}
\dquad$\st{L_\theta}(X) =
min(\{\st{L}(X)\}\cup\{\st{L}(X_i)|\dg{X_i}{t_i}\in\theta$ {\it and
  $X$ appears in $t_i$}$\})$
\end{tabbing}
otherwise.
\end{defn}

Having provided the intuition behind the abstract interpreter
structure, we now begin to present it formally. The first aspect to be
made precise is the structure of a state within the interpreter.

\begin{defn}
A computation state is a tuple of form
$\tuple{\st{G}, \st{D}, \st{C}, \st{V}, \st{L}, \theta}$ where
\begin{enumerate}
\item $\st{G}$ is a set of triples of the form $\tuple{G, \st{P}, N}$ where
  $G$ is a goal, $\st{P}$ is a collection of program clauses and $N$
  is a natural number,
\item $\st{D}$ is a disagreement set,
\item $\st{C}$ and $\st{V}$ are (finite) sets of constants and
logic variables respectively,
\item $\st{L}$ is a labeling function whose domain is $\st{C}\cup
  \st{V}$, and
\item $\theta$ is a substitution for logic variables.
\end{enumerate}
\end{defn}

The syntax that we have used for program clauses in the logic
underlying $\lp$ allows them to have a conjunctive structure. This is
useful, for instance, in writing augment goals but in describing
computation it is preferable to be dealing only with clauses of the
form $\allx{x_1}\ldots\allx{x_n}A$ or
$\allx{x_1}\ldots\allx{x_n}(\impxy{G}{A})$ where $A$ is an atomic
formula. We describe a function on program clauses that allows us to
extract a set of clauses in this reduced form from them.
\begin{defn}
The elaboration of a program clause $D$, denoted by {\it elab$(D)$}, is
the set of formulas defined as the follows:
\begin{enumerate}
\item If $D$ is an atomic formula $A$ or of the form
$\impxy{G}{A}$, then it is $\{D\}$.
\item If $D$ is $\andxy{D_1}{D_2}$,
then it is {\it elab$(D_1)$} $\cup$ {\it elab$(D_2)$}.
\item If $D$ is $\allx{x}{D_1}$
then it is $\{\allx{x}{D_2}\ |\ D_2 \in {\it elab}(D_1)\}$.
\end{enumerate}
The elaboration of a program $\mathcal{P}$ is the union of the
elaboration of all the clauses in $\mathcal{P}$.
\end{defn}

We now formalize the notion of state transitions that underlies our
abstract interpreter for the logic underlying $\lp$.
\begin{defn}\label{def:interpreter}
A state $\tuple{\st{G}_2, \st{D}_2, \st{C}_2, \st{V}_2, \st{L}_2,
\theta_2}$ is derivable from another state of form $\tuple{\st{G}_1,
\st{D}_1, \st{C}_1, \st{V}_1, \st{L}_1, \theta_1}$ if one of the
following holds.
\begin{enumerate}
\item $\tuple{\top,\st{P},N}\in \st{G}_1$, $\st{G}_2 = \st{G}_1-\{\tuple{\top,\st{P},N}\}$,
    $\st{D}_2 = \st{D}_1$,  $\st{C}_2=\st{C}_1$, $\st{V}_2=\st{V}_1$,
    $\st{L}_2 = \st{L}_1$ and $\theta_2 = \emptyset$.
\item $\tuple{\andxy{G_1}{G_2},\st{P},N}\in \st{G}_1$,  $\st{G}_2 =
(\st{G}_1-\{\tuple{\andxy{G_1}{G_2},\st{P},N}\})\cup\{\tuple{G_1,\st{P},N}, \tuple{G_2,\st{P},N}\}$,
    $\st{D}_2 = \st{D}_1$,  $\st{C}_2=\st{C}_1$, $\st{V}_2=\st{V}_1$,
    $\st{L}_2 = \st{L}_1$ and $\theta_2 = \emptyset$.
\item $\tuple{\orxy{G_1}{G_2},\st{P},N}\in \st{G}_1$, for $i=1$ or $i=2$,
   \begin{tabbing}
    \dquad\dquad$\st{G}_2 = (\st{G}_1-\{\tuple{\orxy{G_1}{G_2},\st{P},N}\})\cup\{\tuple{G_i,\st{P},N}\}$,
   \end{tabbing}
   $\st{D}_2 = \st{D}_1$,  $\st{C}_2=\st{C}_1$, $\st{V}_2=\st{V}_1$, $\st{L}_2 = \st{L}_1$ and $\theta_2 = \emptyset$.
\item $\tuple{\somex{x}{G},\st{P},N}\in \st{G}_1$, for a logic variable $X\notin\st{V}_1$,
   \begin{tabbing}
   \dquad\dquad  $\st{G}_2 = (\st{G}_1-\{\tuple{\somex{x}{G},\st{P},N}\})\cup\{\tuple{G[x:=X],\st{P},N}\}$,
   \end{tabbing}
   $\st{D}_2 = \st{D}_1$, $\st{C}_2=\st{C}_1$, $\st{V}_2=\st{V}_1\cup\{X\}$, $\st{L}_2 = \st{L}_1\cup\{\dg{X}{N}\}$ and $\theta_2 = \emptyset$.
\item $\tuple{\impxy{D}{G},\st{P},N}\in \st{G}_1$, $\st{G}_2 = (\st{G}_1-\{\tuple{\impxy{D}{G},\st{P},N}\})\cup\{\tuple{G_1,\andxy{\st{P}}{D},N}\}$,
    $\st{D}_2 = \st{D}_1$, $\st{C}_2=\st{C}_1$, $\st{V}_2=\st{V}_1$,
    $\st{L}_2 = \st{L}_1$ and $\theta_2 = \emptyset$.
\item $\tuple{\allx{x}{G},\st{P},N}\in \st{G}_1$, for a constant $c\notin\st{C}_1$,
   \begin{tabbing}
   \dquad\dquad  $\st{G}_2 = (\st{G}_1-\{\tuple{\allx{x}{G},\st{P},N}\})\cup\{\tuple{G[x:=c],\st{P},N+1}\}$,
   \end{tabbing}
   $\st{D}_2 = \st{D}_1$, $\st{C}_2=\st{C}_1\cup\{c\}$, $\st{V}_2=\st{V}_1$, $\st{L}_2 = \st{L}_1\cup\{\dg{c}{N+1}\}$ and $\theta_2 = \emptyset$.
\item Let $\tuple{A, \st{P}, N}\in\st{G}_1$,
      let $\allx{x_1}\ldots\allx{x_n}A'\in{\it elab}(\st{P})$ and,
      for $1\leq i\leq n$,
      let $X_i$ be a distinct logic variable such that $X_i\notin\st{V}_1$.
      Further, assume
      \begin{tabbing}
      \dquad\dquad$\st{D}' = \st{D}_1\cup\{\dg{A}{A'[x_1:=X_1]\ldots[x_n:=X_n]}\}$ and\\
      \dquad\dquad $\st{L}' = \st{L}_1\cup\{\dg{X_1}{N},\ \ldots,\dg{X_n}{N}\}$.
      \end{tabbing}
      Suppose that a unification procedure applied to $\st{D}'$
      produces a substitution $\sigma$ that is proper with respect to
      $\st{L}'$, and a disagreement set $\st{D}''$.
      Then $\st{G}_2 = \sigma(\st{G}_1-\{\tuple{A,\st{P},N}\})$,
      $\st{D}_2 = \st{D}''$, $\st{C}_2 = \st{C}_1$,
      $\st{V}_2 = \st{V}_1\cup\{X_1,\ \ldots,X_n\}$,
      $\theta_2 = \sigma$ and
      $\st{L}_2 = \st{L'}_\sigma$.
\item Let $\tuple{A, \st{P}, N}\in\st{G}_1$,
      let $\allx{x_1}\ldots\allx{x_n}(\impxy{G}{A'})\in{\it elab}(\st{P})$ and,
      for $1\leq i\leq n$,
      let $X_i$ be a distinct logic variable such that $X_i\notin\st{V}_1$.
      Further, assume
      \begin{tabbing}
      \dquad\dquad$\st{D}' = \st{D}_1\cup\{\dg{A}{(\impxy{G}{A'})[x_1:=X_1]\ldots[x_n:=X_n]}\}$ and\\
      \dquad\dquad $\st{L}' = \st{L}_1\cup\{\dg{X_1}{N},\ \ldots,\dg{X_n}{N}\}$.
      \end{tabbing}
      Suppose that a unification procedure applied to $\st{D}'$
      produces a substitution $\sigma$ that is proper with respect to
      $\st{L}'$, and a disagreement set $\st{D}''$.
      Then
      \begin{tabbing}
      \dquad\dquad$\st{G}_2 = \sigma((\st{G}_1-\{\tuple{A,\st{P},N}\})\cup\{\tuple{G[x_1:=X_1]\ldots[x_n:=X_n],\st{P},N}\})$,
      \end{tabbing}
      $\st{D}_2 = \st{D}''$, $\st{C}_2 = \st{C}_1$,
      $\st{V}_2 = \st{V}_1\cup\{X_1,\ \ldots,X_n\}$,
      $\theta_2 = \sigma$ and
      $\st{L}_2 =\st{L}'_\sigma$.
\end{enumerate}
A sequence of the form
$\tuple{\st{G}_1, \st{D}_1, \st{C}_1, \st{V}_1, \st{L}_1, \theta_1}$,
$\ldots$,
$\tuple{\st{G}_n, \st{D}_n, \st{C}_n, \st{V}_n, \st{L}_n, \theta_n}$
is a derivation sequence if the $(i+1)$th tuple in it is derived from the
$i$th tuple. Such a derivation sequence
terminates if no tuple can be derived from
$\tuple{\st{G}_n, \st{D}_n, \st{C}_n, \st{V}_n, \st{L}_n, \theta_n}$.
\end{defn}

\begin{defn}\label{def:interp}
Let $G$ be a closed goal formula, let $\st{P}$ be a set of closed program
clauses and let $\st{C}$ be the set of constants occurring in $G$ and
$\st{P}$.
Further, let $\st{L}$ be a labeling function of form $\{\dg{c}{0}|c\in\st{C}\}$.
Now assume $\st{G}_1 = \{\tuple{G, \st{P}, 0}\}$, $\st{D}_1 = \emptyset$,
$\st{C}_1 = \st{C}$, $\st{V}_1 = \emptyset$, $\st{L}_1 = \st{L}$ and
$\theta = \emptyset$. Then a derivation sequence of the
form $\tuple{\st{G}_1, \st{D}_1, \st{C}_1, \st{V}_1, \st{L}_1,
\theta_1}$, $\ldots$, $\tuple{\st{G}_n, \st{D}_n, \st{C}_n,
\st{V}_n, \st{L}_n, \theta_n},\ldots$ is said to be a
$\st{P}$-derivation sequence for $G$. Such a sequence may terminate
because no further rules are applicable to the last tuple in it. If
such termination occurs at the $m$th tuple because $\st{G}_m$ is empty
and $\st{D}_m$ is either empty or contains only flexible-flexible
pairs, then the
sequence is called a $\st{P}$-derivation of $G$. A sequence of this
kind embodies a
solution to the query $G$ in the context of the program $\st{P}$ and
the answer substitution corresponding to it is obtained by composing
$\theta_m\circ\ldots\circ\theta_1$ with any unifier for $\st{D}_m$
and restricting the results substitutions to the logic variables
corresponding to the top-level existentially quantified variables in
$G$.
\end{defn}

An abstract interpreter for our language can be described as one that
searches for a $\st{P}$-derivation of $G$ for any closed goal $G$ and
closed program $\st{P}$.
The soundness and completeness of such an interpreter with respect to
the high-level description of computation in
Chapter~\ref{chp:language} is demonstrated in \cite{N93JAR}.
Notice that our abstract interpreter still has elements of
non-determinism in it. In particular, it has to select the next goal
to try from the collection of goals in the state, it has to make a
choice between the two disjuncts when solving a disjunctive goal and
it also needs to pick the program clause to try from the elaboration
of the program when it reaches an atomic goal.
These issues are present in the setting of a first-order logic
language as well and similar  solutions can be used in our context. In
particular, we impose a left to right order on the goal set and use
this order to determine the next goal to act upon, we use a
left-to-right processing  order in the treatment of disjunctive goals
and we select clauses in solving atomic goals based on the order of
their presentations in the program. There is no need to reconsider the
order in which we select goals from the goal set. In all other cases
we use a depth-first approach with the possibility of backtracking
when faced with alternatives.

Definition~\ref{def:interp} requires that the final disagreement set
consist of only flexible-flexible pairs. The ability to produce a set
satisfying this requirement depends on the unification procedure that
is used. The unification algorithm that we
discuss next is guaranteed to produce an empty disagreement set when
the unification problems that have to be solved all fall within the
higher-order pattern fragment. However, we shall sometimes apply this
procedure to cases where this restriction is not satisfied. In this
case, it is possible that the final disagreement set is not empty and
contains at least one rigid-flexible pair. In this case the original
goal is to be understood to be solvable provided the final
disagreement set has a solution.

\section{Higher-Order Pattern Unification}\label{sec:pattern_unif}
The implementation scheme underlying this thesis specializes the
abstract interpreter that we have described by using a unification
algorithm that completely solves higher-order pattern unification
problems. We say that a unification problem, given by a
disagreement set, is in this class if the following syntactic
constraint is satisfied by every term in the set: for any subterm of the
term that has the the form $(X\ t_1\ ...\ t_n)$ where $X$ is a logic
variable, it must be the case that $t_1$, ..., $t_n$  are distinct
constants or de Bruijn indexes and, further, if they are constants
then they must have originated from the processing of (essential)
universal quantifiers appearing inside the scope of the quantifier
whose processing gave rise to $X$. Given the labeling function
discussed in the previous section, the latter condition can be stated
also in the following way: if $t_i$ is a constant then it must be the
case that $\st{L}(X) < \st{L}(c)$, where $\st{L}$ is the labeling
function associated with the state in which the unification problem is
encountered. As a concrete example, consider the disagreement set
$\{\dg{(X\ c_2)}{(c_1\ c_2)}\}$, where $X$ is a logic
variable and $c_1$ and $c_2$ are constants. If the labeling function
associated with the state is $\{\dg{X}{1}, \dg{c_1}{1},
\dg{c_2}{2}\}$ then this disagreement set constitutes a higher-order
pattern unification problem. However, it is not a higher-order pattern
unification problem if the labeling function is $\{\dg{X}{2},
\dg{c_1}{1}, \dg{c_2}{2}\}$ instead. It is not too difficult to see
that with scopes corresponding to the second labeling function the
problem has two solutions: $\{\dg{X}{\lambdadb(c_1\ \# 1)}\}$
and $\{\dg{X}{\lambdadb(c_1\ c_2)}\}$. The scoping corresponding to
the first labeling function rules out the second of these
unifiers. More generally, it has been observed that unification
problems that are in the higher-order pattern class have unique
most general solutions whenever they are solvable~\cite{Miller91jlc}.

The unification procedure that we will use has two phases, one for
term simplification and another for variable binding. In the first
phase, rigid-rigid pairs are handled in a way similar to that in
Huet's unification algorithm by matching
the heads of terms and progressing into subproblems formed by the
arguments pairwise when the head match succeeds. In the binding phase,
rigid-flexible (symmetrically, flexible-rigid) and flexible-flexible
pairs are examined and a substitution is generated for the variable
head(s) only if the flexible term(s) satisfy the higher-order pattern
restriction. The transformation of the terms to their head normal
forms is assumed implicitly prior to the application of either of
these phases.
We also assume a slight modification of the term representation that
collapses a sequence of abstractions into a consolidated form: in
particular the term $(\lambdadb\ldots\lambdadb\ t)$ with a binder
length $n$ in our previous discussions is now represented as
$(\lambdadb(n, t))$. By an abuse of notation, we shall allow
the binder length to be equal to $0$, viewing $(\lambdadb(0, t))$ as
identical to $t$.

The binding phase of our algorithm utilizes some optimizations over
Huet's procedure that become possible when we restrict attention to
the higher-order pattern case. To understand one of these
optimizations, consider a rigid-flexible disagreement pair of form
\begin{tabbing}
\dquad\dquad\dquad$\dg{(X\ a_1\ \ldots\ a_n)}{(r\ s_1\ \ldots\ s_m)}$,
\end{tabbing}
where $X$ is a logic variable and $r$ is a constant or de Bruijn index
and assume that the higher-order pattern requirements are
satisfied. We do not show binders at the heads of the terms in a
disagreement pair here or below because these can be made identical
and, under the de Bruijn representation, they can then be
ignored. Now, a solution to this pair must rely on a substitution for
$X$. Suppose that the term so substituted has the structure
$\lambdadb(n,\ r'\ t_1\ \ldots\ t_n)$.  If $r$ is a de Bruijn index or
a constant $c$ such that $\st{L}(X)<\st{L}(c)$ where $\st{L}$ is the
relevant labeling function, $r$ cannot appear directly in a
substitution for $X$ that is proper with respect to $\st{L}$.
Consequently, the only way the pair can be solved is if $r$ appears in
the list of arguments for $X$, \ie, in $a_1\ \ldots\ a_n$ and, in this
case we would need to substitute for $X$ a term that projects onto the
corresponding argument.  The other possibility is for $r$ to be a
constant $c$ such that $\st{L}(c)\leq\st{L}(X)$. In this case, $c$
cannot occur in the list $a_1$, ..., $a_n$, and for this reason, $r'$
would have to be identical to $c$. These observations allow us to
uniquely determine the head of the substitution to be generated and to
thereby avoid any of the branching that would be manifest in an
application of Huet's algorithm that is blind to the situation being
considered.

Another place where an optimization is possible is in the treatment of
flexible-flexible pairs. Huet's algorithm does not treat such pairs at
all, as we have noted earlier. However, if the higher-order pattern
restriction is adhered to then it is possible to solve such pairs
in a most general way. For example suppose that the pair under
consideration is of form
\begin{tabbing}
\dquad\dquad\dquad$\dg{(X\ a_1\ \ldots\ a_n)}{(Y\ b_1\ \ldots\ b_m)}$,
\end{tabbing}
where $X$ and $Y$ are distinct logic variables. Let us first assume
that the quantifiers from which $X$ and $Y$ result have (effectively)
the same scopes, \ie, that $\st{L}(X)=\st{L}(Y)$. Then it can be
seen that a most general solution to this pair can be given by
substitutions for $X$ and $Y$ of the form
\begin{tabbing}
\dquad$\dg{X}{\lambdadb(n,\ (H\ t_1\ \ldots\ t_k))}$\quad
and\quad $\dg{Y}{\lambdadb(m,\ (H\ s_1\ \ldots\ s_k))}$
\end{tabbing}
where $H$ is a new logic variable with the same scope as that of
$X$ and $Y$ and $t_1,\ldots,t_k$ and $s_1,\ldots,s_k$ are de Bruijn
indices for variables bound by the abstractions in the binder of the
substitution terms. The purpose of the arguments in the two
substitutions is to preserve parts of the arguments in the terms in
the disagreement pairs that cannot be absorbed into any subsequent
substitutions for $H$. Of course, when this substitution is applied to
the terms that are to be unified, it should produce identical terms.
From this, it is easy to see that $t_1,\ldots,t_k$ and
$s_1,\ldots,s_k$ should be such that they both generate the same
permutation $z_1,\ldots,z_k$ of the common elements of the argument
lists $a_1\ldots a_n$ and $b_1\ldots b_m$ of the terms in the
disagreement pair.

The notation introduced in the following definition is useful in
making the substitutions described above precise.
\begin{defn}
Let $[a_1, \ldots, a_n]$ be a non-empty list of distinct constants
or de Bruijn indexes, and let $z$ be a constant or de Bruijn index
occurring in $[a_1\ \ldots\ a_n]$. Then $z\downarrow[a_1, \ldots, a_n]$
denotes the de Bruijn index $\#(n+i-1)$ where $i$ is such that $z = a_i$.
Suppose that $[a_1\ \ldots\ a_n]$ and $[z_1,\ldots,z_k]$ are two
lists of distinct de Bruijn indices or constants such that
$\{z_1,\ldots,z_k\}\subseteq\{a_1,\ldots,a_n\}$,
then $[z_1, \ldots, z_k] \downarrow [a_1, \ldots, a_n]$ denotes the
list $[i_1, \ldots, i_k]$ such that
for $1 \leq j \leq k$, $i_j = z_j\downarrow [a_1, \ldots, a_n]$. We
include the case where $k = 0$ in this definition by deeming the
result to be the empty list.
\end{defn}
Using the selection operator, we can define a most general unifier for
the pair of (higher-order pattern) terms
$\dg{(X\ a_1\ \ldots\ a_n)}{(Y\ b_1\ \ldots\ b_m)}$ where $X$ and $Y$
are logic variables such that $\st{L}(X) = \st{L}(Y)$ as
\begin{tabbing}
\dquad$\{\dg{X}{\lambdadb(n, (H\ t_1\ \ldots\ t_k))}$,
$\dg{Y}{\lambdadb(m,(H\ s_1\ \ldots\ s_k))}\}$,
\end{tabbing}
where
\begin{enumerate}
\item $H$ is a new logic variable,
\item $[z_1, \ldots, z_k]$ is some listing of the elements of
  $\{a_1,\ldots,a_n\} \cap \{b_1,\ldots,b_m\}$, and
\item $[t_1, \ldots, t_k]=[a_1, \ldots, a_n]\!\downarrow\![z_1, \ldots, z_k]$ and
$[s_1, \ldots, s_k]=[b_1, \ldots, b_n]\!\downarrow\![z_1, \ldots, z_k]$.
\end{enumerate}
As a concrete example, suppose the terms to be unified are
\begin{tabbing}
\dquad\dquad\dquad$(X\ c_4\ c_1\ c_2\ c_3)$\dquad
and \dquad$(Y\ c_5\ c_2\ c_1\ c_3)$,
\end{tabbing}
where $X$ and $Y$ are logic variables such that $\st{L}(X)=\st{L}(Y)=0$,
and $c_i$'s are constants where $\st{L}(c_i) = i$, for $1\leq i\leq 5$.
This pair has a most general unifier
\begin{tabbing}
\dquad\dquad$\{\dg{X}{\lambdadb(4, \ (H\ \#3\ \#2\ \#1))},
\dg{Y}{\lambdadb(4,\ (H\ \#2\ \#3\ \#1))}\}$.
\end{tabbing}
The listing of the common argument elements in the terms to be unified
that produces the sequence of argument elements in the substitution
terms is $[c_1,c_2,c_3]$.

Of course, the labels on the flexible heads of the terms that are to
be unified need not be the same. Let us assume, without losing
generality, that $\st{L}(X)<\st{L}(Y)$. We can describe a most general
unifier in this case as well. This solution can be arrived at in two
steps. The first step, that is called {\it raising}, adjusts the head
of the second term so that its scope is made identical to that of
$X$. At this point, a unifier can be generated as in the case already
considered. The main issue with the label of $Y$ being larger than
that of $X$ is that some of the constants that appear as arguments in
the first term can appear in the substitution term for $Y$. These
constants are the ones that have a label that is less than or equal to
that of $Y$. We introduce the following notation to identify them
collectively:
\begin{defn}
Given a list of distinct constants and de Bruijn indexes
$[a_1, \ldots, a_n]$, a labeling function $\st{L}$
and a logic variable $Y$, let $\{c_1,\ldots,c_k\}$ be the set of
constants in $[a_1,\ldots,a_n]$ whose labels are less than or equal to
$\st{L}(Y)$. Then the expression $[a_1, \ldots, a_n]\Uparrow Y$
denotes some listing of $\{c_1,\ldots,c_k\}$.
Note that the set of constants satisfying the condition may be empty
in which case $[a_1, \ldots, a_n]\Uparrow Y$ is an empty list.
\end{defn}
The raising substitution is identified in this context to be
$\{\langle Y, Y'\app c_1\app \ldots\app c_k\rangle\}$ where $Y'$ is a
new logic variable that is assigned the same label as $X$ and
$[c_1,\ldots,c_k] = [a_1,\ldots,a_n] \Uparrow Y$.

To complete our consideration of the flexible-flexible case, we need
also to deal with the situation where the heads of the two terms are
identical, \ie, where the pair in question is
$\dg{(X\ a_1\ \ldots\ a_n)}{(X\ b_1\ \ldots\ b_n)}$. This differs from
the earlier case in that the {\it same} substitution gets applied to
both terms. From this it follows easily that a most general solution
is one that preserves exactly the common elements of $[a_1\ldots,a_n]$
and $[b_1,\ldots,b_n]$ that also appear in identical positions in the
two lists.

The above discussion provides an overview of the higher-order
pattern unification algorithm that is used as the basis of the
implementation scheme developed in this thesis. The actual algorithm
we use is the one developed by Nadathur and
Linnell~\cite{NL05HOP}. This algorithm uses the fact that the partial
substitutions described above are actually most general to generate a
complete solution for a flexible-rigid pair in one recursive pass over
the rigid term. (The flexible-flexible case is completely treated
already by the substitution discussed.) As is to be anticipated, this
algorithm has two phases, one for term simplification and the other for
binding. The simplification phase is characterized by the rules in
Figure~\ref{fig:hopu_simplification}. In the application of these rules,
a unification problem is assumed to be given by a tuple
$\dg{\st{D}}{\theta}$ where $\st{D}$ is the disagreement set under
consideration and $\theta$ is a set of substitutions which is initially empty.
Further, a labeling function $\st{L}$ is assumed to be available during
the entire unification process as an implicit global component of the
state.
\begin{figure}
\begin{tabbing}
\dquad\=(5) \=$\dg{\dg{\lambdadb(n, t)}{\lambdadb(m, s)}::\st{D}}{\theta}$ \= $\lra$ \=\kill
\> (1)\>$\dg{\dg{\lambdadb(n, t)}{\lambdadb(n, s)}::\st{D}}{\theta}$\> $\lra$
\> $\dg{\dg{t}{s}::\st{D}}{\theta}$, provided $n > 0$. \\
\> (2)\>$\dg{\dg{\lambdadb(n, t)}{\lambdadb(m, s)}::\st{D}}{\theta}$\> $\lra$
\> $\dg{\dg{t}{\lambdadb(m-n, s)}::\st{D}}{\theta}$, \\
\>\> provided $n > 0$ and $m > n$. \\
\> (3)\>$\dg{\dg{(r\ t_1\ \ldots\ t_n)}{\lambdadb(m, s)}::\st{D}}{\theta}$ $\lra$ \\
\>\> $\dg{\dg{((\env{r,0,m,nil}\ \env{t_1,0,m,nil}\ \ldots\ \env{t_n,0,m,nil})\ \#m\ \ldots\ \#1)}{s}::\st{D}}{\theta}$, \\
\>\> provided $r$ is a constant or a de Bruijn index and $m>0$. \\
\> (4)\>$\dg{\dg{(r\ t_1\ \ldots\ t_n)}{(r\ s_1\ \ldots\ s_n)}::\st{D}}{\theta}$ $\lra$
$\dg{\dg{t_1}{s_1}::\ldots ::\dg{t_n}{s_n}::\st{D}}{\theta}$, \\
\>\> provided $r$ is a constant or a de Bruijn index.\\
\> (5)\>$\dg{\dg{(X\ a_1\ \ldots\ a_n)}{t}::\st{D}}{\theta}$ \= $\lra$
$\dg{\sigma(\st{D})}{\sigma\circ\theta}$, \\
\>\> provided $X$ is a logic variable, $(X\ a_1\ \ldots\ a_n)$ is $\Ll$ with respect to $\st{L}$, \\
\>\> and {\it mksubst}$(X,t,[a_1, \ldots, a_n])\lra\sigma$. 
\end{tabbing}
\caption{Term simplification in higher-order pattern unification.}\label{fig:hopu_simplification}
\end{figure}
The binding phase is realized through the function {\it mksubst} that
takes as its arguments the head of (one of the) flexible term(s), the
arguments of this term and the other term in the disagreement pair.
The definition of this function together with those of other two auxiliary
ones {\it bnd} and {\it foldbnd} are given by the rules in
Figures~\ref{fig:hopu_mksubst}, \ref{fig:hopu_bnd} and \ref{fig:hopu_foldbnd}.

\begin{figure}
\begin{tabbing}
\quad\={\it mksubst}$(X, \lambdadb(k, X\ b_1\ \ldots\ b_m), [a_1\ \ldots\ a_n])$ \= $\lra$\kill
\>{\it mksubst}$(X, \lambdadb(k, X\ b_1\ \ldots\ b_m), [a_1, \ldots, a_n])$ \>$\lra$ $\{\dg{X}{\lambdadb(k+n, H\ w_1\ \ldots\ w_l)}\}$, \\
\>\quad \=where $H$ is a new logic variable and $\st{L} = \st{L}\cup\{\dg{H}{\st{L}(X)}\}$, provided \\
\>      \>(1) \=$(X\ b_1\ \ldots\ b_m)$ is $\Ll$ with respect to $\st{L}$ and \\
\>      \>(2) \>for $1\leq i\leq n+k$, $w_i = \#(n+k-i)$, if $al[i]=b_i$ where \\
\>      \>    \>$al = [\env{a_1,0,k,nil}, \ldots, \env{a_n,0,k,nil}, \#k, \ldots, \#1]$. \\
\>{\it mksubst}$(X, t, [a_1, \ldots, a_n])$ $\lra \{[X:=\lambdadb(n, s)]\}\circ\theta$, \\
\>\quad if the head of $t$ is not $X$ and {\it bnd}$(X, t, [a_1\ \ldots\ a_n], 0)\lras\dg{\theta}{s}$
\end{tabbing}
\caption{Top-level control for calculating variable bindings}\label{fig:hopu_mksubst}
\end{figure}

\begin{figure}\small
\begin{tabbing}
\={\it bnd}$(X, \lambdadb(m, t), [a_1, \ldots, a_n], l)\lra\dg{\theta}{\lambdadb(m, s)}$, \\
\>\quad if $m > 0$ and {\it bnd}$(X, t, [a_1, \ldots, a_n], l+m)\lras\dg{\theta}{s}$. \\

\>{\it bnd}$(X, r\ t_1\ \ldots\ t_m,[a_1, \ldots, a_n], l)\lra\dg{\theta}{r'\ s_1\ \ldots\ s_m}$ \\
\>\quad provided {\it foldbnd}$(X, [t_1, \ldots, t_m], [a_1, \ldots, a_n], l, \dg{\emptyset}{[]})\lras\dg{\theta}{[s_m, \ldots, s_1]}$ and \\
\>\quad \=(1) \=$r' = r$, if $r$ is a constant such that $\st{L}(r)\leq\st{L}(X)$, or \\
\>      \>(2) \>$r' = r\downarrow al$, if $r$ is a de Bruijn index occurring in $al$ where \\
\>      \>    \>$al = [\env{a_1,0,l,nil}, \ldots, \env{a_n,0,l,nil}, \#l, \ldots, \#1]$. \\

\>{\it bnd}$(X, Y\ b_1\ \ldots\ b_m, [a_1, \ldots, a_n], l)\lra$ \\
\>\squad\dquad$\dg{\{\dg{Y}{\lambdadb(m, H\ c_1\ \ldots\ c_k\ u_1\ \ldots\ u_q)}\}}{H\ w_1\ \ldots\ w_p\ v_1\ \ldots\ v_q }$, \\
\>\quad \= where $H$ is a new logic variable and $\st{L} = \st{L}\cup\{\dg{H}{\st{L}(X)}\}$,\\
\>      \> $[c_1,\ldots,c_k] = al\Uparrow Y$, $[w_1,\ldots w_p] =
        [c_1,\ldots,c_k]\downarrow al$, $[u_1,\ldots,u_q] =
        zl\downarrow [b_1,\ldots, b_m]$ and\\
\>\>$[v_1,\ldots,v_q] = zl\downarrow al$, with $al = [\env{a_1,0,l,nil},
  \ldots, \env{a_n,0,l,nil}, \#l, \ldots, \#1]$ and\\
\>      \> $zl = [z_1,\ldots,z_q]$ as a permutation of\\
\>      \> $\{\env{a_1,0,l,nil}, \ldots, \env{a_n,0,l,nil}, \#(l-1), \ldots, \#1\}\cap\{b_1,\ldots, b_m\}$, \\
\>      \> provided $X$ and $Y$ are distinct logic variables such that $\st{L}(X)<\st{L}(Y)$, \\
\>      \> and $Y\ b_1\ \ldots\ b_m$ is $\Ll$ with respect to $\st{L}$. \\

\>{\it bnd}$(X, Y\ b_1\ \ldots\ b_m, [a_1, \ldots, a_n], l)\lra$ \\
\>\squad\dquad$\dg{\{\dg{Y}{\lambdadb(m, H\ w_1\ \ldots\ w_p\ v_1\ \ldots\ v_q)}\}}{H\ c_1\ \ldots\ c_k\ u_1\ \ldots\ u_q }$, \\
\>\quad \= where $H$ is a new logic variable and $\st{L} = \st{L}\cup\{\dg{H}{\st{L}(X)}\}$,\\
\>      \> $[c_1,\ldots,c_k] = bl\Uparrow X$, $[w_1,\ldots w_p] = [c_1,\ldots,c_k]\downarrow bl$,
           $[v_1,\ldots,v_q] = zl\downarrow bl$ and \\
\>      \> $[u_1,\ldots,u_q] = zl\downarrow [\env{a_1,0,l,nil}, \ldots, \env{a_n,0,l,nil}, \#l, \ldots, \#1]$ with \\
\>      \> $bl = [b1,\ldots,b_m]$ and $zl = [z_1,\ldots,z_q]$ as a permutation of \\
\>      \>     $\{\env{a_1,0,l,nil}, \ldots, \env{a_n,0,l,nil}, \#(l-1), \ldots, \#1\}\cap\{b_1,\ldots, b_m\}$, \\
\>      \> provided $X$ and $Y$ are distinct logic variables such that $\st{L}(Y)\leq\st{L}(X)$, \\
\>      \> and $Y\ b_1\ \ldots\ b_m$ is $\Ll$ with respect to
$\st{L}$.
\end{tabbing}
\vspace{-0.5cm}
\caption{Calculating variable bindings.}\label{fig:hopu_bnd}
\end{figure}

\begin{figure}
\begin{tabbing}
\quad\=\qquad\qquad\qquad\=\kill
\>{\it foldbnd}$(X, [], al, l, \dg{\theta}{sl})\lra\dg{\theta}{sl}$.\\
\>{\it foldbnd}$(X, [t1,\ldots,t_n], al, l, \dg{\theta}{sl})$\\
\>\>$\lra$
  {\it foldbnd}$(X, [\sigma(t_2),\ldots,\sigma(t_n)],
al, l, \dg{\sigma\circ\theta}{s::sl}),$ \\
\>\quad provided $n > 0$ and {\it bnd}$(X, t_1, al, l)\lras\dg{\sigma}{s}$.
\end{tabbing}
\caption{Iterating the variable binding calculation over an argument list}\label{fig:hopu_foldbnd}
\end{figure}

The pattern unification procedure terminates when none of the
transformation rules can be applied to the disagreement set that has
been produced. If this is because the disagreement set is empty, then
a most general unifier has been computed for the original
problem. On the other hand, if the disagreement set is not empty then
non-unifiability can be concluded in a context where all disagreement
pairs adhere to the higher-order pattern restriction. Such failures
are characterized concretely by the following situations:
\begin{enumerate}
\item there are rigid-rigid pairs left of the form
$\dg{r\ t_1\ \ldots\ t_n}{r'\ s_1\ \ldots\ s_m}$ where $r \neq r'$.
\item the attempt to apply a {\it bnd} rule encounters a tuple of the
  form
\begin{tabbing}
\dquad$(X, r\ t_1\ \ldots\ t_m, [a_1\ \ldots\ a_n], l)$
\end{tabbing}
where $r$ is a de Bruijn index or a constant with $\st{L}(r)>\st{L}(X)$
that does not occur in
$[\env{a_1,0,l,nil}, \ldots, \env{a_n,0,l,nil}, \#l, \ldots, \#1]$.
\item the attempt to apply a {\it bnd} rule encounters a tuple of the
  form
\begin{tabbing}
\dquad$(X, X\ b_1\ \ldots\ b_m, [a_1, \ldots, a_n], l)$.
\end{tabbing}
\end{enumerate}
In analogy with first-order unification, the first of these failures
corresponds to a clash of constants and the latter two constitute
failure because of an ``occurs-check.''

As we have already noted, our implementation of $\lp$ allow for a more
liberal syntax that could lead to disagreement pairs that do not
satisfy the higher-order pattern restriction. Given this, a non-empty
disagreement set may also be a signal of the fact that the unification
process should be suspended till further variable bindings have been
determined in some other way. The actual realization of the
unification procedure should therefore be on the lookout for errant
disagreement pairs and should defer the processing of these to a later
point; the structure of the abstract interpreter already accommodates
such a possibility.

The last issue to be mentioned in this section is the usage
of types in the pattern unification procedure.
Unlike Huet's algorithm, types are not needed during the binding phase
of unification, but have a relevance to the applicability
of rule (4) in term simplification. In particular, two constants with
the same name are viewed as being equal only when they have equal
types and this forces the terms in this situation to have the
same number of arguments.
Based on the type system of our language, the determination of the equality
of types is carried out by first-order unification,
which should interleave with the application of rule (4).
The details of the treatment of types relative to pattern
unification are discussed in Chapter~\ref{chp:types}.

\chapter{Machine-Level Term Representation}\label{chp:machineTermRep}
The discussions in the previous chapters have gradually refined the
representation of $\lambda$-terms from a conceptual form into
one more suitable to be used as a basis of implementation.
Two issues still remain to be dealt with in a concrete
implementation. First, we need an actual procedure for converting
terms to head normal form. Second, we still have to discuss the
reflection of the suspension calculus into lower-level machine
structures. We take these issues up in this chapter.
In Section~\ref{sec:reduction_strategies} we discuss
different strategies for producing a head normal form for terms in the
suspension calculus, leading eventually to one that has been shown
empirically to have good time and space
characteristics. Section~\ref{sec:internal_encoding} then describes
the
low-level encoding of $\lambda$-terms used in our implementation
and it also discusses the pragmatic issues underlying our choices.

\section{Implementation of Head Normalization}\label{sec:reduction_strategies}
An efficient implementation of
the normalization of terms is clearly important to the
performance of an overall system realizing the $\lp$ language.
The suspension calculus
serves as a suitable basis for such an implementation by providing a
control over the substitution operation and, hence, a
flexibility in the ordering of the steps involved in reduction.
A high-level, non-deterministic description of the process of
reducing a term to one of its head normal forms has also been
identified in Chapter~\ref{chp:termRep}; this is a process in which
a head normal form is produced by repeatedly rewriting
a head redex. Once a head normal form has been produced, there is
still some flexibility with regard to how to treatment the arguments of the
term. For example,
consider the term
\begin{tabbing}
\dquad\dquad$(\lambdadb(n, \env{(c\ t_1\ldots\ t_m),ol,nl,e}))$,
\end{tabbing}
where $c$ is a constant and $t_1$, ..., $t_m$ are arbitrary
terms. The applications of rule $(r5)$ and $(r1)$ in
Figure~\ref{fig:susp_rules} results in the structure
\begin{tabbing}
\dquad\dquad$(\lambdadb(n, (c\ \env{t_1, ol, nl, e}\ \ldots\ \env{t_m, ol, nl, e})))$,
\end{tabbing}
that is a head normal form.
At this point, there are choices in what to do with the arguments,
whether to leave them as suspensions, or to transform them into de
Bruijn terms or perhaps even to reduce them too to normal
forms. Different reduction strategies can be characterized in terms of
the choices that they make at this stage.

One reduction strategy that can be considered is that which uses the
suspension calculus only as an implementation device, keeping explicit
representations only of terms in de Bruijn form. Within this strategy,
the old and new embedding levels and the environment in a suspension
would be reflected in the parameters of the reduction procedure but
not in terms.
Consequently, the substitutions remaining on the arguments of the head
normal form shown above would have to be carried out eagerly, possibly
combined with additional $\beta$-reductions applied to these terms.
In this strategy, the rewriting steps shown in
Figure~\ref{fig:susp_rules} and Figure~\ref{fig:susp_r9} would be
carried out implicitly and hence would not themselves give rise to
intermediate terms.
An alternative strategy would be one that dispenses with the recursive
structure of the first reduction procedure by actually explicitly
creating the righthand sides of each of the rewrite rules and by using
a stack to provide any additional control. Such a procedure would have
to be complemented by an explicit representation of suspensions and
hence it could also potentially leave the arguments of a head normal
form as suspensions.

A drawback with the second approach is that it requires new terms
to be explicitly created as the result of each rewriting step, even if the
terms only serve as intermediate results of the head normalization process.
For instance, consider the original term in the previous example.
As the result of the applications of the rule $(r5)$ in
Figure~\ref{fig:susp_rules}, this approach requires the explicit creation
of the structure
\begin{tabbing}
\dquad\dquad$(\lambdadb(n, (\env{c, ol, nl, e}\ \env{t_1, ol, nl, e}\ \ldots\ \env{t_m, ol, nl, e})))$,
\end{tabbing}
only to see the head $\env{c, ol, nl, e}$ being rewritten by the
immediately following step through an application of rule
$(r1)$. As another example, it is possible for the term $t$ in the
suspension  $\env{t, ol, nl, e}$ to be a $\beta$-redex, in which case
new suspensions will be created through the use of rules (r5) and (r6)
only to be discarded when the rule ($\beta'_s$) is applied.

The redundancy in the creation of the intermediate terms is
avoided by the first strategy.
However, the eager performance of the substitutions over the arguments of
head normal forms leads to a traversal of these arguments, which may
turn out to be redundant in a context where term comparison can be
interleaved with reduction steps; just exposing the heads may suffice
to show non-unifiability. Performing just substitutions also misses
out on the sharing of walks between different reduction steps.
For example, consider the pair of terms
\begin{tabbing}
\dquad\dquad$\dg{(c\ \env{(\lambdadb t_1)\ t_2, 1, 0, (c,0)::nil})}{(c\ t_3)}$,
\end{tabbing}
where $c$ is a constant and $t_1$, $t_2$ and $t_3$ are arbitrary terms.
The application of the term simplification rule (4) in
Figure~\ref{fig:hopu_simplification} results a new pair of terms
formed by the arguments of the original terms which need to be
head normalized immediately. If the transformation of
$\env{(\lambdadb t_1)\ t_2, 0, 1, (c,0)::nil}$ into a de Bruijn term
is carried out eagerly at the end of the previous invocation of head
normalization, as being
required by the eager substitution strategy, a separate traversal has
to be carried out over the structure of $t_1$ when the redex
$(\lambdadb t_1)\ t_2$ is rewritten. Of course, we could also reduce
such redexes when calculating out suspension terms. However, this
corresponds to always producing $\beta$-normal forms, something that
is costly especially in a setting where failure can be registered by
looking at only parts of terms.

The above discussion of the characteristics of the two strategies that
we have considered suggests an intermediate version that combines the
benefits of both: the normalization procedure can use suspensions
implicitly, embedding their components in its arguments rather than in
explicitly constructed suspensions but, in the end leaving the
arguments in the head normal forms it finds in the form of
suspensions. Studies have been conducted in~\cite{LNX03} using the
{\it Teyjus Version 1} implementation of $\lp$ to understand the
performance differences between the combination reduction strategy
just described and the first strategy considered which evaluates
substitutions eagerly on the arguments of head normal forms.  These
studies indicate a significant performance benefit to the combination
strategy: specifically, an average of $32\%$ reduction in execution
time and $81\%$ reduction in memory usage was observed over a set of
practical $\Ll$-style programs with this strategy. We have accordingly
chosen to use this combination strategy in the new implementation of
$\lp$. In the rest of this section, we elaborate on the structure of the
reduction procedure used in the implementation. To keep this
description brief and understandable, we present this structure
through SML style pseudo-code.

\begin{figure}[t]
\begin{verbatim}
datatype rawterm = const of string
  | lv of string
  | db of int
  | ptr of (rawterm ref)
  | lam of (rawterm ref)
  | app of (rawterm ref) * (rawterm ref)
  | susp of (rawterm ref)*int*int*(eitem list)
and eitem = dum of int
  | bndg of (rawterm ref) * int

type env = (eitem list)
type term = (rawterm ref)
\end{verbatim}
\caption{A SML encoding of suspension terms in head normalization.}\label{fig:hnorm_datatype}
\end{figure}

The first task in presenting the procedure is to provide datatype
declarations for the terms in the suspension calculus. These
declarations are contained by Figure~\ref{fig:hnorm_datatype}.
As in usual implementations, a graph-based representation is assumed for
terms. SML expressions of types {\it rawterm} and {\it term} can be
viewed as directed graph, which are assumed to be acyclic
during the reduction process. It in fact can be observed
from the head normalization procedure discussed subsequently
that if the input to the procedure has this property, it is
preserved in the normalization process.

Terms of the suspension calculus are realized as
references to appropriate SML expressions of the type {\it
rawterm}. The environments and environment items in this
calculus are presented as expressions of types {\it env}
and {\it eitem}. An expression of form {\it dum(l)} is used
to encode environment item $@l$, whereas {\it bndg(t, l)}
corresponds to $(t, l)$.
Value constructors {\it fv} and {\it db} are used
to encode logic variables and de Bruijn indexes respectively.
The encoding of abstractions, applications and suspensions
is achieved by supplying constructors {\it lam}, {\it app}
and {\it susp} to arguments of proper types.
The constructor {\it ptr} serves to aid the sharing of reduction results
which means that at certain points in our reduction process, we want to
identify (the representations of) terms in a way that makes the
subsequent rewriting of one of them correspond to the rewriting
of the others. Such an identification is usually realized by
representing both expressions as pointers to a common location
whose contents can be changed to effect shared rewritings. In SML
it is possible to update only references and so the common
location itself must be a pointer. The constructor {\it ptr} is
used to encode indirections of this kind when they are needed.
Complementing this encoding, we use the following functions to,
respectively, dereference a term and assign a new value to a given term.
\begin{verbatim}
fun deref(term as ref(ptr(t))) = deref(t)
  | deref(term) = term
fun assign(t1,ref(ptr(t))) =  assign(t1,t)
  | assign(t1,t2) = t1 := ptr(t2)
\end{verbatim}
In addition, we use the following function to help with looking for
a value in an environment during the reduction process.
\begin{verbatim}
fun nth(x::l,1) = x
  | nth(x::l,n) = nth(l,n-1)
\end{verbatim}

Based on the given SML encoding of the terms suspension calculus,
the main work of the head normalization procedure can be defined as that in
Figure~\ref{fig:hnorm_procedure}. The first four arguments
are used to represent a (possibly trivial) suspension implicitly.
The fifth argument of boolean type is used to control that the rewriting of
head redexes is performed in a left-most and outer-most order: it is set to
true when the term under reduction has been found as the function part of
an application at the outside, and the normalization process stops
rewriting the redexes contained by it once an abstraction structure
is revealed, so that the outer redex can be rewritten first.
(There is in fact one exception to the outer-most order
of rewriting in the presence of nested suspensions, which will
be explained shortly.)
The application of the $\beta_s$ and $\beta_s'$
rules in Figure~\ref{fig:susp_rules} is carried out in the
application case of {\it hnorm}. Further, when the head of a head
normal form is exposed and the head normal form still has an application
structure, the implicitly recorded non-trivial suspensions
over the arguments are explicitly reflected into the term structure.
The value returned by {\it hnorm} is a quadruple
that can be interpreted as an implicit suspension.
In reality, this suspension is a trivial one in all cases
other than when the call to {\it hnorm} has its fifth
argument being set to true, and the term component in
the resulting suspension is an abstraction.

\begin{figure}\footnotesize
\begin{verbatim}
fun hnorm(term as ref(db(i)),0,0,[],_) = (term,0,0,[])
 |  hnorm(term as ref(db(i)),ol,nl,e,whnf) =
      if (i > ol) then (ref(db(i+ol-nl)),0,0,nil)
      else (fn dum(l)   =>(ref(db(nl-l)),0,0,nil)
             | bndg(t,l)=>(fn ref(susp(t2,o,n,e)) => hnorm(t2,o,n+nl-l,e,whnf)
                            | t => hnorm(t,0,nl-l,[],w)) (deref(t))) (nth(env,i))
 |  hnorm(term as ref(lam(t)),ol,nl,e,true) = (term,ol,nl,env)
 |  hnorm(term as ref(lam(t)),ol,nl,e,false) =
      let val (t',ol',nl',e')=if (ol=0) andalso (nl=0) then hnorm(t,0,0,[],false)
                              else hnorm(t,ol+1,nl+1,dum(nl)::e,false)
      in (ref(lam(t')), ol', nl', e') end
 |  hnorm(term as ref(app(t1,t2)),ol,nl,e,whnf) =
      let val (f,fol,fnl,fe) = hnorm(t1,ol,nl,e,true)
      in (fn ref(lam(t))=>
           let val t2' = if ((ol=0) andalso (nl=0)) then t2
                         else ref(susp(t2,ol,nl,env))
               val (t',ol',nl',e') = hnorm(t,fol+1,fnl,bndg(t2',fnl)::fe,whnf)
           in ((if (ol<>0) orelse (nl<>0) orelse (ol'<>0) orelse (nl'<>0) then ()
                else assign(term, t'))); s end
          | t => if ((ol = 0) andalso (nl = 0))
                 then (assign(term, ref(app(f, f2))); (term,0,0,nil))
                 else (ref(app(f,ref(susp(t2,ol,nl,e)))),0,0,nil)) (deref f) end
 |  hnorm(term as ref(susp(t,ol,nl,e)),ol',nl',e',whnf) =
      let val s = mk_explicit(hnorm(t,ol,nl,env,whnf),ol',nl',e')
      in  (assign(term, s);
           if (ol'=0) andalso (nl'=0) then s
           else hnorm(term,ol',nl',env'))  end
 |  hnorm(ref(ptr(t)),ol,nl,env,whnf) = hnorm(deref(t),ol,nl,env,whnf)
 |  hnorm(term,_,_,_,_) = (term,0,0,nil)
\end{verbatim}
\vspace{-0.5cm}
\caption{An environment based head normalization procedure with lazy substitutions.}\label{fig:hnorm_procedure}
\end{figure}

When the call to {\it hnorm} on the inner suspension has its
fifth argument being set to true, it is possible that the
returned value of the call is a non-trivial suspension with
its term component being an abstraction. This suspension should
be made explicit, and further, it should be transformed into
an abstraction using the reading rule (r6) in
Figure~\ref{fig:susp_rules} before computation can proceed.
The described behavior is carried out by the suspension case
of the procedure {\it hnorm}. The effect of making the
suspension returned by the rewriting of the inner suspension
explicit after applying rule (r6) is accomplished by an invocation
of the auxiliary function
{\it mk\_explicit} defined as the following.
\begin{verbatim}
fun mk_explicit(t, 0, 0, nil) = t
  | mk_explicit(ref(lam(t)), ol, nl, e) =
      ref(lam(ref(susp(t, ol+1, nl+1, dum(nl)::e))))
\end{verbatim}

Any given term t may be transformed into a head normal form by invoking
the interfacing procedure {\it head$\_$norm} that is defined as follows:
\begin{verbatim}
fun head_norm(t) = hnorm(t, 0, 0, nil, false).
\end{verbatim}

The correctness of {\it head\_norm} is the content of the following theorem,
whose proof can be found in~\cite{XQ04Master}.
\begin{theorem} Let {\it t'} be a reference to
the representation of a suspension term $t$ that translates via
the reading rules (r1)-(r8) in Figure~\ref{fig:susp_rules} and
(r9) in Figure~\ref{fig:susp_r9} to a de Bruijn term
with a head normal form.
Then {\it head\_norm(t')} terminates and, when it does, $t'$ is a
reference to the representation of a head normal form
of the original term in the suspension calculus.
\end{theorem}

\section{Representation of Terms}\label{sec:internal_encoding}
We discuss now the scheme for encoding terms that will become the
basis for their manipulation in the abstract machine for $\lp$.

The most natural encoding for a term is one that uses a memory unit
with a tag  indicating the syntactic category of the term with additional
parts for any other components.
These additional components vary according to the specific kinds of term.
For a de Bruijn index, all that is needed is a positive number for the index
itself. As discussed in Chapter~\ref{chp:interpreter}, a labeling function
associating constants and logic variables with their
universe levels is essential to the unification operation.
This information is then succinctly maintained by recording numeric tags
of non-negative integer values along with constants and logic
variables.
In addition to such a label, a reference should also be maintained
with a constant to its descriptor.
An (un-instantiated) logic variable should serve as a place holder occupying
enough space so that the instantiation can be realized by destructively
changing the cell to other sort of terms.
The content of this cell is not important except for its tag and label.
A suspension term $\env{t, ol, nl, e}$ requires the maintenance of
its two embedding levels $ol$ and $nl$, a reference to its term component $t$
and a reference to its environment $e$ which can be represented as a list.
An abstraction cell contains a positive number corresponding to the binder
length and a reference to its body. In this way, nested abstraction
structures can be denoted by a single term.  An alternative in this encoding
is to require each abstraction to be represented separately. However,
considering the term  decomposition requests issued by
the pattern unification algorithm described in
Section~\ref{sec:pattern_unif}, it is
apparent a faster access to the subcomponents of nested abstractions can be
supported by the encoding we have chosen.

The encoding of application terms requires more careful consideration.
Applications in a higher-order setting are best thought of in a
curried fashion, thus making their components (references to) their
function and argument parts,
respectively. However, a curried rendition of applications leads to a
high cost in the most
common form of access to terms needed by unification:
the access to the head of a head normal form with $n$ arguments
requires working through $n$ applications starting from the outermost
one. In addition, it can also be observed that the pattern
unification algorithm discussed in Section~\ref{sec:pattern_unif}
is best supported if the arguments of a flexible term in head normal
form are available
as a vector. The ability to immediately access the heads and argument
vector of an application is also useful when we consider the
compilation of unification. If a curried representation is used,
runtime effort has to be paid to traverse nested applications for the
purpose of exposing their structure in this form before the rest of
the computation can proceed.

A concrete encoding of an application that is reminiscent of their
treatment in conventional logic language implementations is to
use a structure containing three components: a function part,
(a reference to) a vector of arguments and an arity corresponding to
the size
of the vector. Such a representation has especially nice properties in
our setting when
the program at hand is a first-order one. In this case the head normal
form of the term
is already available at compilation time.
With the described representation, the head and the
argument vector information can be simply obtained from the top-level term
structure, which also lets it be determined that reduction is not
necessary. These benefits appear to be important since
efficiency in realizing first-order style computations  is of special
importance  to the overall performance in practical $\lp$
applications~\cite{MP92}. Our low-level representation accordingly
adopts such an encoding for applications.
In the first-order context, term structures can be modified only via
bindings for (first-order) logic variables which cannot appear in the
function position in an application. Thus, applications themselves
have an unchanging structure.  Taking advantage of
this fact, the first-order representation of an application usually
folds the function part and the argument vector into one contiguous
sequence of terms.
This optimization can, however, not be used in our setting where the
heads of applications can also sometimes change. For this reason,
references are maintained in an application referring to the function
part of it and its argument vector respectively.

The only remaining category of terms provided for in our
representation is that of references. References are necessitated by
the fact that we use a graph-based realization of reduction to foster
sharing, as should be clear from the SML rendition  of the procedure
that we have provided in the previous section.
Thus, the destructive update of the term $(\lambdadb{t_1})\ t_2$ can be
effected by changing the application cell representing this term into
a reference to the representation of the term
$\env{t_1,1,0,(t_2,0)::nil}$. Notice that a reference has the smallest
amount of data amongst all the terms---its encoding needs just a
category tag and a pointer to another term---and so it can be used
conveniently in such destructive updates.
Another use for a reference is in recording the binding of a logic
variable. For example, the binding of
a logic variable $X$ to term $t$ can be registered by changing the
cell for $X$ into a reference to the the representation of $t$.

As we have noted in Section~\ref{sec:pattern_unif}, types
may sometimes be needed at run-time for the purpose of determining the
identity of constants. Consequently, we need to have an explicit
encoding for them as well in the representation of constants. In
particular, a constant cell gets an extra component in the form of a
reference to its type. As for the
encoding of types themselves, since the computation on them
is solely first-order unification, it is sufficient to
adopt the conventional first-order style encoding.
For the purpose of minimizing the run-time
cost on the maintenance and manipulation of types, this
approach can be further refined by separating a type into
a fixed part that is available during compilation and
a dynamic part that should be decided at run time. The former
information can be combined with the descriptor of a given
constant, and the association of a constant cell with its type is
then reduced to only the dynamic part.
A detailed discussion on this topic appears in
Section~\ref{sec:compile_type} after we have obtained a
concrete understanding of the run-time type processing
scheme in our implementation.

\chapter{An Abstract Machine and Processing Model}\label{chp:compile}
We are interested in this thesis in a compilation-based model for
realizing $\lp$. One possible target for a compiler that emerges from
our considerations could be the instruction set for standard
hardware. This is, in fact, the usual choice for conventional
languages. However, the distance between typical machine architectures
and the computational model for $\lp$ is too significant to bridge in
one step. Moreover, these differences make it difficult to visualize
and to state precisely the optimizations that can be performed on
particular instruction sequences that a compiler might generate. For
this reason, we introduce an intermediate level ``abstract machine''
for $\lp$. We describe the structure of this abstract machine in this
chapter, also interleaving with this description a presentation of the
process of compiling $\lambda$Prolog programs into instructions for
this machine. Our abstract machine will inherit its basic structure
from the developments related to compiling Prolog programs that have
resulted in the abstract machine designed by Warren for that language
\cite{Warren83WAM}. We will also make use of a previous machine
designed by Nadathur and colleagues for $\lp$ \cite{KNW94cl,
  N03treatment, NJK95lp} that underlies the {\em Version 1} of the {\em Teyjus}
implementation of this language \cite{NM99cade}. However, unlike this
earlier implementation that tackled full higher-order unification
using Huet's procedure, we will exploit the possibility of using the
higher-order pattern unification algorithm described in
Chapter~\ref{chp:interpreter}.  This
choice simplifies the structure of the abstract machine considerably,
leads to optimizations in the treatment of types as we discuss in the
next chapter and also has the potential for impacting the overall
runtime performance on $\lp$ programs.

This chapter is organized as follows. We introduce the basic
processing model underlying the new abstract machine in
Section~\ref{sec:basic_model}; as mentioned already, this model is
based on the Warren Abstract Machine (WAM) for Prolog
\cite{Warren83WAM}, with which we shall assume that the reader to have
some familiarity.
Section~\ref{sec:ho_control} and Section~\ref{sec:ho_unif}
then discuss the details of the enhancements to this model that are
needed for handling the higher-order features of $\lp$.
Specifically, Section~\ref{sec:ho_control} addresses the treatment
of generic and augment goals, and Section~\ref{sec:ho_unif}
discusses how the higher-order pattern unification is embedded
into the overall processing.
A complete example of a compiled $\lp$ program is presented in
Section~\ref{sec:inst_exp}.
Section~\ref{sec:misc} sketches the treatments to flexible and disjunctive
goals. Significant aspects of the treatment of generic and augment
goals and almost all of the treatment of flexible and disjunctive
goals are inherited from the earlier abstract machine for $\lp$ but
their presentation is, nevertheless, needed here for the sake of
completeness.

Our focus in this chapter will be on
the {\em conceptual structure} of the new abstract machine and the
processing model embodied by it.
This design has been realized in an actual implementation of
$\lp$---{\em Version 2} of the {\it Teyjus} system---a presentation
  of which appears in Chapter~\ref{chp:system}.

\section{The Processing Model}\label{sec:basic_model}
The WAM provides a basic framework for compiling the aspects of
control and unification that are part of the computation in {\em
  Prolog}-like languages. These aspects appear also in $\lp$ and so we
can use this structure in our abstract machine as well.
In this context, we note that the compilation of control refers to the
translation of the dynamic analysis of the structure of complex goals
carried out by the abstract interpreter described in
Chapter~\ref{chp:interpreter} into low-level abstract machine
instructions. The compilation of unification, on the other hand,
corresponds to using knowledge of one half of the disagreement
pairs to reduce the amount of work that needs to be done at
runtime. Specifically, this translates into generating instructions for
analyzing the structure of terms that arrive in argument registers
when attempting to match with the head of a clause and for
correspondingly setting up the argument registers when calling
predicates.

The basic WAM model is enhanced in our implementation in order to
support the richer set of features present in $\lp$.
First, our compilation treatment of control computations should include that
of generic and augment goals in addition to the set of goal structures
contained in the Horn clauses that underlie {\em Prolog}.
Second, in comparison to first-order setting, the unification
operation of interest to us deals with a richer term structure and involves a
more complicated notion of equality. To accommodate this, we make the
following additions. To treat the richer equality notion, we utilize
invocations to a head normalization procedure at relevant points in
the computation. Then we partition the unification computation into
first-order and higher-order parts, so that the former can be handled by
(compiled) WAM style instructions, and the latter by an auxiliary
interpretive procedure that is based on the higher-order pattern
unification algorithm. Notice that this partitioning is something that
must happen dynamically because whether the unification problem is a
first-order or a higher-order one depends also on a term whose
structure is known only at runtime. To deal with this, we build
appropriate machinery into the abstract machine instruction
set that is responsible for recognizing and delaying the higher-order
parts of unification, and for invoking the interpretive phase of
unification at chosen computation points. Further, we provide
devices for delaying the unification problems that are recognized to
be beyond the $\Ll$ subset during the interpretive phase and for
carrying them across goal invocations, to be re-examined when variable
bindings may have altered their status.

The details of the additional support that is summarized above are
presented in the next two sections.
The rest of the discussion in this section provides a sketch of the memory
structure of our abstract machine and the underlying processing model
that will be needed in order to explain these details.

The basic data areas in our abstract machine consist of a code area,
a heap, a stack, a collection
of registers, a push down list (PDL) and a trail.
The first four categories of data areas are familiar from conventional
machine architectures although some of them have different actual
purposes in our setting.
The code area contains the compiled forms of clauses that constitute
the definitions of predicates.
The heap is a global memory space for holding data
that is accessible at any point of computation; specifically, this is
where complex terms that survive after the successful completion of a
goal must be placed. One of the uses of the stack, that is similar to
the use made of it in conventional languages, is to record environment
frames for calls to particular clauses that constitute the definition
of a predicate; such frames will store register images and other
relevant data that need to be maintained between calls to goals that
are part of the body of the clause in question.
The stack is also used to store information for handling
nondeterminism, a feature that is peculiar
to logic programming languages.
In particular, when alternatives are available during clause selection,
the contents of relevant registers should be saved,
so that the execution context can be recovered when it is necessary to
attempt a different clause choice, \ie, when backtracking
occurs. Such information is maintained
in structures called  choice points, which are interleaved with
environment frames on the stack.
In our abstract machine, the
stack is also used to maintain information that is needed to support
augment goals. We defer discussion of this usage till the next
section.
Registers are of two kinds: those that store data and those that
are needed for execution control.
Examples of the former include a set of data registers, $A_1$, ..., $A_n$ that
are used for passing arguments across calls to clause definitions, and
the $S$ register that points to the next argument of a complex
first-order term (which is an application with a rigid head).
The set of registers relevant to execution control consist of the program
pointer, {\it P}, the continuation pointer, {\it CP}, the top of the
heap register,
{\it H}, the most recent environment frame register,
{\it E}, the most recent choice point register, {\it B} and the top of
the trail register, {\it TR}.
Both sets of registers will be enriched to support higher-order
features, as we discuss in the next two sections.
The push down list, PDL, is used within the interpretive
unification process for recording the subproblems that are created by
the process of term simplification discussed in
Section~\ref{chp:interpreter}.
The trail area is also used to
assist the branching behavior, which records images of the fragments
of the heap and stack that need to be recovered upon backtracking.

\begin{figure}\footnotesize
\begin{tabbing}
\quad\= c1\ \ : \=\{\ \=Set up a choice point on the top of stack and record that the next candidate clause \ \ \ \=\}\kill
\>{\it /*  copy a a. */}                                                                                     \\
\> {\it C1}:  \>\{  \>Set up a choice point on the top of stack and record that the next candidate clause       \\
\>            \>    \>is available at {\it C2}.                                                            \>\} \\
\> {\it L1}:  \>\{  \>Unify arguments of the incoming goal with those of the clause head.                  \>\} \\
\>            \>\{  \>Return by continuing from the continuation point.                                    \>\} \\
\>{\it /*  copy (app T1 T2) (app T3 T4)} $\pif$ {\it copy T1 T3, copy T2 T4. */}                             \\
\> {\it C2}:  \>\{ \>Recover relevant registers from the information in the latest choice point on the stack,      \\
\>            \>   \>update the choice point and record that the next candidate clasue is available at {\it C3}. \>\} \\
\> {\it L2}:  \>\{ \>Set up an environment frame on the top of the stack.                                        \>\} \\
\>            \>\{ \>Unify arguments of the incoming goal with those of the clause head.                         \>\} \\
\>            \>\{ \>Set up arguments for {\it copy T1 T3}.                                                      \>\} \\
\>            \>\{ \>Shrink the environment frame, update the continuation point to the next instruction              \\
\>            \>   \>and call {\it copy}.                                                                        \>\} \\
\>            \>\{ \>Set up arguments for {\it copy T2 T4}.                                                      \>\} \\
\>            \>\{ \>Remove the latest environment frame from the stack.                                            \>\} \\
\>            \>\{ \>Call {\it copy}.                                                                            \>\} \\
\> {\it /* copy (abs T1) (abs T2)} $\pif$ $Pi\ c\plam \ (copy\ c\ c\ \pimp \ copy\ (T1\ c)\ (T2\ c))$. {\it */} \\
\> {\it C3}:  \>\{ \>Recover relevant registers from the information in the latest choice point on the stack,        \\
\>            \>  \>and remove the choice point.                                                             \>\} \\
\> {\it L3}:  \>\{ \>Set up an environment frame on the top of the stack.                                    \>\} \\
\>            \>\{ \>Unify arguments of the incoming goal with those of the clause head.                     \>\} \\
\> {\it H1}:  \>\{ \>Carry out control actions for entering a generic goal and then an augment goal.         \>\} \\
\>            \>\{ \>Set up arguments for {\it copy (T1 c) (T2 c)}.                                          \>\} \\
\>            \>\{ \>Shrink the environment frame, update the continuation point to the next instruction          \\
\>            \>   \>and call {\it copy}.                                                                    \>\} \\
\> {\it H2}:  \>\{ \>Carry out control actions for leaving an augment goal and then a generic goal.          \>\} \\
\>            \>\{ \>Remove the latest environment frame from the stack.                                        \>\} \\
\>            \>\{ \>Return by continuing from the continuation point.                                       \>\}
\end{tabbing}
\caption{Compiled computations underlying the program {\it
copy}.}\label{fig:copy_proc_model}
\end{figure}

Computations occur within our abstract machine from executing a sequence
of instructions that are generated from compiling goals, which
correspond to the user query or the bodies of clauses, or from
compiling the selection of a clause for a predicate and the subsequent
unification with the head of the clause.
The compilation of a goal
is organized as follows. First, instructions are generated to realize
the processing of the logical symbols that appear in a complex
goal. Eventually, an atomic goal is reached. At this point,
instructions are produced to set up the arguments of the
goal in the argument registers; if these arguments are complex
terms or variables, they will reside in either the heap or in the
environment frame and the relevant registers will contain references
to these structures. The last instruction for the atomic goal will be
a call to the code for the predicate in question.
Apart from transferring control to the (next) relevant clause for a
predicate, the code for clause selection has the responsibility of
setting up a choice point in the stack to represent the remaining
alternatives.
The first action that the code for unification with the head of a
selected clause must do is set up an environment frame if one is
needed. The remaining instructions are responsible for carrying out
the needed unification between the arguments appearing in the clause
head and the ones passed in the argument registers from the invocation
of the atomic goal. If this unification is successful, computation
passes to the instructions arising from the compilation of the goal
constituting the clause body, whose treatment we have already
described.
If this goal is solved successfully, then computation must return to
the caller and the last instruction for the clause body will have the
effect of realizing this. Notice that the environment frame that was
created for this clause can be released at this point provided it is
not needed for backtracking, in which case it will be protected by a
choice point that appears above it in the stack.
Of course, failure can occur in the course of unification with the
head of a clause. This triggers a backtracking procedure whose first
task is to carry out a resetting of the heap and stack state to what
it was at the current most recent choice point. The information for
such a resetting is stored in the trail and, hence, this process
is referred to as the ``unwinding of the trail.''  Once this is done,
the relevant registers are restored from the information available
from the most recent choice point and computation proceeds
to the next clause definition (also recorded with the choice point),
after updating or discarding the choice point itself depending on
whether or not further alternatives are available.

The control computations are optimized in a manner similar to that
in the WAM within our abstract machine as well.
First, upon making a call, the environment frame of the caller is dynamically
shrunk by discarding permanent variables whose binding information are no
longer needed for solving the goals in the clause body remained to be
processed; this process is referred to commonly as ``environment trimming.''
Second, when a clause body is constructed from a sequence of conjunctions and
the last conjunct is atomic, last call optimization~\cite{Warren80} is
performed.
Essentially, the caller's environment frame is deallocated from the stack
before computation actually proceeds to the callee, and the call is carried
out after setting the continuation point register to the continuation
point passed to the caller, so that the callee can directly
returns to its grand parent in the call graph. This optimization
subsumes the traditional tail recursion optimization in the logic
programming setting.

\begin{figure}\footnotesize
\begin{tabbing}
\quad\= {\it copy}\ : \=\{\ \=\dquad\=variable:\dquad\dquad\= continue with the instruction at xxxxxxxxxxxxxxxxxxxxxxxx \=\}\kill
\>      {\it copy}\ : \>\{  \>Switch on the head of the (head normal form of) the first actual argument of {\it copy}: \\
\>                    \>    \>      \>variable:        \> continue with the instruction at {\it C1}.                   \\
\>                    \>    \>      \>de Bruijn index: \> continue with the instruction at {\it C1}.                   \\
\>                    \>    \>      \>constant:        \> continue with the instruction at {\it S}.       \>\}         \\
\>      {\it S}\ :    \>\{  \>Switch on the given constant:                                                            \\
\>                    \>    \>      \>{\it a}\ :       \> continue with the instruction at {\it L1}.                   \\
\>                    \>    \>      \>{\it app}\ :     \> continue with the instruction at {\it L2}.                   \\
\>                    \>    \>      \>{\it abs}\ :     \> continue with the instruction at {\it L3}.      \>\}         \\
\>      {\it C1}\ :   \>...
\end{tabbing}
\caption{Indexing on {\it copy}.}\label{fig:copy_index}
\end{figure}

We illustrate the compilation model and the associated processing
scheme that we have described relative to the simple $\lp$ program
appearing in Figure~\ref{fig:copy} that defines the {\it copy}
predicate. A high-level pseudo code description of the compiled
program in our implementation is contained in
Figure~\ref{fig:copy_proc_model}.

A final aspect to be mentioned with regard to the compilation model is
the optimization corresponding to the detection of determinism.
The runtime treatment of nondeterminism involves the manipulation of choice
points that is known to be costly and can often be eliminated by
utilizing the structure of actual arguments of atomic goals to prune choices
early during execution. For this purpose, a special set of instructions are
included that allow clause choices to be indexed by the head arguments.
Taking the {\it copy} example, instructions in Figure~\ref{fig:copy_index} can
be added to those in Figure~\ref{fig:copy_proc_model} for the purpose of
indexing.

\section{Compiling the New Search Primitives in $\lambda$Prolog}\label{sec:ho_control}
We now consider the extensions to the basic processing model to deal
with generic and augment goals. Our discussion only sketches these
extensions to the extent needed for a complete description of our
abstract machine. A more thorough treatment may be found in
~\cite{NJK95lp}.

As described in the previous chapters, the presence of generic goals
requires a more careful treatment of unification. More specifically,
to deal with the scoping effect of such goals on names, universe
levels are associated with constants and logic variables and are
examined and adjusted by the unification process. The determination of
the appropriate universe level in our abstract machine is based on a
global universe counter, which starts from $0$ on the top-level query,
and is increased or  decreased upon entering or leaving each generic
goal. This global universe counter is maintained in a new
register called {\it UC}. This register is incremented and decremented
by two new instructions, {\it incr\_universe} and {\it decr\_universe},
respectively.
Some of the actions in the WAM based model are also modified to
facilitate the proper manipulation of the universe counter.
The contents of the {\it UC} register is stored
in choice points so that this register can be restored upon
backtracking. These contents are also recorded in environment frames;
the instructions that create terms corresponding to the arguments of
atomic goals appearing in the body of a clause and possibly embedded
within generic goals may need the old value in this register for
tagging variables that are bound by the implicit quantifiers at the
clause level.

It is necessary also to deal with the direct effects of a generic
goal: such a goal must give rise to a new constant that is tagged with
the (incremented) value of the UC register and that must then be
substituted in the body of the goal for the quantifier variable. In
our abstract machine, we deal with these requirements by assigning a
slot in the environment frame to the quantified variable---thereby
treating it as a {\em permanent} variable in WAM terminology---and by
storing the appropriate constant in this slot. These actions are
carried out by a new instruction called {\it set\_univ\_tag}: as
expected, this instruction takes as operands a displacement in the
environment frame and a constant.
As a concrete example of the design above, the pseudo instructions
from label {\it H1} to {\it H2} in Figure~\ref{fig:copy_proc_model} that
corresponds to the generic goal
$Pi\ c\plam \ (copy\ c\ c\ \pimp \ copy\ (T1\ c)\ (T2\ c))$
can take the following structure.
\begin{tabbing}
\quad\= c1\ \ : \=\{\ \=Set up a choice point on the top of stack and record ssssss\ \ \ \=\}\kill
\> {\it H1}\ :  \>\{ \>{\it incr\_universe}                                                                          \>\} \\
\>              \>\{ \>{\it set\_univ\_tag} $<$offset to the environment frame$>$, $c$                               \>\} \\
\> {\it H3}\ :  \>\{ \>Carry out control actions for entering an augment goal.                                       \>\} \\
\>              \>\{ \>Instructions for $copy\ (T1\ c)\ (T2\ c)$.                                                    \>\} \\
\> {\it H4}\ :  \>\{ \>Carry out control actions for leaving an augment goal.                                        \>\} \\
\> {\it H2}\ :  \>\{ \>{\it decr\_universe}                                                                          \>\}
\end{tabbing}

Goals in $\lambda$Prolog could also have the form
$(Sigma\ x\plam\ G)$, \ie, they could be explicitly existentially
quantified. Such goals may be permitted in Prolog too, but, because of
the simple syntactic structure of goals in that setting, in
particular, the absence of generic goals, such goals can be treated
statically by moving the existential quantifiers out into universal
ones over the entire clause and can then be treated via standard
techniques. In our case, we can almost use the same scheme. There is,
however, one exception: the particular location of the existential
quantifier may have an impact on what universe index is to be
stored with the variable. To accommodate this, we add a further
instruction that is called {\em tag\_exists} to our abstract
machine. This instruction takes a variable, which is eventually a stack
or heap location, as an argument and sets its universe index to the
value currently in the {\em UC} register.

The semantics of an augment goal $D\pimp G$ require the addition of $D$
to the existing set of program clauses before the processing of $G$,
and the retraction of these added clauses upon the successful
solution of $G$.
The searching mechanism used for clause selection has to therefore
support dynamic modifications to the available predicate definitions.
To realize this, a memory component called an {\it implication point}
is introduced. These implication points are stored on the stack and a
new register, the $I$ register, is introduced to record the most
recent implication point. Each implication point also records the most
recent implication point at the time of its creation; in other words,
the sequence of implication points themselves form a stack. Suppose
that $D$ provides (additional) clauses for the predicates
$\{p_1,\ldots,p_n\}$. Then one of the components contained in the
implication point corresponding to the addition of $D$ is a search
table that will ultimately yield a pointer to the compiled form of the
code for each of these predicates. If no entry is found for a
particular predicate when searching from this implication point, the
search continues from the implication point that this one points to;
thus, the overall program context existing at any stage of computation
is completely defined by the contents of the $I$ register.  The
implication point also contains a {\em next clause} table of size $n$
that provides pointers to the definition (or code) for each of the
predicates $p_1,\ldots,p_n$ that existed at the time of its creation
paired with the implication point that corresponds to this
definition. This table complements a special instruction called {\it
  trust\_ext} to complete the compiled form of the code for the
predicates $p_1,\ldots,p_n$ as we describe later. Notice that the
right next clause table to use is determined by the implication point
that added the code currently being tried for the relevant
predicate. To isolate this implication point, we add to the abstract
machine yet another register called {\it CI}.

Two new instructions are introduced to support the compilation of an
augment goal. The {\it push\_impl\_point} instruction is used upon
entering an augment goal for the creation of an implication
point. This instruction is also responsible for setting up the next
clause table for the implication point, something that is done by
searching the program context given by the current contents of the $I$
register for definitions for each of the relevant predicate. The {\it
  push\_impl\_point} instruction takes as argument a pointer to a
compile-time prepared table that contains information about the
predicates for which code is being added and also pointers to the
specific code that needs to be included.
Symmetrically, the instruction {\it pop\_impl\_point} serves to remove
the latest implication point from the stack upon leaving an augment
goal. This action is carried out simply by setting the {\it I}
register to the implication point reference stored in the one that
this register currently points to.
Considering the {\it copy} example, now the pseudo instructions labeled from
{\it H3} to {\it H4} that correspond to the augment goal
$(copy\ c\ c\ \pimp \ copy\ (T1\ c)\ (T2\ c))$
can take the following form.
\begin{tabbing}
\quad\= c1\ \ : \=\{\ \=Set up a choice point on the top of stack and record ssssss\ \ \ \=\}\kill
\> {\it H3}\ : \>\{ \>{\it push\_impl\_point t}                                                                    \>\} \\
\>             \>\{ \>Instructions for $copy\ (T1\ c)\ (T2\ c)$.                                                 \>\} \\
\> {\it H4}\ : \>\{ \>{\it pop\_impl\_point}                                                                     \>\}
\end{tabbing}
We assume above that {\it t} is a pointer to a table prepared for the
addition of the clause $copy\ c\ c$ to the existing collection of
predicate definitions.

Code that is added dynamically for a predicate must allow for the
possibility that it is extending an already existing definition. To
support this situation, the code that is normally generated from the
clauses for the predicate is enclosed within a
{\it try\_me\_else} and a {\it trust\_ext} instruction. The leading
{\it try\_me\_else} sets up a choice point with the indication that
the alternative definition starts from the {\it trust\_ext}
instruction at the end of this segment of code.
The {\it trust\_ext} instruction takes as argument an index into a
next clause table.
The {\it trust\_ext} instruction first retrieves a pointer to the next
clause to try for the predicate from the next clause table stored in
the implication point referenced by the {\it CI} register and it
resets this register to the associated implication point also obtained
from this table. It then transforms the rest of the computational
context as needed for backtracking by using the contents of the
current choice point, which it then discards.

A subtle but important point to be noticed about the clauses that
appear in augment goals is that these may contain free variables in
them. For example, consider the following generic goal that appears in
one of the clauses for the {\it copy} predicate:
\begin{tabbing}
\qquad\=\kill
\>$Pi\ c\plam \ (copy\ c\ c\ \pimp \ copy\ (T1\ c)\ (T2\ c))$
\end{tabbing}
Recall that the quantified variable $c$ is treated as a variable for
which space is allocated in the environment record for the parent copy
clause. Further, the processing of the universal quantifier results in
a constant (with appropriate universe index) being bound to this
variable. When interpreting the embedded clause $copy\ c\ c$,
therefore, it is important to have available the environment record of
the parent clause in order to interpret the ``variable'' corresponding
to the occurrences of $c$. In short, we treat clauses as closures, to
be interpreted relative to an environment that is pointed to by a
special register called {\it CE}. Use is made of a new instruction
called {\it init\_variable} whenever it is necessary to get the
binding for a variable from the ``parent'' environment. This
instruction takes two arguments: a register or an environment slot
designating the location of the variable local to the clause being
considered and the environment slot for the parent clause from which
the binding must be obtained. The instruction uses its two arguments
to tie these two variables together.

As an illustration of the discussion of the compilation of embedded
clauses, the clause {\it copy c c} that occurs within the generic goal
just considered would be compiled into the following sequence of
(pseudo-)instructions:
\begin{tabbing}
\quad\= c1\ \ : \=\{\ \=Set up a choice point on the top of stack and record ssssss\ \ \ \=\}\kill
\> {\it D1}\ :  \>\{ \>{\it try\_me\_else} {\it D2}                                            \>\} \\
\>              \>\{ \>{\it init\_variable} $\langle${\it local location of
c}$\rangle$, {\it Yi} \>\} \\
\>              \>\{ \>Code for unifying first two argument registers \\
\>              \>   \> with variable denoting $c$ local to this environment.                                      \>\} \\
\>              \>\{ \>Return control to the continuation point.  \>\} \\
\> {\it D2}\ :  \>\{ \>{\it trust\_ext} {\it 1}                                             \>\}

\end{tabbing}
Here {\it Yi} denotes the location of the slot assigned to the
universally quantified variable corresponding to $c$ in the
environment record pointed to by the {\it CE} register. It is, of
course, necessary to set this register appropriately for each clause
that is being tried. To facilitate this, a pointer to the relevant
environment record is stored in the implication point at the time that
it is set up. Notice also that the index for the {\it trust\_ext}
instruction here is {\it 1} because there is code for exactly one
predicate that is added by the associated augment goal.

A final point concerns the instructions for invoking the code
for predicates. As we have noted in this section, the entry
point into such code can change during execution. For this reason, we
need a special set of calling instructions that will
initiate the search for appropriate code from the implication point
referenced by the {\it I} register. These instructions will,
for instance, have to be used for any calls to the {\it copy}
predicate whose compilation we have just considered.
Note, however, that the old WAM style calling instructions are also
retained in our abstract machine. These can be used for predicates
whose code cannot be altered dynamically. Moreover, it is preferable
to use them wherever possible because the address to which control
needs to be transferred then does not need to be calculated at
runtime.

\section{Compilation of Higher-Order Pattern Unification}\label{sec:ho_unif}
We now turn our attention to providing support for higher-order
pattern unification. We first consider extensions for this purpose to
the data areas present in the original structure of the WAM.
These extensions are of two kinds: the introduction of new devices and
enhancements and modification to the ones already present in the WAM.
The specifics of these changes are as follows. First, we add new
registers called {\it Head}, {\it
  ArgVector}, {\it NumArgs} and {\it NumAbs} that provide access to
the head, the arguments, the number of arguments and the binder length
of a head-normal form right after it has been computed.
Second, in addition to the role it plays in realizing the interpretive
unification process,
the PDL is also used to temporarily maintain higher-order unification
problems that are delayed when executing the compiled form of
unification arising from matching with the clause head.
Third, unification problems that lie outside the $\Ll$ subset need to
be carried as constraints across goal invocations and the heap is used
to maintain such problems in the form of a list of disagreement
pairs. The beginning of this list is recorded in a new register called
{\it LL}.
The heap is further used to store the terms that are created in the
course of head-normalization and in the binding phase of pattern
unification.
In the intended scheme, $\beta$-contractions are carried out
destructively during head-normalization so as to share
the effects of such rewriting steps. Since it may be necessary to undo
these mutations on backtracking, we also change the trail so that it
additionally maintains a record of any such mutations that arise
during processing.

As mentioned in Section~\ref{sec:basic_model}, the unification
on the arguments of a clause essentially consists of a
first-order and a higher-order part,
whereas WAM style instructions
for unification are only sufficient in handling the former.
Our abstract machine still uses the WAM
style instructions to solve the first-order subproblems, and delays the
higher-order ones by pushing them onto the PDL. The problems left on
the PDL in this way are examined by an interpretive pattern
unification procedure that is invoked as the culminating instruction
in the sequence that realizes unification with the clause head.
The structure of the unification part of the processing model can thus
be described schematically as follows:
\begin{tabbing}
\quad\= \{\ \ \=  For \=each argument in the clause head       \\
\>       \>         \>\{ \ \=Instructions for carrying out the first-order part unification and  \\
\>       \>         \>     \>postponing the higher-order part onto the PDL.     \\
\>       \>         \>\} \\
\>       \> Invoke the interpretive pattern unification procedure on the PDL. \\
\> \}
\end{tabbing}

\begin{figure}\footnotesize
\begin{tabbing}
\dquad {\it unify}\=\ $(t,\ ${\boldmath $s$}$)$  \\
                  \> switch on the structure of $t$\ : \\
                  \> case \=$\lambdadb(n, t')$\ : \\
                  \>      \>{\bf create {\boldmath $t$} on the heap}  \\
                  \>      \>{\bf interp\_unify({\boldmath $t,\ s$})} \\
                  \> case \>$(F\ a_1\ ...\ a_n)$, where $F$ is a variable and $n > 0$\ : \\
                  \>      \> let \= $t'$ be a term of form   \\
                  \>      \>     \> $(F\ a_1\ ...\ a_n)$, where $F$ is a new logic variable, \\
                  \>      \>     \> if this is the first occurrence of the variable in the clause;\\
                  \>      \>     \> $(f\ a_1\ ...\ a_n)$, where $f$ is the term to which the variable $F$ is bound, \\
                  \>      \>     \> if this is the subsequent occurrence of the variable in the clause. \\
                  \>      \>{\bf create {\boldmath $t'$} on the heap} \\
                  \>      \>{\bf interp\_unify({\boldmath $t,\ s$})} \\
                  \> case \>$X$, where $X$ is a variable\ : \\
                  \>      \>if this is the first occurrence of $X$ in the clause, then {\bf bind({\boldmath $X,\ s$})}. \\
                  \>      \>else {\bf interp\_unify({\boldmath $t',\ s$})}, where $t'$ is the term to which $X$ is bound. \\
                  \> case \>$(c\ a_1\ ...\ a_n)$, where $c$ is a constant, and $n \geq 0$\ : \\
                  \>      \>{\bf head\_norm(\boldmath $s$)} \\
                  \>      \>{\bf if \boldmath $s$ is \boldmath $(r'\ b_1\ ...\ b_m)$, where \boldmath $r'$ is rigid and \boldmath $m\geq 0$} \\
                  \>      \>{\bf then}\ \={\bf if \boldmath $r' \neq c$ or \boldmath $n \neq m$ then backtrack} \\
                  \>      \>            \> {\bf else } for $1\leq i \leq n$: {\it unify}\ $(a_i,\ ${\boldmath $b_i$}$)$  \\
                  \>      \>{\bf else}  \> \\
                  \>      \>            \>{\bf create \boldmath $t'$ as \boldmath $(c\ X_1\ ...\ X_n)$ on heap, where \boldmath $X_i$ are new variables} \\
                  \>      \>            \>{\bf if \boldmath $s$ is a logic variable \boldmath $X$ }\\
                  \>      \>            \>{\bf then}\ \={\bf if \boldmath $uc(X) \leq uc(c)$ then backtrack} \\
                  \>      \>            \>            \>{\bf else \boldmath bind$(X,\ t')$} \\
                  \>      \>            \>{\bf else} /* $s$ must be a higher-order term */ \\
                  \>      \>            \>            \>{\bf push the pair \boldmath $(t', \ s)$ onto PDL} \\
                  \>      \>            \>for $1\leq i \leq n$: {\it unify}\ $(a_i, X_i)$
\end{tabbing}
\caption{The unification model in our compilation
implementation.}\label{fig:unify}
\end{figure}

Now we consider the compilation
of the unification on each pair of arguments. Compared with what has
to be dealt with by the WAM, the following new issues arise in our
setting.
First, a richer collection of term structures participate in the
computation. Second, a head normalization procedure has to be invoked
to bring terms into comparable forms at the necessary
points. Finally, relevant instructions have to be enhanced with
the ability to properly separate higher-order subproblems from
first-order ones, taking the necessary steps to solve the latter while
pushing the former onto the PDL.
Taking these issues into account, the processing in our implementation can be
described by the {\it unify} procedure in Figure~\ref{fig:unify}.
The first argument to this procedure is the argument from a clause head,
\ie, whose structure is statically known, and is assumed to be normalized
at compilation time. The second argument is the one dynamically appearing
at runtime. It should also be noted that the actions carried out in
compilation and at runtime are both present in this procedure, and we use
bold letters to distinguish the latter.

The auxiliary functions {\it interp\_unify}
and {\it head\_norm} in {\it unify} denote the interpretive
pattern unification and head normalization procedures respectively.
A call to the procedure {\it bind} in a form {\it bind\ $(X,\ t)$}
essentially carries out the action of binding a logic variable $X$ to the
term $t$.
In the situation when $X$ is from a static argument of the clause head,
the logic variable is not explicitly created, but, rather, given by a
data register or a slot in the environment frame.
Binding in this case is carried out by placing a reference to the term
$t$ in the relevant place.
Finally, in the case when the static term $t$ input to the {\it unify}
procedure is a first-order application and the dynamic term $s$ is a logic
variable or higher-order term, the recursive calls to {\it unify} simply
serve to construct the arguments of $t$ on the heap. For this reason, it is
not necessary to actually create the new variables $X_i$'s that are used
in the presentation of the pseudo code. Instead, space is allocated on
the heap for an argument vector of size $n$ and the recursive calls to
{\it unify} enter a term creation mode---known as the WRITE mode in
contrast to the READ mode that is used when term structure needs to be
analyzed---during which the arguments of $t$ are
created and references to them are placed into the relevant slots in the
argument vector.

The conventional WAM style term creation and unification instructions
are categorized into the {\it put}, {\it set}, {\it get}
and {\it unify} classes. Roughly mapping to the {\it unify} procedure in
Figure~\ref{fig:unify}, the {\it get} class of instructions can be used to
carry out the actions required by the cases where the static term is a
first-order application and where it is a constant or variable that appears
directly as an argument of the clause head.
When the {\it unify} procedure is invoked recursively over the
arguments of the (static) applications, the unifications over the embedded
variables and constants can be handled by the set of {\it unify} instructions.
The {\it put} and {\it set} instructions are used in the WAM solely for
setting up the the actual arguments of atomic goals and do not get used
in head unification. In our context, when the static
term has a higher-order structure, it has to be first
created and then handed to the interpretive unification process.
The term creation actions are carried out by
the {\it put} and {\it set} classes of instructions, \ie, these
instructions may be interleaved with {\it get} and {\it unify}
instructions in the compilation of head unification.

Within this picture, now we start to examine the enhancements to each
category of instructions for supporting the higher-order aspects
of unification.
Since the {\it set} category of instructions are in
fact a light-weight form of those in the {\it unify} class, \ie,
their actions are the same as those carried out
by the {\it unify} instructions in the WRITE mode, we do not discuss
these separately in what follows.

In contrast to the first-order setting, term creation in our context
has to deal with a
richer collection of structures. First, the head of (a head normal form of)
an application can be a de Bruijn index or a logic variable in
addition to being a constant.
For this reason, the {\it put\_structure} instruction in the WAM is
generalized into {\it put\_app}. This instruction gets three
arguments: a data (argument) register $A_i$, a data register or an
offset into an environment frame $X_j$ and a positive number $n$.
This instruction first creates an application term on the
heap with its head being the term referred to by $X_j$ and an empty argument
vector of size $n$. Then $A_i$ is set to refer to the new application term
and the $S$ register
is prepared to refer to the beginning of the argument vector for the subsequent
instructions to actually fill in the arguments.
The second source of higher-order structures is the appearances of de Bruijn
indexes and abstractions. For the creation of the former, new instructions
\begin{tabbing}
\dquad {\it put\_index\ $A_i$, $n$} \dquad and\dquad {\it unify\_index $n$}
\end{tabbing}
are introduced. The first one is used for a de Bruijn index
that is not directly an argument of an application. Its execution
constructs a term corresponding to the de Bruijn index $n$ on the top
of the heap and sets the data register $A_i$ to refer to it.
The {\it unify\_index} instruction corresponds to an application argument.
It can be only invoked in the WRITE mode and its effect is to create
a term corresponding to the de Bruijn index $n$ in the
heap location given by the register $S$ and to increment $S$ to point
to the next argument vector slot.
Similarly, the creation of an abstraction $\lambdadb(n, t)$ is realized by
the pair of new instructions
\begin{tabbing}
\dquad{\it put\_lambda\ $A_i$, $X_j$, $n$} \dquad and\dquad {\it unify\_lambda\ $X_j$, $n$},
\end{tabbing}
depending on whether the abstraction appears directly as an argument of an
application. A reference to the term $t$ is assumed to be contained by the data
register or environment offset $X_j$.

The instructions constructing compound terms assume that the head of an
application and the body of an abstraction are given by data registers.
However, these components can in particular situations correspond to
permanent variables which reside in environment frames
on the stack. In these situations, the relevant permanent variables have to be
globalized prior to use. To facilitate this,
our abstract machine include the instructions
\begin{tabbing}
\dquad{\it globalize\ $Y_i$, $A_j$} \dquad and \dquad {\it globalize\ $A_i$}.
\end{tabbing}
The first one dereferences the permanent variable $Y_i$ given by an
offset to an environment frame. If the resulting term still resides on the
stack, it is copied to the top of the heap and then sets both that
stack cell and the data register $A_j$ to refer to the newly created
heap cell. Otherwise $A_j$ is made to be a reference to the
dereferenced result.
The second instruction simply dereferences the given
$A_i$, carries out the globalizing actions described before if necessary and
leaves a reference to the appropriate heap term in $A_i$.

The {\it get} and {\it unify} instructions are used for carrying out
compiled unification. These instructions are enhanced to handle terms
whose structures may be revealed to be higher-order at runtime.
Changes are made for the instructions
\begin{tabbing}
\dquad{\it get\_structure\ $A_i$, $f$, $n$},\quad {\it
  get\_constant\ $A_i$, $c$ }\quad and\ {\it unify\_constant\ $c$},
\end{tabbing}
in which $A_i$ is required to be a data register referring to the
incoming term, $f$ and $c$ are required to be constants and $n$ is a
number denoting the arity of the application.
Executing these instructions (in the READ mode for the last
instruction) first invokes the interpretive
head normalization procedure on the term referred to by $A_i$ for the
first two instructions and the one referred to by the $S$ register for
the last. Let the resulting term be $s$; as already explained, its
decomposition will be given by the contents of the registers  {\it Head},
{\it ArgVector}, {\it NumArgs} and {\it NumAbs} at the end of head
normalization. If $s$ has a higher-order
structure, \ie, if it is an abstraction or a flexible application, a
disagreement  pair with the first term being (a reference to) $s$ and
the second referring to the current top
of heap or to the location given by $S$ is created on the PDL.
In the situation when {\it get\_constant} or {\it unify\_constant}
is executed, the constant $c$ is then created as the second term of the
disagreement pair. When the executed instruction is {\it get\_structure},
the term pushed onto the top of heap is then an application with an empty
argument vector of size $n$ and with its head referring to a new constant term
corresponding to $f$. Further, the $S$ register is set to the first entry of
the argument vector, and execution proceeds to the following
{\it unify} instructions in WRITE mode. The {\it unify\_value\ $X_i$}
instruction is also changed so that when it is executed in the READ
mode, it causes the pattern unification procedure, rather than the
first-order unification procedure, to be invoked in interpretive mode
on the pair of terms given by the register or environment offset $X_i$
and the $S$ register.
In addition, a new instruction
\begin{tabbing}
\dquad{\it pattern\_unify $X_j$,\ $A_i$}
\end{tabbing}
is introduced as a variant of {\it unify\_value} in the READ mode.
This instruction appears at the end of a sequence of {\it put} and
{\it unify} (in the WRITE mode) instructions that serves to create a
higher-order term appearing in a clause head. This instruction also
invokes the higher-order pattern unification procedure in interpretive
mode to unify the created term that is referenced by $X_j$ and the
incoming term that is given by the argument register $A_i$.

For a concrete example of the usage of our unification and term creation
instructions, we can consider the compilation of the term
$(app\ X\ (abs\ (y\plam \ X)))$ as an argument within a clause head,
assuming that $app$ and $abs$ are the constants that we encountered in
the {\it copy} program.
The instructions resulting from a compilation of this term are shown
in Figure~\ref{fig:unify_exp}.
\begin{figure}
\begin{tabbing}
\dquad\dquad\= {\it get-structure-}\dquad\={\it A22,}\ \={\it A22,}\ \={it A22} \= \kill
\>{\it get\_structure}  \> {\it A1,}\> {\it app,}\>{\it 2} \> \%\ {\it A1 = (app\ }\= \\
\>{\it unify\_variable} \> {\it A2}\>          \>        \> \%                   \>{\it X\ }\=         \\
\>{\it unify\_variable} \> {\it A3}\>          \>        \> \%                   \>         \> {\it A3)} \\
\>{\it get\_structure}  \> {\it A3,}\> {\it abs,}\>{\it 1} \> \%\ {\it A3 = (abs\ } \\
\>{\it unify\_variable} \> {\it A4}\>          \>        \> \%                   \> {\it A4)} \\
\>{\it put\_lambda}     \> {\it A5,}\> {\it A2,} \>{\it 1} \> \%\ {\it A5 = $\lambdadb(1, X)$} \\
\>{\it pattern\_unify}  \> {\it A4,}\> {\it A5} \>        \> \%\ {\it A4 = A5}
\end{tabbing}
\caption{Compiled unification over a head argument $(app\ X\ (abs\
(y\plam \ X)))$.}\label{fig:unify_exp}
\end{figure}

The instruction set for our abstract machine includes a new
instruction called {\it finish\_unify} that is used at the end
of the processing of the entire clause head. This instruction invokes
the interpretive
pattern unification procedure over the disagreement pairs that have been
pushed onto the PDL during the head processing. Further, if bindings to logic
variables have actually occurred during head unification,
the global disagreement set recording non-$\Ll$ problems generated
from computation steps prior to the processing of the current clause is also
examined at this stage with the expectation that some of them could actually
become $\Ll$ after the bindings.
It is interesting to note that this way of examining the global disagreement
set could in theory lead to bad performance: if a large number of non-$\Ll$
pairs are carried along across the solutions of atomic queries and only a
relatively small portion of it actually becomes $\Ll$ after the processing
of each clause head, then the repeated examination on the contained
disagreement pairs will be mostly redundant. This conceptual problem can be
solved by using a sophisticated {\it freeze-wake} mechanism proposed
by~\cite{MP93ppcp}. Within this scheme, a unsolvable disagreement pair
is directly associated with the logic variables contributing to it, and
the re-examination is triggered only when the binding of the logic
variable actually occurs. However, the ``extreme'' case described
above in fact rarely occurs in the context that we are interested in: in most
practical $\lp$ programs, it is either the case that
all the disagreement pairs are $\Ll$ the first time they are looked
at, usually because the program itself has been written to adhere to
the $\Ll$ style, or the case that a non-$\Ll$ pair is
transformed into an $\Ll$ one at the end of the processing of the
clause head in which the pair was encountered.
Based on this observation, the simple processing scheme that we have
chosen for delayed disagreement pairs seems justified.

A final new instruction for our abstract machine is
{\it head\_normalize\ $X_i$}, which carries out the head normalization of a
term referred to by the data register or
environment offset given by $X_i$. This instruction is used in the term
creation process
needed for setting up the arguments of atomic goals when it is obvious that a
higher-order structure
has been created. The purpose of enforcing head normalization over such
structures at an early stage is to reduce the overhead of backtracking.
The actual arguments have to be in head normalized form
during the unification operations carried out during the clause
selection. If this normalization is done before a choice point
corresponding to clause selection is created, then the process of
undoing and then redoing it because of a backtracking internal to this
selection process can be avoided.

A comparison between the processing model we have described here and
the one underlying the implementation of {\it Version 1} of the {\it
  Teyjus} system is in order. We focus here only on the issues that
have been discussed so far; more differences will arise when we
consider the treatment of types in the next chapter.
In the earlier abstract machine, the higher-order part of the
unification problems are separated from the first-order ones in a way similar
to our scheme and are also handed to an interpretive unification
procedure for their solution.
However, due to the branching nature of the unification procedure
dealt with in that abstract machine, a more sophisticated (and more
costly) control mechanism has to be considered. In particular, in
addition to the choice point, a
structure known as {\it branch point} had to be introduced for the
purpose of recording choices in the incremental steps taken to solve
rigid-flexible pairs \cite{N03treatment}. Further, these branch points
have to be
examined during backtracking for attempting the next alternative. This
also introduces further complexity in treating choice points at least
in that they have to be differentiated from branch points so that it
is clear what action needs to be taken in the relevant cases.
To avoid the storage of redundant control information for affecting
backtracking caused by the
branching of unification, special attention was paid in the design of
that abstract machine to the precise structure of a branch point.
The creation and the maintenance of branch points is carried out in
that machine by an instruction that is also called {\it finish\_unify}.
The necessity of branch points is entirely eliminated
in our context because we simply delay unification on any pairs that
could cause branching. This has lead to a considerable simplification
of the processing model and is also expected to lead to improvements
in the execution behavior over practical $\lp$ programs.

\section{An Complete Example of Compilation}\label{sec:inst_exp}
We are now in a position to show the complete sequence of instructions
that would be generated for the {\it copy} clauses shown in
Figure~\ref{fig:copy}. The code that we expect a compiler to generate
corresponding to the first two clauses is shown in
Figure~\ref{fig:copy_1}, the code for the last clause is appears in
Figure~\ref{fig:copy_2} respectively, and Figure~\ref{fig:copy_3}
contains the instructions for the embedded clause in the body of the
last clause for the predicate.

\begin{figure}\footnotesize
\begin{tabbing}
\dquad\=copy\ :\dquad\=switch-on-constant\dquad\=L2,\ L1,\ fail,\ fail,\ \dquad\dquad a\= \kill
\> {\it copy}\ :     \>{\it switch\_on\_term}  \>{\it L2,\ \ L1,\ \ fail,\ \ fail}                \\
\> {\it L1}\ :       \>{\it switch\_on\_constant}\>{\it 3,\ \ ht}                               \\
\> {\it L2}\ :       \>{\it try\_me\_else}     \>{\it 2,\ \ L4}  \> \%\ {\it copy}\  \=         \\
\> {\it L3}\ :       \>{\it get\_constant}     \>{\it A1,\ \ a}    \> \%\              \>{\it a}  \\
\>                   \>{\it get\_constant}     \>{\it A2,\ \ a}    \> \%\              \>{\it a}  \\
\>                   \>{\it finish\_unify}                                           \\
\>                   \>{\it proceed}                                                 \\
\> {\it L4}\ :       \>{\it retry\_me\_else}   \>{\it 2,\ \ L6}  \> \%\ {\it copy}\  \=                                       \\
\> {\it L5}\ :       \>{\it allocate}          \>{\it 3}                                                                     \\
\>                   \>{\it get\_structure}    \>{\it A1,\ \ app,\ \ 2}\> \%              \>{\it (app\ }\=                     \\
\>                   \>{\it unify\_variable}   \>{\it A1}            \> \%              \>           \>{\it T1}\ \=          \\
\>                   \>{\it unify\_variable}   \>{\it Y1}            \> \%              \>           \>          \>{\it T2)} \\
\>                   \>{\it get\_structure}    \>{\it A2,\ \ app,\ \ 2}\> \%              \>{\it (app\ }\=                     \\
\>                   \>{\it unify\_variable}   \>{\it A2}            \> \%              \>           \>{\it T3}\ \=          \\
\>                   \>{\it unify\_variable}   \>{\it Y2}            \> \%              \>           \>          \>{\it T4)} \\
\>                   \>{\it finish\_unify}     \>                    \> \%  $\pif$        \\
\>                   \>{\it head\_normalize}   \>{\it A1}            \> \%  {\it A1 = T1} \\
\>                   \>{\it head\_normalize}   \>{\it A2}            \> \%  {\it A2 = T3} \\
\>                   \>{\it call\_name}        \>{\it 2,\ \  copy}    \> \%  {\it copy A1 A2, } \\
\>                   \>{\it put\_value}        \>{\it Y1,\ \ A1}      \> \%  {\it A1 = T2} \\
\>                   \>{\it head\_normalize}   \>{\it A1}            \\
\>                   \>{\it put\_value}        \>{\it Y2,\ \ A2}      \> \%  {\it A2 = T4} \\
\>                   \>{\it head\_normalize}   \>{\it A2}            \\
\>                   \>{\it deallocate}        \\
\>                   \>{\it execute\_name}     \>{\it copy}          \> \%  {\it copy A1 A2.}
\end{tabbing}
\caption{Instructions for the first two clauses of {\it copy}.}\label{fig:copy_1}
\end{figure}

\begin{figure}\footnotesize
\begin{tabbing}
\dquad\=copy\ :\dquad\=switch-on-constant\dquad\=L3,\ L1,\ fail,\ fail\ \dquad\dquad a\= \kill
\> {\it L6}\ :      \>{\it trust\_me}          \>{\it 2}                 \> \% {\it copy\ }\= \\
\> {\it L8}\ :      \>{\it allocate}           \>{\it 2}                 \\
\>                  \>{\it get\_structure}     \>{\it A1,\ \ abs,\ \ 1}    \> \%             \> {\it (abs\ }\= \\
\>                  \>{\it unify\_variable}    \>{\it A3}                \> \%             \>             \> {\it T1)} \\
\>                  \>{\it get\_structure}     \>{\it A2,\ \ abs,\ \ 1}    \> \%             \> {\it (abs\ }\= \\
\>                  \>{\it unify\_variable}    \>{\it A4}                \> \%             \>             \> {\it T2)} \\
\>                  \>{\it finish\_unify}      \>                        \> \% $\pif$          \\
\>                  \>{\it incr\_universe}     \>                        \> \% {\it (}\={\it Pi\ }\=    \\
\>                  \>{\it set\_univ\_tag}     \>{\it Y1,\ \ c}           \> \%        \>     \> {\it c$\plam$\ }\= \\
\>                  \>{\it push\_impl\_point}  \>{\it 1,\ \ t}            \> \%        \>     \>                 \> {\it (}\={\it (copy c c) $\pimp$} \\
\>                  \>{\it put\_app}           \>{\it A1,\ \ A3,\ \ 1}     \> \%        \>     \>                 \>        \>{\it A1 = (T1\ }\= \\
\>                  \>{\it globalize}          \>{\it Y1,\ \ A255}        \\
\>                  \>{\it set\_value}         \>{\it A255}              \> \%        \>     \>                 \>        \>                \>{\it c)} \\
\>                  \>{\it head\_normalize}    \>{\it A1}                \\
\>                  \>{\it put\_app}           \>{\it A2,\ \ A4,\ \ 1}     \> \%        \>     \>                 \>        \>{\it A2 = (T2\ }\= \\
\>                  \>{\it set\_value}         \>{\it Y1}                \> \%        \>     \>                 \>        \>                \>{\it c)} \\
\>                  \>{\it head\_normalize}    \>{\it A2}                \\
\>                  \>{\it call\_name}         \>{\it 1,\ \ copy}         \> \%        \>     \>                 \>        \> {\it copy A1 A2}  \\
\>                  \>{\it pop\_impl\_point}   \>                        \> \%        \>     \>                 \> {\it )}\\
\>                  \>{\it decr\_universe}     \>                        \> \% {\it ).}\\
\>                  \>{\it deallocate}         \\
\>                  \>{\it proceed}
\end{tabbing}
\caption{Instructions for the last clause of {\it copy}.}\label{fig:copy_2}
\end{figure}

\begin{figure}\footnotesize
\begin{tabbing}
\dquad\=copy\ :\dquad\=switch-on-constant\dquad\=L3,\ L1,\ fail,\ fail\ \dquad\dquad a\= \kill
\>{\it copy}\ : \> {\it try\_me\_else}    \> {\it 0,\ \ L9}   \> \% {\it copy\ }\= \\
\>              \> {\it init\_variable}   \> {\it A3,\ \ Y1}  \\
\>              \> {\it pattern\_unify}   \> {\it A3,\ \ A1}  \> \%             \> {\it c} \\
\>              \> {\it pattern\_unify}   \> {\it A3,\ \ A2}  \> \%             \> {\it c.} \\
\>              \> {\it finish\_unify}    \\
\>              \> {\it proceed}          \\
\>{\it L9}\ :   \> {\it trust\_ext}       \> {\it 2, \ \ 1}
\end{tabbing}
\caption{Instructions for the dynamic clause of {\it copy} in the augment goal in its last definition.}\label{fig:copy_3}
\end{figure}

The instructions {\it switch\_on\_term} and {\it switch\_on\_constant}
in Figure~\ref{fig:copy_1} are used for indexing clause choices in a way
described in Section~\ref{sec:basic_model}.
Specifically, the former
takes the form
\begin{tabbing}
\dquad{\it switch\_on\_term\ \ V,\ \ C,\ \ L,\ \ BV}
\end{tabbing}
where $V$, $C$, $L$ and $BV$ are instruction addresses to which control must
be transferred to when head normal form of the term referred to by $A1$
is a flexible term, a rigid term with a constant head other than $::$,
a nonempty list and a bound variable head respectively. The label {\it fail}
is assumed to be the location of code that causes backtracking.
The other instruction {\it switch\_on\_constant} carries out the second-level
indexing among different constant heads. The first argument of it is a positive
number indicating the number of constants under consideration and the second
argument refers to a hash table in which the mapping from the constants to
the addresses of the corresponding clause definitions are stored.

Among the control instructions appearing in the figures,  {\it try\_me\_else},
{\it retry\_me\_else} and {\it trust\_me} are used for the manipulation
of choice points, and the former two have their second argument being
the address of the clause definition that should be attempted upon
backtracking.
Their first numeric argument is used to indicate the number of argument
registers that are to be saved or retrieved as relevant.
The instructions {\it allocate} and {\it deallocate} are used for the
creation and deletion of environment frames on the stack. The argument
of the former contains a positive number corresponding to the number of
permanent variables that are to be allocated on the frame.
The calls to clause definitions that need to be dynamically determined
are handled by the instructions
{\it call\_name} and {\it execute\_name}, whereas the return from a
clause definition is effected by the instruction
{\it proceed}. The instruction
{\it execute\_name} is specially intended for the last call optimization
mentioned
in Section~\ref{sec:basic_model}. The numeric argument of the call instructions
is used to indicate the number of variables that remain on the
caller's environment frame at the time of the call.
The instruction {\it trust\_ext\ n, i} in Figure~\ref{fig:copy_3} is used
to search for dynamically extended clause definitions in a way described in
Section~\ref{sec:ho_control}.
The first argument
{\it n} is the number of argument registers that should be recovered
before the control is transferred to the found clause definition.
Figure~\ref{fig:copy_2} also illustrates the usages of the higher-order
control instructions {\it push\_impl\_point}, {\it pop\_impl\_point},
{\it incr\_universe} and {\it decr\_universe},
the computations underlying which are described in
Section~\ref{sec:ho_control}.

Following the WAM convention, in the instructions shown in the
figures, we have used the name {\it Yi} to depict the $i$th variable
that is allocated in the environment frame. Also, the instructions
{\it unify\_variable} and {\it put\_value} are identical to the ones
with the same name in the WAM and the instruction {\it set\_value}
is used as a special case of {\it unify\_value} in the WRITE mode.

\section{Treatment of Flexible and Disjunctive Goals}\label{sec:misc}
Up to this point, we have provided a conceptual picture of our
abstract machine and compilation model insofar as these related to
the treatment of higher-order pattern unification. There are two
issues that are remained to be explained. First, the implementation
discussed so far assumes a monomorphic type system for our language,
within which no runtime processing of types is necessary. This
restriction has to be removed in the presence of the first-order
polymorphic types, on which our language is actually based. A
treatment of this aspect is deferred to the next chapter. Second, it
is not clear yet on how the flexible and disjunctive goals are
handled. We discuss these aspects in this section.

The appearance of flexible goals, \ie, of goals of form
$(P\ t_1\ \ldots\ t_n)$, where
$P$ is a variable, embodies the ability to mix in our language
meta and object level usages of predicate expressions.
A predicate definition that exploits this ability is shown below:
\begin{tabbing}
\dquad\dquad\={\it kind}\dquad\dquad\={\it i}\dquad\dquad\={\it type}.\\
            \>{\it type}\>{\it mappred}\>{\it (list i) $\ra$ (i $\ra$ i $\ra$ o) $\ra$ (list i) $\ra$ o. } \\[5pt]
            \>{\it mappred}\>{\it nil P nil.} \\
            \>{\it mappred}\>{\it (X :: L1) P (Y :: L2)} $\pif$ {\it P X Y, mappred L1 P L2.}
\end{tabbing}
Let {\it bob, john, mary, sue, dick} and {\it kate} be constants declared
with type {\it i}, and let {\it parent} be a constant of type
{\it i $\ra$ i $\ra$ o}. Then the following additional clauses
define a ``parent'' relationship between different individuals.
\begin{tabbing}
\dquad\dquad\={\it parent}\dquad\dquad\={\it bob}\dquad\={\it john}.\\
            \>{\it parent}\>{\it john}\> {\it mary}.\\
            \>{\it parent}\>{\it sue}\> {\it dick}.\\
            \>{\it parent}\>{\it dick}\> {\it kate}.
\end{tabbing}
In this context, a query of form
\begin{tabbing}
\dquad\dquad{\it ?-}\quad{\it mappred}\dquad{\it (bob :: sue :: nil)}\dquad {\it parent}\dquad {\it L}
\end{tabbing}
can be asked, and can be solved with the answer
substitution $\{\dg{L}{john::dick::nil}\}$.
Following the operational semantics of our language specified in
Section~\ref{sec:interpreter}, it can be observed that in the course of
solving this query, two new goals
\begin{tabbing}
\dquad\dquad{\it parent  bob  Y1}\dquad and\dquad {\it parent  sue  Y2}
\end{tabbing}
will be dynamically formed and solved. Another example of a query is
\begin{tabbing}
\dquad{\it ?- mappred (bob :: sue :: nil) (x$\plam$ y$\plam$ (Sigma z$\plam$ (parent x z, parent z y))) L}.
\end{tabbing}
This goal asks for the grandparents of {\it bob} and {\it sue} and has as
its solution the substitution $\{\dg{L}{mary::kate::nil}\}$. Finding this
answer requires two new goals of complex structures---each with an
embedded conjunction and existential quantifier---to be constructed
dynamically and then solved.

As illustrated by the {\it mappred} example, flexible goals may be
instantiated by terms containing predicate constants and with complex
logical structures, thereby dynamically
reflecting object-level occurrences of quantifiers and connectives into
positions where they function as search directives.

The problem faced in supporting flexible goals is that instantiations of
their heads can change their structure  dynamically, and so it is
impossible to know at compile time the specific control action that
they would give rise to during computation.
However, we can provide a partial compilation in that we can use the
top level structure of these goals at runtime to pick between
different compiled treatments of control structure. In particular,
flexible goals
can be compiled into calls to a special procedure named {\it solve} to which
(the instantiated version of) the goal is provided as an argument. In the case
that the incoming goal has a complex structure, the behavior of {\it solve}
can be envisaged as of it were based on a compilation of the following
clauses:
\begin{tabbing}
\dquad\dquad\={\it solve (G1 , G2)}\dquad\=$\pif$\quad\={\it solve G1, solve G2}.\\
\>{\it solve (G1 ; G2)}\>$\pif$\>{\it solve G1; solve G2}.\\
\>{\it solve (Sigma G)}\>$\pif$\>{\it solve (G X)}.\\
\>{\it solve (Pi G)}\>$\pif$\>{\it Pi x$\plam$ (solve (G x))}.
\end{tabbing}
When the argument given to {\it solve} is an atomic goal with a rigid head,
then its arguments are loaded into appropriate data registers and the head is
used to determine the code to be invoked subsequently. The only other situation
that could possibly arise is that the actual argument passed to {\it
  solve} remains a flexible atomic goal; the syntactic restriction on
the appearance of logical symbols in terms makes it impossible for
any other case to arise. In this last case---when the argument of {\it
  solve} is a flexible goal---we follow the suggestion in
\cite{NM98Handbook} and solve the goal immediately with a substitution
of the form $\lambda x_1\ldots\lambda x_n\top$ for the variable that
appears as the head of this goal.

In our implementation, the {\it solve} predicate is
treated as a builtin one whose realization is ``hard-wired'' into
the abstract machine.

Our treatment of disjunctive goals is based on a compile-time
pre-processing of clauses to eliminate such disjunctions. Upon seeing
a goal of the form {\it (G1 ; G2)}, the compiler creates a new
predicate definition consisting of the following clauses:
\begin{tabbing}
\dquad\dquad\=
{\it new\_pred  X1 ... Xn}\quad $\pif$\quad {\it G1.}\\
\>
{\it new\_pred  X1 ... Xn}\quad $\pif$\quad {\it G2.}
\end{tabbing}
Here, {\it new\_pred} is a name chosen such that it is distinct from
any other name used in the program and $\{X_1, ..., X_n\}$ is the
set of variables occurring free in {\it (G1 ; G2)}. After generating
and adding these clauses to the program, the compiler replaces the
disjunctive goal with the atomic goal {\it (new\_pred X1 ... Xn)}.
As a concrete example, a clause presented in the form
\begin{tabbing}
\dquad\dquad{\it foo X}\quad $\pif$\quad {\it bar1 U V , (bar2 (f X) U ;  bar3 (f X) V)}.
\end{tabbing}
will be transformed into the sequence of clauses
\begin{tabbing}
\dquad\dquad\={\it foo X}\dquad\dquad\dquad\= $\pif$\quad\= {\it bar1 U V , new\_pred X U V}.\\
\>{\it new\_pred X U V}\>$\pif$\>{\it bar2 (f X) U}. \\
\>{\it new\_pred X U V}\>$\pif$\>{\it bar3 (f X) V}.
\end{tabbing}
by the pre-processing pass just described.

An alternative treatment to disjunctive goals is possible: we could
build in mechanisms for creating choice points in the bodies of
clauses. Thus, in the example just considered, we could use the
following structure to compile the body of the clause for {\it foo}:
\begin{tabbing}
\dquad\=\dquad\= \{\ \ \= Instructions for {\it (bar1 U V)}\ \ \}\\
      \>       \>\>{\it try\_me\_else\_disj L} \\
      \>       \>\{\ \ Instructions for {\it (bar2 (f X) U)}\ \ \} \\
      \>{\it L:}\>\>{\it trust\_me\_disj} \\
      \>       \>\{\ \ Instructions for {\it (bar3 (f X) V)}\ \ \}
\end{tabbing}
Here, the instructions {\it try\_me\_else\_disj} and {\it
  trust\_me\_disj} are like the WAM instructions {\it try\_me\_else}
and {\it trust\_me} except that it is the free variables occurring on
the disjunctive goal that are recorded and used by these instructions
rather than the argument registers. In the above example, instead of
the contents of registers
{\it A1} and {\it A2}, the actual information
recorded in the choice point should be the bindings of the variables
{\it X} and {\it V}.
Notice that we do not need to keep the information about {\it U} in
this example. In general, the compilation process would have to carry
out a ``usefulness'' analysis on the free variables that appear in
disjunctive goals to determine the ones that really have to be
remembered.

Compared with the approach of creating new predicates, this
alternative direct compilation of disjunctive goals has some
advantages. First, it obviates the call to the additional predicate
{\it new\_pred} and consequently avoids the runtime overhead for such
calls. Second, it provides a framework for analyzing which variables
really need to be stored and hence for avoiding redundant
book-keeping.
For these reasons, the direct compilation of disjunctive goals is
something that might be explored further as an improvement to our
implementation ideas.

\chapter{Efficient Support for Runtime Type Processing}\label{chp:types}
The processing model that we have developed for $\lp$ in the previous
chapter has ignored the presence of types in the language and the
impact these might have on computations. This model is accurate if the
language uses a monomorphic type system, \ie, one in which all types
are determined at compile time and do not subsequently
change. However, this is not the true situation in $\lp$ as we have
discussed in Section~\ref{sec:types_in_computation}; $\lp$ uses a
first-order polymorphic type system that leads to the possibility that
the types associated with variables and constants may evolve during
execution. Given this situation, it is important to determine the
exact manner in which the evolution of types may impact on computation
and to take account of this in the processing model. As we shall see
in this chapter, the place at which the identity of types is
needed is in comparing constants. In particular, two constants may
actually share a name but may be different in reality because their
types are distinct and, moreover, do not even have a common
instance. Unification must fail in this situation. To be able to
determine failure, however, it is necessary to bring types along into
the computation at relevant places and to actually check them for
compatibility.

We discuss the impact of polymorphic typing in detail in this chapter
to make the above picture explicit and we develop the needed machinery
for treating types appropriately. In the first section, we indicate
the refinement that is needed to the basic higher-order pattern
unification algorithm from Chapter~\ref{chp:interpreter} to account for
types. A straightforward solution to this problem would simply
construct types at runtime to attach them to constants and to pass
them as additional arguments to predicates. However, types can be
large in practice and constructing them explicitly each time they are
needed can be costly both in time and space. In
Section~\ref{sec:type_skel}, we describe an approach to using
information available at compile time to reduce the type analysis
needed at runtime; this approach has the additional benefit of
reducing the amount of type information that has to be garnered at
runtime. Unfortunately, the approach cannot be used to eliminate type
information to be associated with predicates in some situations when
these are really not necessary. In Section~\ref{sec:pred_type} we
discuss a different form of static analysis that captures these
situations. The work described in Sections~\ref{sec:type_skel} and
\ref{sec:pred_type} has previously been presented in
\cite{NQ05types}. We conclude the chapter by using the approaches we
develop to augment the abstract machine and
compilation structure described in the previous chapter to incorporate
a treatment of types.

Our discussion of the treatment of types pertains only to the
situation where the processing model is based on the use of
higher-order pattern unification. The abstract machine and
compilation model underlying {\em Version 1} of the {\em Teyjus}
system had used Huet's procedure for higher-order
unification. We note that considerably more type information
needs to be carried along and this also needs to be analyzed more
carefully in this situation. The choice we have made in this thesis
has therefore resulted in a significant simplification in the abstract
machine structure along this dimension as well.

\section{Types and Higher-Order Pattern Unification}\label{sec:hopu_type}

The term formation rules presented in Section~\ref{sec:lambda_terms}
associate a type with every well-formed term of $\lp$.
To determine this type, it is important to know the types of all the
constants and (bound) variables that appear in the term. The usual
practice, however, is to not specify types with variables. When we
allow for polymorphic types as in $\lp$, it is possible to infer a
most general type for each term even when the types of (some)
variables have not been provided. We assume such a procedure in our
context. Thus at the end of the compilation phase we assume that every
term has been determined to be type correct and that the type of each
term is also known. In a typical programming language, the usefulness
of types would end at this point. However, this is not the case in
$\lp$ as we have discussed in
Section~\ref{sec:types_in_computation}. In particular, constants
and variables may be used within a term at refinements of their
declared types and such refinements may impact on the precise
computation to be carried out.

Looking naively at the relevance of types to computation, we see that
the abstract interpreter presented in Section~\ref{sec:interpreter}
has use for types in two different forms: first, in rules $4$, $6$,
$7$ and $8$ of Definition~\ref{def:interpreter}, when a logic variable
or a constant is introduced into the computation context, it should have
the same type as the existential or universal variable that is
replaced; second, the unification invoked in rules $7$ and $8$ should be
a typed one. We observe, however, that the types introduced in the
first set of situations do not have a real impact on the steps in
computation. Types are needed in checking identity in unification as
we shall see shortly and, in the case of each of these created
objects, every instance of them share the same type. Thus, when
checking their identity, a simple lookup of the names suffices; the
types would have to match if the names are the same.

The introduction
of types in the higher-order pattern unification can generally be
viewed as
maintaining a type along with every logic variable and constant and
using it to determine computation at necessary points. However, the
types of logic variables are neither examined nor refined in the
process of constructing bindings. Further, the comparison of
constants in this phase are restricted to being between those
appearing as arguments of logic variables in the appropriate
instance of rule $(5)$ in Figure~\ref{fig:hopu_simplification}. The
higher-order pattern constraint requires such constants to have a
larger universe index than
the logic variable as the head, implying thereby that they must have
been introduced by generic goals. Hence every instance of any such
constant must already be known to have the same type. From these
observations, it is evident that types are incidental to the binding
phase of the higher-order pattern unification.

The real substantial usage of types in the pattern unification is in
fact in the simplification phase for determining the applicability
of rule $(4)$ in Figure~\ref{fig:hopu_simplification}: the identity
checking on the rigid heads of the pair of terms may also require
the matching of their types. Observe, however, that if these heads
are matching de Bruijn indexes (abstracted variables) or constants
introduced by generic goals, then the types must already be
identical. Thus the matching or unification of types is necessary
only for the genuinely polymorphic constants declared at the
top-level in the program.

Based on these observations, rule $(4)$ in
Figure~\ref{fig:hopu_simplification} can now be modified into the
following.
\begin{tabbing}
\dquad\=($5.1''$) \=$\dg{\dg{\lambdadb(n, t)}{\lambdadb(m, s)}::\st{D}}{\theta}$ \= $\lra$ \=\kill
\> ($4.1$)\>$\dg{\dg{(c_{\tau}\ t_1\ \ldots\ t_n)}{(c_{\sigma}\ s_1\ \ldots\ s_n)}::\st{D}}{\theta}$ \\
\>      \>\dquad\dquad$\lra$\dquad $\dg{\dg{t_1}{s_1}::\ldots ::\dg{t_n}{s_n}::\st{D}}{\phi\circ\theta}$, \\
\>      \> provided $c$ is a constant such that $\st{L}(c) = 0$ and \\
\>      \>$\phi$ is the most general unifier of $\tau$ and $\sigma$. \\
\> ($4.2$)\>$\dg{\dg{(r\ t_1\ \ldots\ t_n)}{(r\ s_1\ \ldots\ s_n)}::\st{D}}{\theta}$ $\lra$ $\dg{\dg{t_1}{s_1}::\ldots ::\dg{t_n}{s_n}::\st{D}}{\theta}$, \\
\>\> provided $r$ is a constant such that $\st{L}(r)>0$ or a de Bruijn index.
\end{tabbing}
In the rules $(4.1)$ and $(4.2)$, the type association to relevant
constants is represented as a subscript. Further the labeling function
$\st{L}$ of the abstract interpreter is used to help differentiating
between constants from the top-level and those introduced by the
execution of generic goals. Finally, since the polymorphic types in
our language can be essentially viewed as the terms in the
first-order logic, a first-order unification process is assumed to
be invoked on the types of the constant heads in the application of
rule $(4_1)$ to either decide the non-applicability of this rule or
to compute the most general unifier of them. Note also that we might
want to provide the type instantiations back to the user along with
answers. For this reason, we have assumed that our substitutions also
maintain information about the ones made to type variables.

From the above considerations, it is clear that the only sort of
terms with which we need to maintain types at runtime are the
top-level declared constants. Such association of types can be
further reduced to minimize runtime type processing overhead, which
is discussed in the next two sections.

\section{Reducing Type Association for Constants}\label{sec:type_skel}
An obvious solution to making types available with top-level constants
is to add them as a special argument. For example, consider the list
constructors $nil$ of defined type {\it (list A)} and {\it ::} of
defined type {\it (A\ $\ra$\ (list\ A)\ $\ra$\ (list\ A))}. When these
are used in constructing particular lists, the type variable {\it A}
would be instantiated and the resulting type might be added as an
annotation as illustrated by the following terms:
\begin{tabbing}
\dquad\dquad {\it (1\ (::\ int\ $\ra$ (list\ int)\ $\ra$\ (list\ int))\ (nil\ list\ int))}\dquad and \\
\dquad\dquad {\it ("a"\ (::\ string\ $\ra$ (list\ string)\ $\ra$\ (list\ string))\ (nil\ list\ string))}.
\end{tabbing}

This solution is adequate but also contains redundant information. The
declaration of a top-level constant ensures that the type of every
occurrence of the constant in the program has a common skeleton part
that is known at compile-time and that differences arise
between the types of distinct occurrences of that constant only in the
instantiations of variables occurring in the skeleton. Thus, the type
of each
legitimate occurrence of {\it ::} must have a skeletal structure {\it
  (A\ $\ra$\ (list\ A)\ $\ra$\ (list\ A))} that is further refined by
an instantiation for {\it A}. This information can be exploited by
avoiding the construction at runtime of the skeleton that often is the
most complex part of the type. Moreover, compile-time type checking
also ensures that two different occurrences of {\it ::} share this
skeletal structure. Hence the matching of their types can be achieved
simply by matching the particular instantiations of the variable {\it
  A}.

We use the idea above by changing the annotation associated with each
top-level constant from a complete type to a list of types that
instantiate the variables that occur in its skeleton; the annotation
must now be a list of types because there could be more than one
variable appearing in the skeleton. Concretely, the representations of
the two lists considered earlier in this section now become
\begin{tabbing}
\dquad {\it (1\ (::\ [int])\ (nil\ [int]))}\dquad and\dquad {\it ("a"\ (::\ [string])\ (nil\ [string]))}.
\end{tabbing}

Based on this annotation scheme, we modify the transformation rules
$(4.1)$ and $(4.2)$ used in unification to the following:
\begin{tabbing}
\dquad\=($5.1''$) \=$\dg{\dg{\lambdadb(n, t)}{\lambdadb(m, s)}::\st{D}}{\theta}$ \= $\lra$ \=\kill
\>      ($4.1'$)\> $\dg{\dg{(c\ [\tau_1, \ldots, \tau_m]\ t_1\ \ldots\ t_n)}{(c\ [\sigma_1, \ldots, \sigma_m]\ s_1\ \ldots\ s_n)}::\st{D}}{\theta}$   \\
\>      \>\dquad\dquad$\lra$\dquad$\dg{\dg{t_1}{s_1}::\ldots ::\dg{t_n}{s_n}::\st{D}}{\phi\circ\theta}$, \\
\>      \> where $\phi$ is the most general unifier for $\{\dg{\tau_1}{\sigma_1}, \ldots, \dg{\tau_1}{\sigma_1} \}$, \\
\>      \> $n \geq 0$ and $m \geq 0$, if $c$ is a constant.\\
\>      ($4.2'$)\>$\dg{\dg{(r\ t_1\ \ldots\ t_n)}{(r\ s_1\ \ldots\ s_n)}::\st{D}}{\theta}$ $\lra$ $\dg{\dg{t_1}{s_1}::\ldots ::\dg{t_n}{s_n}::\st{D}}{\theta}$, \\
\>\> provided $r$ a de Bruijn index.
\end{tabbing}
Notice that the type annotation for a monomorphic constant, \ie, a
constant whose declared type does not contain variables, and for
a constant introduced by a generic goal is an empty list. These cases
are then uniformly handled
by rule ($4.1'$) as the case where $m = 0$.

The manner in which unification problems are processed actually
allows for a further refinement of type annotations. The use of
the transformation rules in Figure~\ref{fig:hopu_simplification}
begins with a pair of atomic predicates whose heads will first have to
be verified to have the same name and whose types will have to be
matched; the matching of the types can be achieved by adding the
instantiations of the type variables in the skeleton type as explicit
arguments to the predicate and then compiling unification of these
types as we shall see shortly. Once we have checked the matching of
these types, we will then be assured that the actual argument terms
that have to be unified have the same types. Further the
unification transformation rules preserves this relationship between
the terms in
each disagreement pair. Thus, at the time when the types of
different instances of a constant are being unified in the rule
$(4.1')$, their target types are known to be
identical. This fact implies that once we have checked that the
constants heading the two terms have a common name,
there is no need to
perform unification over the instances of type variables that appear
in the target type of their type skeleton. In the case that all
the variables in the declared type also appear in the target type,
\ie, when the constant type satisfies what is known as the {\it type
preservation
property}~\cite{H89lp}, there is really no need to maintain any
type annotations with the constant. This happens to
be the case for both {\it ::} and {\it nil}, for instance, and so
all type information can be elided from lists that are implemented
using these constants. A further observation that can be made is
that when the disagreement pair under consideration consists two
constants only, their types are guaranteed to be identical already,
so that type unification can be completely eliminated in this case.
This leads to the final form of the transformation rules for
simplifying rigid-rigid pairs that we present in
Figure~\ref{fig:typed_simple}.
\begin{figure}
\begin{tabbing}
\dquad\=($5.1''$) \=$\dg{\dg{\lambdadb(n, t)}{\lambdadb(m, s)}::\st{D}}{\theta}$ \= $\lra$ \=\kill
\>      ($4.1''$)\> $\dg{\dg{(c\ [\tau_1, \ldots, \tau_m]\ t_1\ \ldots\ t_n)}{(c\ [\sigma_1, \ldots, \sigma_m]\ s_1\ \ldots\ s_n)}::\st{D}}{\theta}$   \\
\>      \>\dquad\dquad$\lra$\dquad$\dg{\dg{t_1}{s_1}::\ldots ::\dg{t_n}{s_n}::\st{D}}{\phi\circ\theta}$, \\
\>      \> where $\phi$ is the most general unifier for $\{\dg{\tau_1}{\sigma_1}, \ldots, \dg{\tau_1}{\sigma_1} \}$, \\
\>      \> $n > 0$ and $m \geq 0$, if $c$ is a constant.\\
\>      ($4.1''$)\> $\dg{\dg{c\ [\tau_1, \ldots, \tau_m]}{c\ [\sigma_1, \ldots, \sigma_m]}::\st{D}}{\theta}$ $\lra$ $\dg{D}{\theta}$, \\
\>      \> where $m \geq 0$, if $c$ is a constant.\\
\>      ($4'_2$)\>$\dg{\dg{(r\ t_1\ \ldots\ t_n)}{(r\ s_1\ \ldots\ s_n)}::\st{D}}{\theta}$ $\lra$ $\dg{\dg{t_1}{s_1}::\ldots ::\dg{t_n}{s_n}::\st{D}}{\theta}$, \\
\>\> provided $r$ a de Bruijn index.
\end{tabbing}
\caption{The type annotated simplification rules for pattern unification.}\label{fig:typed_simple}
\end{figure}

We now consider the correctness of the rules in
Figure~\ref{fig:typed_simple} relative to the original rule for
simplifying rigid-rigid pairs. We begin with the assumption that the
two terms in any disagreement pair considered by the transformation
rules for unification have the same types. It is easy to see then
that this property is preserved by the transformation rules in
Figure~\ref{fig:hopu_simplification}. The first refinement to rule
$(4)$, \ie, the one contained in the rules $(4.1)$ and $(4.2)$, is
easily seen to be
correct once we note that the identity of a constant is determined
also by its type. The correctness of the subsequent refinements to
this rule that lead to the rules in Figure~\ref{fig:typed_simple} then
relies on the facts that, given two rigid terms of equal types that
have a constant with the same name as their heads, unifying the
instantiations of the variables that appear only in the argument types
of the constant head in its two different occurrences will ensure
that the types of these occurrences are equal and, furthermore, will
make the types of the arguments in the two rigid
terms also equal. The following theorem shows this to be the case.

\begin{theorem}\label{thm:type_const}
Let $c$ be a constant that has as its type skeleton the type $\alpha$
with $n$ argument types. Further, let $\{U_1,\ldots,U_k\}$ be
the set of variables that appear in the target type of $\alpha$ and
let $\{V_1,\ldots,V_l\}$ be the variables that appear only
in the argument types of $\alpha$. Now suppose that $(c\ t_1\ \ldots\
t_n)$ and $(c\ s_1\ \ldots\ s_n)$ are two terms that have the same
type $\beta'$ and let $\alpha_1$ and $\alpha_2$ be the type of $c$ in
these two terms. Obviously, $\alpha_1$ and $\alpha_2$ are generated by
applying substitutions to $\alpha$. We assume that any variables
appearing in the ranges of these substitutions are fresh, \ie, they
have not been used previously in the computation. Let
\begin{tabbing}
\dquad\dquad $\phi_1=\{\dg{V_i}{r^1_i} \vert 1\leq i\leq l\}$\quad and\quad
$\phi_2=\{\dg{V_i}{r^2_i}\vert 1\leq i\leq l\}$
\end{tabbing}
be the restrictions of these respective substitutions to the variables
appearing only in the argument types of $\alpha$. Then $\alpha_1$ and
$\alpha_2$, the types of $c$ in the two terms, are unifiable by a
substitution $\theta$ if and only if $\theta(r^1_i)
= \theta(r^2_i)$ for $1\leq i\leq l$. Moreover, any $\theta$
satisfying this property makes the types of $t_i$ and $s_i$ identical
for $1 \leq i \leq n$.
\end{theorem}

\begin{proof}
Any substitution $\theta$ that unifies $\alpha_1$ and $\alpha_2$ makes
the argument types of $c$ in the two terms identical. This is the same
as saying that the types of the arguments of $c$ must be identical
under the substitution. Thus, it only remains to show that $\theta$
unifies $\alpha_1$ and $\alpha_2$ if and only if the condition
mentioned in the theorem is satisfied.

Restricting attention to only the variables appearing in $\alpha$, the
substitutions that produce $\alpha_1$ and $\alpha_2$ from $\alpha$ can
be partitioned into substitutions for the variables
$\{U_1,\ldots,U_k\}$ and the substitutions $\phi_1$ and $\phi_2$
respectively. Moreover, since the target types of $\alpha_1$ and
$\alpha_2$ are identical, the former substitution can be assumed to be
the same in both cases. Let us take it to be $\phi$. By assumption,
the domains of $\phi_1$ and $\phi_2$ do not contain any variables in
the range of $\phi$. Thus, we may write $\alpha_1$ and $\alpha_2$ as
$\phi_1(\phi(\alpha))$ and $\phi_2(\phi(\alpha))$, respectively. Now,
for any unifier $\theta$ of $\alpha_1$ and $\alpha_2$ we have the
following:
\begin{tabbing}
\dquad\dquad\=$\Longleftrightarrow$\dquad\= \kill
            \>                          \>$\theta(\alpha^1) = \theta(\alpha^2)$\\
            \>$\Longleftrightarrow$     \>$\theta(\phi_1(\phi(\alpha))) = \theta(\phi_2(\phi(\alpha)))$   \\
            \>$\Longleftrightarrow$
            \>$(\theta\circ\phi_1)(\phi(\alpha)) =
            (\theta\circ\phi_2)(\phi(\alpha))$
\end{tabbing}
Since the range of $\phi$ does not contain $V_1,\ldots,V_l$, it is
easy to see that the last condition holds if and only if $\theta \circ
\phi_1(V_i) =
\theta \circ \phi_1(V_i)$ for $1 \leq i \leq l$. But this clearly
holds if and only if $\theta(r^1_i) = \theta(r^2_i)$ for $1\leq
i\leq l$.
\end{proof}

The ideas we have described may be applied to the {\it append}
program appearing in Section~\ref{sec:types_in_computation}. In the
type skeleton of the predicate constant {\it append}, {\it (list\ A)
$\ra$ (list\ A) $\ra$ (list\ A) $\ra$ o}, the type variable $A$
appears in the argument types but not in the target. For this
reason, the binding of $A$ should be associated with the occurrences
of {\it append}. We have already seen that type annotations are
dropped from {\it ::} and {\it nil}. Thus the definition of {\it
append} is viewed as the following in our implementation.
\begin{tabbing}
\dquad{\it append [A] nil L L.}\\
\dquad{\it append [A] $($X :: L1 $)$ L2 $($X :: L3$)$ }$\pif\ ${\it append [A] L1 L2 L3}.
\end{tabbing}
Correspondingly, a query of form {\it (append (1 :: nil) (2 :: nil) L)} becomes
\begin{tabbing}
\dquad{\it append [int] (1 :: nil) (2 :: nil) L}.
\end{tabbing}

The final point to be noticed with regard to our type annotation
scheme is that it is capable also of dealing with the situations
where the type preservation property is violated. For example,
consider a representation of heterogenous list base on the constants
{\it null} and {\it cons} declared as the following.
\begin{tabbing}
\dquad\={\it kind}\dquad\= {\it lst}\dquad\dquad\={\it type.} \\
      \>{\it type}      \> {\it null}           \>{\it lst.} \\
      \>{\it type}      \> {\it cons}           \>{\it A $\ra$ lst $\ra$ lst.}
\end{tabbing}
The list containing the integer $1$ and the string
``list" as its elements would then be represented by the term
\begin{tabbing}
\dquad{\it (cons [int] 1 (cons [string] "list" null))}.
\end{tabbing}
Further, the unification of this term with another term representing a
list would naturally involve unifying the type arguments of {\it cons}
which, by Theorem~\ref{thm:type_const}, would achieve the effect of
checking that the relevant occurrences of {\it cons} actually are (or
can be made) identical.

\section{Reducing Type Annotations with Clauses}\label{sec:pred_type}
None of the type variables appearing in the type of a predicate
constant can appear in its target type since this type is $o$.
Thus it is not possible to use the ideas in the previous section to
drop the annotation corresponding to any of these variables. Despite
this, it can be observed that the bindings for some of the variables
appearing in the heads of clauses defining certain predicates
cannot have any impact on the computation. As a particular example,
consider the predicate $append$, an annotated version of whose
definition was presented at the end of the last section. Since the
annotation does not refine the declared type of $append$ in either of
these clauses, the particular type of $append$ in any well-formed goal
that has this predicate as its head will not be the cause for failure
in head unification. Moreover, the instantiation of this variable only
gets used in the annotation of a recursive call to append where,
by the same analysis, it again cannot cause failure in
unification. Thus, if we maintain an annotation for this type variable
with the clauses for $append$, we would be creating a possibly complex
type term only for the purpose of passing it on from recursive call to
recursive call.

To eliminate the redundant type associations with clause definitions,
we describe in this section a systematic process for determining
the elements of the types list associated with a predicate name that
could potentially influence a computation. For the types not in this
list we can conclude that they can be elided.

The process of determining the potentially ``needed'' elements in
the types list is organized around the full set of clauses defining
the predicate constant, including those contained by augment goals.
If the definition of a predicate can be dynamically extended, \ie,
if there are clauses for the predicate embedded in augment goals, we
assume every element in the types list of the predicate is needed:
specific bindings for
type variables appearing in the embedded clause might be determined
when the enclosing clause is used in a backchaining step, and then
these types will be needed in determining the applicability of the
clause.
For a predicate all of whose clauses appear only at
the top-level, our analysis can be more sophisticated. An element in
the types list of the predicate being defined is needed if the value
in the relevant position in the list associated with the particular
predicate constant occurrence at the clause head is anything other
than a variable: unification over this element must be attempted
during clause selection since it has the possibility of failing in
this case. Another situation in which the element is needed is if it
is a type variable that occurs elsewhere in the same types list or
in the type lists associated with a non-predicate constant that
occurs in the clause. The rationale here is that either the variable
will already have a binding that must be tested against an incoming
type or a value must be extracted into it that is used later in a
unification computation of consequence. A more subtle situation for
the variable case is when it occurs in the types list associated
with the predicate head of a clause contained by an augment goal in
the body. In this case the binding that is extracted at runtime in
the variable has an impact on the applicability of the clause that
is added and consequently is a needed one.

The only case that remains to be considered is that where a variable
element in the types list for the clause head appears also in the
types list associated with a predicate constant in a goal position
in the body, either at the top-level or, recursively, in an embedded
clause definition. It can be observed that a precise neededness
information for the head predicate can be determined only after
those of the body predicates are available. For this reason, our
analysis in this case first determines the neededness information for
the predicate constants appearing at the heads of goals in the body
and then uses this information in the analysis for
the predicate that is being defined by the clause. As an example of
how this might work,
consider the following program annotated in the style of
Section~\ref{sec:type_skel}.
\begin{tabbing}
\dquad\=\kill
\>{\it type}\dquad\={\it print}\dquad\dquad\={\it A $\ra$ o}.\\
\>{\it type}\>{\it print\_list}\>{\it (list A) $\ra$ o}.\\
\\
\>{\it print [int] X :- \{code for printing the integer value bound to X\}.}\\
\>{\it print [string] X :- \{code for printing the string value bound to X\}.}\\
\\
\>{\it printlist [A] nil.}\\
\>{\it printlist [A] (X::L) :- print [A] X, printlist [A] L.}
\end{tabbing}
In this code, {\it print} is a predicate that is defined to be
polymorphic in an {\it ad hoc} way and consequently has genuine use
for its type argument. This information can be
used to determine that it needs its type adornment and the following
analysis exposes the fact that {\it printlist} must therefore carry
its type annotation.

\begin{figure}\footnotesize
\begin{tabbing}
\dquad\={\it find\_needed($\st{P}$)}\quad $\{$ \\
      \>\dquad\= {\it init\_needed($\st{P}$)}; \\
      \>      \> repeat  \\
      \>      \>\dquad\=for each top-level non-atomic clause $C$ in $elab(\st{P})$ \\
      \>      \>      \>\dquad{\it process\_clause($C$)}; \\
      \>      \> until (the value of {\it needed} does not change) \\
      \>$\}$ \\
      \\
      \>{\it init\_needed($\st{P}$)}\quad $\{$ \\
      \>      \>for every embedded clause $C$ in $elab(\st{P})$ with $(p\ [\tau_1,\ldots,\tau_k]\ t_1\ \ldots\ t_n)$ as head\\
      \>      \>     \> for $1\leq i\leq k$ \\
      \>      \>     \>\dquad\= {\it needed[p][i] = true} \\
      \>      \>for every top-level clause $C$ in $elab(\st{P})$ with $(p\ [\tau_1,\ldots,\tau_k]\ t_1\ \ldots\ t_n)$ as head\\
      \>      \>     \> for $1\leq i\leq k$ \\
      \>      \>     \>      \> if $\tau_i$ is not a type variable \\
      \>      \>     \>      \>\dquad\={\it needed[p][i] = true;} \\
      \>      \>     \>      \> else \\
      \>      \>     \>      \>      \> if (\=($\tau_i$ occurs in $\tau_j$ for some $j$ such that $1\leq j\leq k$ and $i \neq j$) or \\
      \>      \>     \>      \>      \>     \>($\tau_i$ occurs in  the types list of a non-predicate constant in $C$) or\\
      \>      \>     \>      \>      \>     \>(\=$\tau_i$ occurs in the types list of a predicate constant appearing \\
      \>      \>     \>      \>      \>     \> \>as the head of an embedded clause in the body of $C$))\\
      \>      \>     \>      \>      \>     \> \>\dquad{\it needed[p][i] = true;}\\
      \>$\}$
\end{tabbing}
\caption{The top-level control for determining if a predicate type argument is needed.}\label{fig:pred_const_top}
\end{figure}

\begin{figure}\footnotesize
\begin{tabbing}
\dquad\={\it process\_clause($C$)} \{ \\
      \>\dquad\= let $C$ be of the form {\it (p [$\tau_1, \ldots, \tau_k$] $t_1$ $\ldots$ $t_n$ $:-$ G).}\\
      \>      \> for $1 \leq i \leq k$\\
      \>      \> \dquad\=if {\it needed[p][i]} is {\it false} \\
      \>      \>       \> \dquad\= {\it needed[p][i] = process\_body(G, $\tau_i$)}\};\\
      \>\}\\
      \\
      \>{\it process\_body(G, $\tau$)} : boolean \{\\
      \>      \>switch on the top-level structure of $G$: \\
      \>      \>\quad\= $\allx{G'}$, $\somex{G'}$:\quad\= return {\it process\_body($G'$, $\tau$)};\\
      \>      \>      \>$G_1 \conj G_2$, $G_1 \disj G_2$:\quad return ({\it process\_body($G_1$, $\tau$)} or {\it process\_body($G_2$, $\tau$)}); \\
      \>      \>      \>$D\imp G$:                \> return ({\it process\_body($G$, $\tau$)} or {\it process\_embedded\_body($D$, $\tau$)}); \\
      \>      \>      \>$A$ of the form ($q\ [\sigma_1,...,\sigma_l]\ s_1\ ...s_m$): \\
      \>      \>      \>\dquad\=if $\tau$ occurs in $\sigma_i$ for some $i$ such that $1 \leq i \leq l$ and {\it needed[q][i]} is {\it true}\\
      \>      \>      \>      \>\dquad\=return {\it true}; \\
      \>      \>      \>      \>else \\
      \>      \>      \>      \>      \>return {\it false};\\
      \>\} \\
      \\
      \> {\it process\_embedded\_body($D$, $\tau$)} : boolean \{\\
      \>      \>switch on the top-level structure of $D$: \\
      \>      \>     \>$\allx{D_1}$:\dquad\= return {\it
        process\_embedded\_body($D_1$, $\tau$)};\\
      \>      \>     \>$D_1 \land D_2$:   \> return {\it
        process\_embedded\_body($D_1$, $\tau$) or process\_embedded\_body($D_2$, $\tau$)};\\
      \>      \>     \>$G\imp A$:         \> return {\it process\_body($G$, $\tau$))};\\
      \>      \>     \>$A$:               \>return {\it false};\\
      \> \}
\end{tabbing}
\caption{The clause processing for determining if a predicate type argument is needed.}\label{fig:pred_const}
\end{figure}

The approach suggested above needs refinement to be applicable to a
context where dependencies between definitions can be iterated and
even recursive; at present, it doesn't apply directly even to the
definition of {\it append}. The solution is to use an iterative,
fixed-point computation that has as its starting point the
neededness information gathered by initially ignoring predicate
constants appearing in goal positions in the body of the clause. In
effecting this calculation relative to a given program
$\mathcal{P}$, we employ a two-dimensional global boolean array
called {\it needed} whose first index, $p$, ranges over the set of
predicate constants appearing in $\mathcal{P}$ and whose second
index, $i$, is a positive integer that ranges over the length of the
types list for $p$; this array evidently has a variable size along
its second dimension. The intention is that if, at the end of the
computation, {\it needed}$[p][i]$ is {\it false} then the $i$th
element in the types list associated with $p$ does not have an
influence on the solution of any goal $G$ from $\mathcal{P}$. We
compute the value of this array by initially setting all the
elements of {\it needed} to {\it false} and then calling the
procedure {\it find\_needed} defined in
Figure~\ref{fig:pred_const_top} and Figure~\ref{fig:pred_const} on
the program $\mathcal{P}$.

There are only finitely many elements in the {\it needed} matrix for
any program $\mathcal{P}$ and,
from this, it is clear that the invocation of {\it find\_needed} must
always terminate. Theorem~\ref{thm:pred_const} below shows that, when
it does terminate, it provides us a conservative estimate of the type
annotations that have a role to play in computation. Using this
theorem, we see that we can correctly eliminate those type variable
locations from clause and goal heads that are determined not to be
needed for any given predicate by this procedure.

\begin{theorem}\label{thm:pred_const}
Let $p$ be a predicate constant defined in $\mathcal{P}$ and let it
be the case that when ${\it find\_needed(\mathcal{P})}$ terminates,
{\it needed[p][i]} is set to {\it false}. Then the $i$th element in
the types list of $p$ has no impact on the solvability of any goal
$G$ from $\mathcal{P}$.
\end{theorem}

\begin{proof}
We shall prove the contrapositive form of the theorem: if the
solvability of $G$ from $\mathcal{P}$ is dependent on the $i$th
element of the types list of a predicate $p$, then {\it needed[p][i]}
must be set to {\it true} by ${\it find\_needed(\mathcal{P})}$.

From an examination of Definitions~\ref{def:interpreter} and \ref{def:interp},
it can be seen that the $i$th element of the types list of $p$ affects
the computation resulting from $G$ relative to $\mathcal{P}$ only if there
is a sequence of atomic formulas of the form $A_1, \ldots, A_n$ with
$A_1$ having the predicate $p$ as its head and there is a sequence
$D_2, \ldots, D_n$ of clauses in the elaboration of $\mathcal{P}$
augmented with type instances of embedded clauses in $\mathcal{P}$ and a
sequence of positive numbers $j_1,\ldots,j_n$ such that
\begin{enumerate}
\item for $1 < i \leq n$, $A_{i-1}$ is an instance of the head of
  $D_i$ and $A_i$ appears as a goal in the body of that instance of
  $D_i$,
\item for $1 < i \leq n-1$, the $j_{i}$th type argument in the head of
  $D_i$ is a variable and, further, it appears in the $j_{i+1}$th type
  argument of the goal in the body of $D_i$ that has $A_i$ as its
  instance,
\item $j_1 = i$, and
\item the $j_n$th type argument of $A_n$ directly affects computation
  either because it has to be unified with a non-variable type
  argument in the head of $D_n$ or because its value imposes a
  structure requirement on some other type argument of the head or on
  the type of an embedded clause or of a constant appearing in a place
  different from the head of an atomic goal in the body.
\end{enumerate}
Letting $p= p_1, \ldots, p_n$ be the predicate heads of the goals in
the sequence $A_1,\ldots,A_n$, we claim that {\it find\_needed} will
result in {\it   needed[$p_i$][$j_i$]} being annotated to {\it true}
for $1 \leq i \leq n$. The desired conclusion follows from this.

We prove the claim by a backwards induction on the sequence.

For the base case, an inspection of the procedure
{\it init\_needed} shows that the possibilities described for the
$j_n$ type argument impacting on the computation can arise only in the
situations in which this procedure causes {\it
  needed[$p_n$][$j_n$]} to be marked {\it true}; the only slightly
tricky situation is that where $D_n$ is a type instance of an embedded
clause but this is handled by noting that {\it needed[$p_n$][$k$]} is
marked {\it true} for {\it all} $k$ in this case. Noting that once an
entry in the {\it needed} matrix has been marked {\it true}, this
marking persists through the rest of the computation of {\it
  find\_needed} then concludes the argument.

Assume now that the claim is true for the sequence
$p_{k+1},\ldots,p_n$. This means in particular that {\it
  needed[$p_{k+1}$][$j_{k+1}$]} must be marked true. If $A_k$ is an
instance of a clause in $elab(\st{P})$, then an
inspection of the procedures {\it process\_clause} and {\it
  process\_body} shows that {\it needed[$p_k$][$j_k$]} must also be
marked {\it true} during some iteration of the loop in {\it
  find\_needed}. If $A_k$ is an instance of a type instance of an
embedded clause on the other hand, then {\it init\_needed} will mark
{\it needed[$p_k$][$j_k$]} {\it true} as a special case of marking
{\it needed[$p_k$][$l$]} {\it true} for all $l$. Since a {\it true}
annotation persists in the computation of {\it find\_needed}, the
claim follows for the sequence $p_k,\ldots,p_n$, thus completing the
inductive argument.

\end{proof}

As a particular example of the use of this theorem, we observe that
the type list argument for the version of {\it append} shown in the
last section can be eliminated, thus reducing the definition of this
predicate that needs to be used at runtime to what is essentially the
untyped form. More generally, if every type argument for the head
predicate of a clause is a variable---a property called {\it type
  generality} in \cite{H89lp}---and every constant is type preserving
and there are no embedded clauses, then types can be eliminated
entirely during computation.

\section{Low-Level Support for Types and their Compilation}\label{sec:compile_type}
We now can consider the integration of the runtime processing of
types into our abstract machine based on our annotation scheme.

The first issue to be solved is the low-level representations of
types. As already mentioned, the types in the $\lp$ language can be
essentially viewed as first-order terms. This allows us to use
the usual encoding of first-order terms in the WAM for types in $\lp$.
In particular, a memory cell is used for each type with a tag
indicating its category as one of type variable, type constant and
type structure. For a type variable, the category tag is the only
important information to be maintained. For a type constant, a
reference to its descriptor is kept along with the tag. The
additional information with a type structure consists of a reference
to a sequence of cells in which the first corresponds to the type
constructor of a fixed arity and the subsequent ones, in the number
given by the arity, to the arguments.

The association of types with (term) constants is realized as the
following. A new class of constants is introduced to the term
representation described in Section~\ref{sec:internal_encoding} as
those with runtime type annotations. The only extra information
maintained with a constant of this sort is a reference to a type
environment that contains the elements in the types list of the
constant decided by the compiler in the way described in
Section~\ref{sec:type_skel}. The size of this type environment is
stored along with the constant descriptor.

The usages of the data areas of our abstract machine are also
extended. First, the heap and the stack are used to store types in
addition to terms. Second, the bindings of type variables are also
trailed whenever it is necessary to do so. Further, the PDL is also
used in the course of type unifications invoked in an interpretive
mode.
Finally, the data registers $A_1$ to $A_n$ can be used to refer to a
type, and an additional register {\it TS}---similar to the register
{\it S} for terms---is used for the decomposition of type structures.

Compilation treatment of type unification is also provided by our
implementation. Essentially, such computation can be encountered in
the following two situations. First, it can be the result of unifying
the types list of a predicate constant appearing as a clause head with
the types appearing appearing in an actual goal. Second, it could be
required during term unification when the types of two occurrences of
the a constant of the same name have to be checked for
compatibility. In both cases, the
elements in the types list are viewed as additional arguments of the
given constant and are handled by the conventional {\it get} and
{\it unify} instructions respectively.

We consider the compilation of the definition of {\it printlist}
provided in the previous section to illustrate  the use of type
unification instructions to handle the types argument of a predicate
constant.
The instructions generated for the clause
\begin{tabbing}
\dquad\dquad{\it printlist [A] (X::L) :- print [A] X, printlist [A] L.}
\end{tabbing}
take the following structure.
\begin{tabbing}
\dquad\dquad\= {\it get-structure-}\dquad\={\it A22,}\ \={\it A22,}\ \={it A22} \= \kill
\>{\it allocate}\>{\it 3} \\
\>{\it get\_type\_variable}\>{\it Y1,}\>{\it A2}\>\dquad\dquad\=\%\quad\= {\it Y1 = A2 = A}\\
\>{\it get\_list}\>{\it A1}                    \>            \>\>\%\> {\it A1 = (:: }\= \\
\>{\it unify\_variable}\>{\it A1}              \>            \>\>\%\>                \>{\it X}\dquad\= {\it (A1 = X)}  \\
\>{\it unify\_variable}\>{\it Y2}              \>            \>\>\%\>                \>{\it L)}     \> {\it (Y2 = L)} \\
\>{\it finish\_unify}\\
\>{\it put\_type\_value}\>{\it Y1,}\>{\it A2}   \>            \>\%\> {\it A2 = A} \\
\>{\it call\_name}\>{2,}\>{\it print}           \>            \>\%\> {\it print A1 A2} \\
\>{\it put\_type\_value}\>{\it Y1,}\>{\it A2}   \>            \>\%\> {\it A2 = A} \\
\>{\it put\_value}\>{\it Y2,}\>{\it A1}         \>            \>\%\> {\it A1 = L} \\
\>{\it deallocate}\\
\>{\it execute\_name}\>{\it print\_list}       \>            \>\>\%\> {\it print\_list A1 A2}
\end{tabbing}
The instructions  {\it get\_type\_variable}, {\it
put\_type\_variable} and {\it put\_type\_value} used here correspond
to the {\it get\_variable}, {\it put\_variable} and {\it
put\_value} instructions of the WAM. From the compiled form, it should
be evident that the variable $A$ in
the types lists of {\it print} and {\it print\_list} is treated as
an additional argument of these predicate constants.

To deal with the situation where it is necessary to compile the
matching with a constant that has a non-empty types list associated
with it, new instructions are
introduced to transit from term unification to type
unification. One of these instructions is
\begin{tabbing}
\dquad\dquad{\it get\_typed\_structure $Ai$, f, n}
\end{tabbing}
that is a variant of the {\it get\_structure} instruction that is used
for compiling a first-order application term whose constant
head has type associations. The action underlying this instruction
differs from its ``untyped'' version in the manipulation of the
constant head given by $f$. If it is in the situation where $f$
should be created on the heap, a typed constant cell is constructed
with an empty type environment and set to be referred to by the
register {\it TS} with the assumption that this type environment
will be filled in by the execution of the subsequent {\it
unify\_type} instructions in the WRITE mode. Alternatively, if the
term referred to by $Ai$ is a first-order application of head $f$,
the {\it TS} register is set to refer to its type environment, and
it is assumed that the actual unification against the types in the
environment will be carried out by the following {\it unify\_type}
instructions executed in the READ mode.

For a concrete example, assume we have a kind {\it pr} corresponding
to the set of tuple types. Further, assume the constants {\it pair}
and {\it first} are used to denote functions returning a pair
consisting of the given two arguments and returning the first
argument of the given pair respectively.
\begin{tabbing}
\dquad\dquad\=\kill
\>{\it kind}\dquad\={\it pr}\dquad\dquad\={\it type $\ra$ type $\ra$ type.}\\
\>{\it type}\>{\it pair}\>{\it A $\ra$ B $\ra$ (pr A B)}.\\
\>{\it type}\>{\it first}\>{\it (pr A B) $\ra$ A}.
\end{tabbing}
Then the compilation of the term {\it (first [B] (pair X Y))}
appearing in a clause head results in the following sequence of
instructions:
\begin{tabbing}
\dquad\dquad\= {\it get-typed-structure-}\dquad\={\it A2222,}\ \={\it A2222,}\ \={it A} \= \kill
\>{\it get\_typed\_structure}\>{\it A1,}\>{\it first,}\>{\it 1}\>\%\ \= {\it A1 = (first }\= \\
\>{\it unify\_type\_variable}\>{\it A2}\>\>\>\%\>\>                            {\it [B]}\\
\>{\it unify\_variable}\>{\it A3}\>\>\>\%\>\>                                  {\it A3)}\\
\>{\it get\_structure}\>{\it A3,}\>{\it pair,}\>{\it 2}\>\%\>{\it A3 = (pair }\= \\
\>{\it unify\_variable}\>{\it A4}\>\>\>\%\>\>{\it X}\\
\>{\it unify\_variable}\>{\it A5}\>\>\>\%\>\>{\it Y)}
\end{tabbing}
The instruction {\it unify\_type\_variable} used above corresponds to
the {\it unify\_variable} instruction in the WAM.

Typed variants of the {\it get\_constant} and {\it unify\_constant}
instructions are also included. These are specifically the following:
\begin{tabbing}
\dquad{\it get\_typed\_constant Ai, c, L}\dquad\quad and\dquad\quad {\it unify\_typed\_contant c, L}.
\end{tabbing}
As in the case of {\it get\_typed\_structure}, when the constant
$c$ is created by these instructions, a typed constant cell
associated by an empty type environment referred to by $TS$ is
constructed. However, in the situation when the term referred to by
$Ai$ is a constant of the same name, the elements in
the types lists of the two instances of $c$ must already be
identical, and so unifications over them can be safely elided. For
this purpose, an additional argument $L$ is used in these
instructions to indicate the address of the instruction immediately
following those for constructing the types list of $c$, so that
execution can jump to the location $L$ in the described situation.

Additions are also made in the {\it put} and {\it set} classes of
instructions to support the creation of typed constants in a
similar manner to that in the {\it get} and {\it unify} classes
described above. Specifically, the new instructions
\begin{tabbing}
\dquad {\it put\_typed\_constant Ai, c}\dquad\quad and\dquad\quad {\it
  set\_typed\_constant c}
\end{tabbing}
are added. Moreover, since the {\it put} and {\it set} instructions for term
creation could interleave with those in the {\it get} and {\it
unify} classes for the purpose of solving the higher-order part of
unification in an interpretive manner, the usages of {\it put\_type}
and {\it set\_type} instructions are also extended to a clause head.

The last issue to be clarified with regard to types is about the
treatment of the types argument of a constant when it is used both
as predicate and non-predicate in a program: when appearing as a
predicate, the types argument of the constant may be further
reduced, making the number of types argument of such an occurrence
of the constant inconsistent with that of its non-predicate
occurrence. This phenomenon can be illustrated by the following
example, which defines the meta-level application of binary
functions.
\begin{tabbing}
\dquad\quad\={\it type}\dquad\={\it apply}\dquad{\it (A $\ra$ A $\ra$ A $\ra$ o) $\ra$ A $\ra$ A $\ra$ A $\ra$ o.} \\
\>{\it apply}\quad{\it P}\quad{\it Arg1}\quad{\it Arg2}\quad{\it Result}\quad$\pif$\quad{\it P Arg1 Arg2 Result}.
\end{tabbing}
Using {\it append} defined before as the ``function'' that is to be
applied, the following query can be asked
\begin{tabbing}
\dquad\quad{\it ?- apply (append [A]) (1 :: nil) (2 :: nil) R.}
\end{tabbing}
Note that the occurrence of {\it append} in the above query should
be associated with the type variable {\it A} based on our type
annotation scheme. The computation of this query requires the
solution of
\begin{tabbing}
\dquad\quad{\it solve (append [A] (1 :: nil) (2 :: nil) R)},
\end{tabbing}
in the course of which the usage of {\it append} is transformed into
the head of a goal, and is decided by the compiler as one without
type annotations.

To solve this problem, the types list of a predicate constant is
carefully organized in our implementation in the way that those
required by a predicate usage of this constant but not by a
non-predicate usage should always appear before the others, and
their lengths are also recorded along with the descriptor of the
constant. This information is then taken into account by {\it solve}
in loading the arguments of the predicate constant into registers:
the types that are not needed for the predicate usage of the
constant are simply discarded.

It is interesting to contrast the treatment of types we have described
in this chapter with the one used in {\it Version 1} of the {\it
  Teyjus} system. In the latter system, types have to be maintained
not only with constants but also with logic variables; this
is necessary because the types of such variables play a role in
determining the structures of bindings calculated in unification.
Among the different ideas that we have described in this chapter for
reducing runtime type computations, the only one that is applicable in
that setting is the one based on separating a type into a skeleton and
type environment part. This optimization is actually also employed by
{\it Version 1} of the {\it Teyjus} system. From an implementation
standpoint, that system also provides
a means for representing types and it includes suitable term and type
unification instructions to support the compilation of relevant
type-related computations.  creation and unification on them.
At a detailed level, there is a difference between the representation
used for function types in our setting and in {\it Teyjus
  Version 1}. In the latter context, it is important to be able to
access the argument and target types quickly and to determine the number of
arguments in the function type; these attributes are used in
generating unifiers. To facilitate such an examination, function types
are represented in ``un-curried'' form, \ie, a type such as
$\alpha_1\ra\ldots\ra\alpha_n\ra\beta$ is represented as a pair of a
vector containing the types $\alpha_1,\ldots,\alpha_n$ and the type
$\beta$. While this representation works well in most instances, is
can occasionally cause problems. In particular, consider the situation
when $\beta$ is a type variable. In this case, it could be
instantiated with a function type, thereby allowing the vector of
arguments to become longer. Having to consider this possibility
complicates the unification computation on types and also leads to
several special instructions to facilitate the compilation of
unification with function
types. In our setting, types do not have a role to play in term
unification and hence it is not important to be able to see the
arguments and target type of a function type in any special
way. Moreover, we expect types themselves to be infrequently accessed
and, when they are accessed we expect them to be even more
infrequently complicated function types; the latter is especially true
because there is never a need in our context to look at skeleton
types whereas this is needed in the setting of {\it Teyjus Version
  1}. Consequently, we have treated the function type constructor as
just another binary function symbol with no special properties in our
representation. This also has the benefit of further simplifying our
already simple adaptation of the instruction set underlying type
unification in {\it Teyjus Version 1}.

\chapter{An Implementation of $\lambda$Prolog}\label{chp:system}
We have, at this point, presented a complete picture of an abstract
machine and compilation model that could underlie an implementation of
$\lp$. As part of this thesis, we have undertaken such an
implementation. This implementation is referred to as {\it Version 2}
of the {\it Teyjus} system or {\it
  Teyjus Version 2} for short.\footnote{As is typical of a software
  project of significant size, {\it Teyjus Version 2} has involved
  contributions from others. However, the underlying implementation
  ideas for all parts except the treatment of modularity notions in
  $\lp$ have derived from this thesis and the bulk of the compiler and
  the abstract machine emulator is also attributable to it.}
 There are three purposes for undertaking this implementation.
  First, we have wanted to provide
researchers interested in experimenting with the specification and
prototyping capabilities of $\lp$ a concrete and efficient vehicle to
use in such endeavors. {\it Teyjus Version 2} already serves this
purpose by forming a suite together with the {\it Abella} system
\cite{gacek08ijcar} that is freely distributed by our research group
to support specification, prototyping and reasoning about
specifications \cite{gacek-abella-website, teyjus.website}. Second, we
want to evaluate the design ideas that we
have developed and for this an actual implementation is
essential. Finally, we believe that there are several language related
issues that can be experimented with relative to $\lp$ and having a
concrete implementation provides the means to do this in a more
comprehensive fashion.

 In this chapter, provide a high-level description of {\it Teyjus
   Version 2}. The particular motivations for building this system
 have imposed additional conditions on its structure. For example, the
 need to make it widely accessible has meant that we pay special
 attention to its portability to different architectures and operating
 systems. Similarly, if {\it Teyjus Version 2} is to be useful for
 evaluation and language extension experiments, then it must have an
 open and easy to modify structure as a software system. Our
 discussion below highlights the impact of such considerations in the
 overall system that we have constructed.

\section{The Language Implemented}\label{sec:general_lang}
The $\lp$ language also encompasses a notion of modularity for organizing
large programs. The support of this feature is orthogonal to
the issues considered by this thesis, but a brief discussion of it is
nevertheless to providing a proper description of {\it Teyjus Version
  2}.

The notion of module underlying $\lp$ permits the space of names and
predicate definitions to be decomposed into smaller units. The
interface of each such unit is provided by a signature, which includes
the names, \ie, type constructors and constants, that are publicly
visible. The implementation of this interface constitutes an
accompanying module, that comprises the predicate definitions as well
as the declarations of the global and local names needed in the
module.  An important interaction between $\lp$ program units takes
place through the medium of module or signature {\it accumulation}
that allows the set of names and the definitions of predicates
available in a particular unit to be extended by using the
declarations in another unit.  The meaning of this construct can be
understood as inlining the contents of the accumulated signature or
module at the place of its occurrence, but only after affecting a
renaming of non-global names to avoid inadvertent and illegal
confusion.

\begin{figure}\footnotesize
\begin{tabbing}
\quad\quad\={\it accum\_sig}\quad\={\it term, form}\quad\quad\={\it type}. \kill
\>{\it sig} \> {\it logic\_base}. \\
\>{\it kind}\> {\it term, form}\> {\it type}. \\
\>\% {\it  Followed by the declarations for other logical connectives and quantifiers.}\\
\\
\>{\it sig} \> {\it logic\_vocab}. \\
\>{\it accum\_sig}\>{\it logic\_base}. \\
\>\% {\it Followed by the declarations for the constants, functions and predicates in the logic.}\\
\\
\>{\it sig} \> {\it syntax\_properties}. \\
\>{\it accum\_sig} \>{\it logic\_base}. \\
\>{\it exportdef} \> $\qf$, $\ia$ \quad$\arrxy{form}{o}$. \\
\>{\it exportdef} \> $\itm$ \>$\arrxy{term}{o}$. \\
\\[1pt]
\>{\it module}\> {\it syntax\_properties}. \\
\>\% {\it Followed by the definitions of $\qf$, $\ia$ and $\itm$.} \\
\\
\>{\it sig} \> {\it pnf}.\\
\>{\it accum\_sig} \>{\it logic\_base, logic\_vocab}.\\
\>{\it exportdef} \>{\it prenex}\> $\arrxy{form}{\arrxy{form}{o}}$.\\
\>{\it useonly}   \>{\it $\qf$, $\ia$} \quad$\arrxy{form}{o}$. \\
\>{\it useonly}   \> $\itm$ \>$\arrxy{term}{o}$. \\
\\[1pt]
\>{\it module} \> {\it pnf}. \\
\>{\it accumulate} \>{\it syntax\_properties}.\\
\>{\it accum\_sig} \>{\it logic\_base, logic\_vocab}.\\
\>{\it type} \>{\it merge}\>$\arrxy{form}{\arrxy{form}{o}}$.\\
\>\% {\it Followed by the definitions of prenex and merge.}
\end{tabbing}
\caption{A module based organization of $\qf$.}\label{fig:module_prenex}
\end{figure}

As a concrete example, we can examine how the program
{\it prenex} introduced in Section~\ref{sec:language_example}
can be organized into different modules. A conceptual consideration of
the problem to be solved leads naturally to the following four components:
\begin{enumerate}
\item a general framework for representing first-order logics, \ie,
  one that identifies the term and formula categories of expressions
  and that defines the logic connectives and quantifiers under
  consideration;
\item the specification of the vocabulary of particular versions of
  the logic, \ie, a component that identifies the sets of constant,
  function, and predicate symbols of interest;
\item a specification of syntactic properties of first-order
  formulas, such as quantifier-freeness, that are of general use in
  addition to being useful in defining the prenexing transformation;
  and
\item a specification of the particular transformations for
  calculating a prenex normal form of a given formula.
\end{enumerate}

These logical components can be mapped into the specification of the
four signature with names {\it logic\_base}, {\it logic\_vocab}, {\it
  syntax\_properties}, and {\it pnf} and the module with name {\it
  pnf} shown in Figure~\ref{fig:module_prenex}. The reading of the
displayed program
should be based on a understanding of new syntactic constructs in
the following way. First, the key word {\it sig} or {\it module}
followed by a name indicates the start of the specification of the
signature or module, respectively. Next, the accumulation of a
signature is denoted by using {\it accum\_sig} followed by the name
of the signature, whereas {\it accumulate} is used to indicate that of
the specification of the module with the name following the keyword.
Finally, {\it exportdef} and {\it useonly} combine a type declaration
with a ``boundary" description for predicate definitions: the former
indicates that {\it all} the definitions of a predicate are contained
by this module, and it is illegal to extend them in any context into
which this module is accumulated; the latter is a directive that
complements the former by specifying that the module corresponding to
the signature in which it appears (or, more directly, the module in
which it appears) may use the predicate identified but guarantees not
to extend its definition.

\section{Structure of the Implementation}\label{sec:impl_structure}

The abstract machine is realized in our implementation through a
software emulator. Thus, the overall software system has at least two
components: a compiler and an emulator. We have also chosen to channel
the interaction between the compiler and the emulator through a
bytecode file that is written to and read from memory. The support of
reading this file into the emulator so as to set the emulator in a state
where it is ready to respond to user provided queries is
realized by a third system called a loader.

An important issue to consider is what constitutes the
appropriate unit for compilation. One simple possibility, in the
context of the module system described in the previous section, is for
the compiler to inline all the accumulated signatures and modules
directly into the module being processed and to produce a bytecode
file from this (large) collection. This is, in fact, the approach used
in {\it Version 1} of the {\it Teyjus} system. However, this approach
does not provide true support of modularity, particular aspects of
which are the ability to compile and test modules separately and to
reuse the results of compilation of common modules in different
systems. In light of this fact, {\it Teyjus Version 2} supports the
ability to compile component modules separately and to realize the
combination inherent in accumulation through a separate linking
phase. Consequently, the overall system includes a fourth
component. This is a linker that has the task of looking at a
collection of (partial) bytecode files and producing from this one
complete bytecode file based on the relevant accumulation
information also contained in the starting files.

Separate compilation generally introduces difficulties in performing
global compiler optimizations because the visibility of code is
limited. In our context, at least one of the optimizations that is
directly impacted is the reduction of runtime type
associations with predicate occurrences at the heads of clauses and at
the heads of goals: the analysis discussed in
Section~\ref{sec:pred_type} for this purpose requires knowledge of the
the complete set of defining clauses for relevant predicates, but this
is not possible to have if the definition could be extended by the
code in an accumulated module that is not being looked at during
compilation of the parent module. However, the {\it exportdef}
annotation discussed in the previous section provides a
partial solution here. In particular this annotation tells the
compiler that the complete set is in fact available in relevant cases
so that it can still perform the optimization in question.

The primary function of the compiler is to translate $\lp$ modules
into bytecode form. However, it has the capability to examine $\lp$
syntax relative to the name declarations contained in a module and
this functionality is useful in one more place: in parsing user
queries. Conceptually this process works in the following way in {\it
  Teyjus Version 2}. When requested to set up for queries against the
declarations in a particular module, the top-level interface invokes
the loader to prime the emulator with the declarations in that
module. Simultaneously, the loader creates relevant symbol tables
for the compiler to use in
parsing queries relative to the vocabulary provided by the
module. Once the loading is complete, an interaction mode is
entered. In this mode, each time a user provides a query, the compiler
is invoked to parse it. The resulting structure is then returned to
the top-level system which wraps it within the {\it solve} predicate
described in Section~\ref{sec:misc} and then passes this along to the
emulator which proceeds to solve it. A fine point to note about this
scheme is that it means that top-level queries are treated in an
interpreted manner. It is also possible to compile the structures
resulting from parsing queries into bytecode form. A realization along
these lines actually has advantages over the interpretation based one
but its  development is left to future work.

We conclude this section with a discussion of two considerations that
have impacted the form of the actual implementation.

The first consideration is that we have wanted an implementation that
is easy to read and modify. This means that it is best to use a
genuinely high-level language---such as a functional or a logic
programming language---wherever this choice does not impact adversely
on efficiency. This condition holds for all those parts of the system
in which closeness to the underlying machine architecture does not
dictate the quality of performance. Specific parts that satisfy this
requirement are the compiler and the top-level interface. These
components have therefore been developed in the functional language
{\it OCaml}. On the other hand, the efficiency of the emulator does
depend on having access to aspects of the machine architecture. For
this reason the language {\it C} has been chosen for implementing this
component.\footnote{The linker and loader might well have been
  implemented in {\it OCaml} but they have in fact been implemented in
  {\it C}.} The decision to use different languages for different
components brings certain complexities to the overall
implementation. For example, the top-level interface has to rely on
the functionality of both the compiler and the emulator and hence
language inter-operability is a concern. Similarly, knowledge of
aspects such as the set of machine instructions needs to be shared
between the compiler and the emulator and such sharing should be
explicit for the ease of modification. We discuss the way
in which we have dealt with such complexities in
Section~\ref{sec:OCaml_C}.

The second consideration is the portability of our
system to different actual machine architectures. Although the {\it OCaml}
implementation naturally relieves this burden from the compiler
development, special attention is still needed on the C based
realization of the emulator to meet this goal: the low-level data
structures should be designed in a way that is not particularized to
any actual machine architecture. This topic is discussed in details
in Section~\ref{sec:term_reps_deployment}.

An interesting statistic is the sizes of the different components of
our system. The compiler comprise roughly 20,000 lines of {\it OCaml}
code whereas the emulator, the linker and the loader comprise about
26,000, 4,500 and 2,000 lines of {\it C} code, respectively.

\section{Term Representation and Portability}\label{sec:term_reps_deployment}

Portability is an important property of our system, the
consideration of which directly affects the design of the C based
emulator, in particular the realization of term and type
representations introduced in Section~\ref{sec:internal_encoding}
and Section~\ref{sec:compile_type} respectively. An conventional C
approach to realizing such encodings is to give explicit control
over the layout of the corresponding memory units by specifying bit
patterns within a word. For example, in {\it Version 1} of the {\it
  Teyjus} system that assumes that words are 32-bits long, the higher-end
4 bits of a word are used to record the category tags of terms,
additional numeric properties such as the universe indexes of
logical variables and constants are encoded by 10 bits, and the
addresses of subterms take the lower 28 bits of a word.
However, the hard-coded
bit patterns make the implementation heavily depend on the
underlying machine architecture: {\it Teyjus Version 1},
for instance, cannot run on 64-bit machines.

A natural way to eliminate this sort of hardware dependency is
to use a high-level data structure provided by the implementation
language to fulfill the encoding task, so that the decision of
actual machine memory layout can be decided by the underlying
compiler. In the context of {\it C}, structures are an encoding facility
of this sort. Based on the understanding of the alignment
rules of {\it C} compiler, the structure types corresponding to terms
and types can be
designed into a form from which the actual memory deployment closely
resemble that of the bit pattern method. For instance, a field of
unsigned 8 bit integer type can be used to encode the category tag
of terms, and by positioning this field as the first in the
structure declarations, the first 8 bits of an encoded term can be
controlled to always contain the category information; fields of
suitable types can be used for the additional information of each
term category and among them, addresses can be directly encoded as {\it C}
pointers; finally, a generic term can be used to control the minimum
size of terms so that they are always aligned to the word boundary
of the underlying machine architecture, as well as to indicate the
position of the category tag. The above discussion can be visualized
through the declarations and the corresponding space allocations shown
in Figure~\ref{fig:layout}.

\begin{figure}
\includegraphics{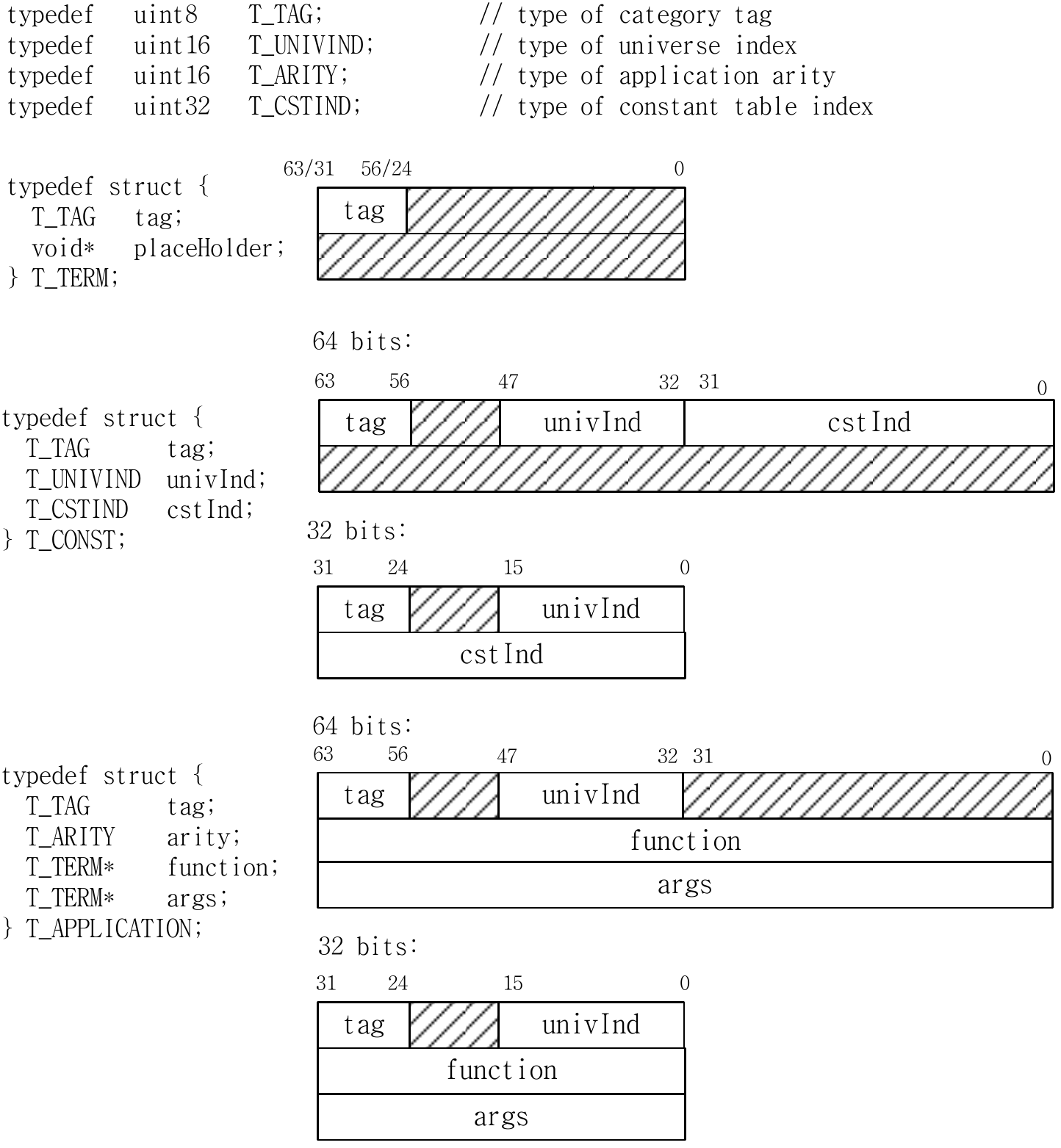}
\caption{Examples of data layout on actual machine architectures.}\label{fig:layout}
\end{figure}

The utilization of structures in {\it C} for data encodings eliminates the
dependency from our system on the word lengths of actual machine
architectures. However, this method may have undesired impacts on
the performance of the emulator. First of all, it can be observed
that the alignment of structure fields carried out by {\it C} compilers
can potentially result in gaps between useful information within a word
and makes the encoding less compacted compared with the bit-pattern
based one. Second, the recognition and decomposition of terms now
have more overhead: as opposed to simple bitwise operations, these
computations now require access to structure fields, which thereby
obtain more complicated formation and consume more CPU cycles.

The structure based approach is adopted in the realization of data
encoding in {\it Teyjus Version 2}. This approach has made our system
portable to different machine architectures, but could potentially
incur additional performance costs. Based on the primary usage of
our system, which is to serve as an experimental framework for
assessing the efficacy of implementation ideas of $\lp$, we argue
that system portability is a more important concern compared
with the possible efficiency improvement that can be obtained from
code tuning at the software development level. Moreover, it should
also be observed that the conceptual design of term and type
representations in our abstract machine does not prohibit the bit
pattern approach. When the system is used in a performance critical
context, this approach can still be adopted to hard-wire the system
to a particular machine architecture. In our software
implementation, the representations of data are encapsulated into a
separate module. The adjustments needed for changing their actual
realization is thus limited to this module and can be made without
affecting its interface and usage.

\section{Issues Related to Multiple Implementation Languages}\label{sec:OCaml_C}

As discussed in Section~\ref{sec:impl_structure}, driven by the
flexibility requirement, our compiler is realized in a high-level
language that differs from the one chosen for the other system
components. This discrepancy, however, poses implementation
challenges with regard to realizing the communication between the
compiler and the emulator and maintaining the integrity of the
software. Discussions in this section are focused on these
difficulties and our solutions to them.

The interaction between the compiler and emulator can occur in two
ways. First, the compilation result of a program has to be
eventually interpreted by the emulator. This sort of communication
is carried out indirectly through bytecode files and is consequently
not affected by the particular language choices of the system
components. However, a direct interaction
between the compiler and the emulator is needed for handling
top-level queries as discussed in Section~\ref{sec:impl_structure}.
Specifically, the runtime execution should pass from the
emulator to the compiler once a query is asked at the top-level;
after performing necessary parsing work, the compiler should
pass the result back and let the emulator take over the control
again. The representation of the query differs in the settings in
which it is needed---it is denoted as an abstract syntax tree
during compilation and should be characterized by the low-level
abstract machine data encoding in the emulator---and consequently
requires a translation from the former to the latter. A
difficulty is then introduced in realizing this process by the
choice of different implementation languages for the compiler and
emulator: the translation has to be carried across the language
boundary between {\it OCaml} and {\it C}.

One way to solve the above problem is to take advantage of the
capability {\it OCaml} has of directly manipulating the memory of {\it
  C}: with an understanding
on the emulator's data representation, the compiler can take the
full control of constructing the relevant terms and types on the
emulator's heap. However, a closer examination reveals that this
choice is not desirable. First, from the
perspective of modularity, this method unnecessarily couples the
implementation of the compiler with that of the emulator by an
agreement on the format of the emulator's data representation.
Second, it also complicates the actual software implementation by
requiring special effort to protect the segment of memory that the
compiler writes to from the garbage collector for {\it OCaml}. For
these reasons, an alternative approach is used in our implementation.
Under this scheme, the task of constructing an emulator term is
separated into smaller steps that are carried out both by the compiler
and the emulator: the compiler is responsible to provide a basic
guidance on term creation with simple information such as the term's
category
and additional numeric properties, for instance the universe index;
the actual deployment of the term into the emulator's memory and the
setting up of references to subcomponents in the graphical
representation of the term is locally maintained by the
emulator. Specifically, for each kind of
term, an {\it OCaml} function is implemented that invokes a corresponding
term creation routine of the emulator (in {\it C}). The parameter passing
between these functions is limited to data of simple types such as
integers. By recursing through the abstract syntax representation of
the term from the top-level, the compiler issues term creation
requests for each subterm through the described {\it OCaml} functions,
which eventually dispatch to the emulator's term construction
routines. When invoked, the emulator's term construction functions
make the decision on the the format of the subterm being created and
connect it to its parent according to the location information
internally maintained on a temporary stack. The actual realization
of the described scheme is based on the foreign language interface
provided by {\it OCaml}. Invocations between {\it OCaml} and {\it C} functions in both
directions are used.

In addition to the interaction issue discussed above, the choice of
multiple implementation languages causes another problem with regard
to maintaining the integrity of our software realization. In
particular, the problem arises in the encoding of concepts that
should be commonly aware by the compiler and other parts of the
system. An example of this sort is the abstract machine
instructions, which are pervasive to all the system components: they
are generated by the compiler, processed by the linker and loader
and eventually interpreted by the emulator. Consequently, a format
for their encoding should be agreed on by the entire system. Specific
information of this sort include the op-code, the number of
arguments and the representations of each kind of argument, such as
the register numbers, the environment frame offsets and the
references to other instructions. The shared view on such data
naturally requires two versions of encoding on them, which, of
course, can be simply hard coded in {\it OCaml} and {\it C} respectively.
However, the duplication of information that is conceptually the same
introduces undesirable costs in maintaining their consistency through
modifications, which could be frequently required in the course of
exploiting new design ideas of our language. To avoid this cost,
an approach based on
automatic code generation is adopted in our implementation.
Specifically, a simple high-level language is designed for the
specification of the conceptual format of instructions with constructs
that can be used to describe the relevant properties of interest. A
translator is then provided, which parses a file written in this
language and automatically generates corresponding {\it OCaml} and {\it C}
source code at the time that the system is installed. As a result, any
addition or modification of the set of instructions or their internal
structures can be made uniformly in the specification file and the
overhead of ensuring consistency between the {\it OCaml} and {\it C}
versions of encoding is eliminated from the software developers.

The issue discussed above is also pertinent to the encoding
of built-in constants (such as the set of logical constants) and type
constructors. Information about these constants such as the names,
arity, and types has to be known both to the compiler (for the purpose
of parsing and code generation) and to the emulator. A similar
translation approach has been adopted in this context as well, thereby
eliminates the replication of such information.

\chapter{Evaluating the Design}\label{chp:expr}
Our focus in this chapter is on assessing the benefits of the ideas we
have described thus far with regard to implementing
$\lambda$Prolog. There is a qualitative aspect to the improvements
these ideas bring about: they have considerably simplified the
structure of the abstract machine and have, in fact, made it possible
to think of using this machine as the target of compilation for other
higher-order logic based languages. However, the impact along this
dimensions can only be gauged indirectly, through factors such as the
relative ease with which the {\it Teyjus Version 2} system has been
developed, the extent to which this implementation is error-free and
the uses that are eventually made of the abstract machine in
implementing other related languages. A more direct and quantifiable
effect of our ideas is on system performance. The availability of two
different implementations makes it possible for us to make comparisons
and to thereby obtain an assessment as we do here.

The key choice underlying this thesis is to orient an
implementation of $\lambda$Prolog around higher-order pattern
unification instead of using the more general procedure described by
Huet. One effect of this choice is to reduce the role of types at
runtime: these types are now only needed for checking the identity of
constants that have the same name. We have also described ideas for
reducing the amount of type information that has to be dynamically
processed even further. One of our goals now is to understand the impact
of these ideas on real programs. We have constructed {\it Teyjus
  Version 2} so that we can turn on and off these type-oriented
optimizations relatively easily. We describe a set of experiments and
the conclusions we draw from doing this in this chapter.

The most interesting aspect is, however, a head-to-head comparison
with {\it Teyjus Version 1} towards gaining an understanding of the
impact on overall performance of the different choices. Some care is
needed, however, in making such a comparison. Certain choices have
been made in the implementation of {\it Teyjus Version 2} that have
the virtues of enhancing its portability and openness at the expense
of performance. A balanced contrasting of the effect of the choice in
unification procedure must factor out the impact of this auxiliary
decision. Towards this end we try first to assess the differences
between the two systems over applications that do not call on
higher-order unification and the mechanisms used to support this and
then use this information to properly understand the differences on
real higher-order applications of the language.

The rest of this chapter is structured as follows. In the first
section, we describe experiments conducted towards understanding the
impact of the choice we have made in low-level term representation. In
Section~\ref{sec:expr_type} we study the benefits of the optimizations
in the treatment of types. Section~\ref{sec:expr_unif} is devoted to a
comparison of the two different versions of {\it Teyjus} on
higher-order applications. Section~\ref{sec:expr_conclusion} concludes
the chapter with a summary of the results of our studies.

Our study in this chapter is based on actual $\lambda$Prolog programs
whose functionality and characteristics are described as relevant. The
code for all these programs can be obtained from the {\it Teyjus} web
site at \verb+http://code.google.com/p/teyjus/+.

\section{The Impact of Low-Level Term Representation}\label{sec:expr_data_reps}

The earlier version of the {\it Teyjus} system uses a highly optimized
form of representation for terms. In particular, that implementation
assumes a 32 bit word and hard-codes the use of particular parts of
such a word to encode specific components of the information contained
in the term. This knowledge is then used to define bit patterns to
extract the relevant information. Finally the use of these bit
patterns is realized through macros in the C code implementing higher
level functionality. While such a low-level encoding has performance
benefits, it also has drawbacks at the level of portability. For
example, {\it Teyjus Version 1} can be run only on 32 bit
architectures and hence cannot take benefit of newer, faster 64 bit
machines that also have larger address spaces. As another
example, since references are encoded using only a fragment of a 32
bit word, the system has to rely on special operating system
capabilities for mapping the heap onto a specific segment of a
larger memory area. A result of this is that the system cannot be
ported to a platform that is running an operating system that does not
provide such mapping capabilities.

Portability has been a major concern within {\it Teyjus Version
  2}. For this reason we have avoided bit patterns and have instead
  relied on using C based structures and a general understanding of
  how a typical C compiler maps such structures onto memory. This has
  also meant using a more expensive structure based decomposition in
  accessing relevant components of a term. Finally, to facilitate
  debugging and code clarity and modifiability, we have used function
  calls rather than macros to realize access to data fields. All of
  these choices impact on performance but none of them are essential
  to the fundamental issue of how we treat higher-order unification;
  our implementation has, in fact, been modularized so that our
  present choices concerning the low-level treatment of terms can be
  replaced by ones closer to those used in {\it Teyjus Version 1} for
  fixed architectures. Thus to get a more accurate assessment of the
  performance impact of our main ideas, it is necessary to factor out
  the effect of this auxiliary aspect.

To assess the impact of the differences in low-level representations,
a comparison was made of the performance of the two versions of the
{\it Teyjus} system on a set of $\lambda$Prolog programs. Care had to
be exercised in choosing the programs for this study. Obviously, these
programs could not be ones that also exercised higher-order aspects of
the language; it is impossible to separate out the differences arising
out of term representation choices and those resulting from the
treatment of high-order unification relative to such programs.
However, first-order programs do provide a suitable means for the
desired comparison. First-order unification obtains the same kind of
compilation and interpretive treatments in the processing model
underlying both of the systems. Moreover, it is a reasonable
hypothesis that the low-level representation choices affect
first-order and higher-order programs in a similar way. Another aspect
that we wished to factor out is the result of optimizing the treatment
of types in {\it Teyjus Version 2}. However, this was easier to do: we
needed simply to turn off the type optimizations in the newer
implementation.

The programs that we chose to use for our study based on the above
considerations are then the ones described below.

\paragraph*{Mono Naive Rev}

This program implements naive reverse on monomorphic lists that are
represented using user-defined constructors. Specifically, a new sort
{\it i} is identified, two new constants {\it mcons} of type {\it i$\
  \ra\ ($list i$)\ \ra\ ($list i$)$} and {\it mnil} of type {\it list
  i} are defined, and the predicates {\it rev} of type {\it $($list
  i$)\ \ra\ ($list i$)\ \ra\ $o} and {\it append} of type {\it $($list
  i$)\ \ra\ ($list i$)\ \ra\ ($list i$)\ \ra\ $o} are defined through
the following set of clauses:
\begin{tabbing}
\dquad{\it rev}\quad{\it mnil}\quad{\it mnil.}\\
\dquad{\it rev}\quad{\it $($mcons X L1$)$}\quad{\it L2}\quad$\pif$\\
\dquad\dquad\dquad{\it rev}\quad{\it L1}\quad{\it L3,}\quad{\it
  append}\quad{\it L3}\quad{\it $($mcons X mnil$)$}\quad{\it
  L2.}\\[5pt]
\dquad{\it append}\quad{\it mnil}\quad{\it mnil.}\quad{\it mnil}.\\
\dquad{\it append}\quad{\it $($mcons X L1$)$}\quad{\it L2}\quad{\it
  $($mcons X L3$)$}\quad$\pif$\quad {\it append}\quad{\it L1}\quad{\it
  L2}\quad{\it L3.}
\end{tabbing}
The actual testing consisted of invoking {\it rev} 30,000 times on a
collection of lists.

\paragraph*{Poly Naive Rev}

This program is a polymorphic version of the
naive reverse described above. In particular, the types of the
predicates {\it rev} and {\it append} in this instance are
\begin{tabbing}
\dquad{\it $($list A$)\ \ra\ ($list A$)\ \ra\ $o}\quad and\quad {\it $($list A$)\ \ra\ ($list A$)\ \ra\ ($list A$)\ \ra\ $o},
\end{tabbing}
An important point concerning this test case is that lists were represented
using user defined constructors called {\it pnil} and {\it pcons}
rather than the system defined list constructors {\it nil} and {\it
  ::}. The actual testing consisted of invoking {\it rev} 30,000 times on a
collection of lists.

\paragraph*{Mono Linear Rev}

This program implements tail recursive reverse on monomorphic lists.
Lists are represented the same way as in {\it Mono
Naive Rev}. The predicate {\it rev} is implemented by the following
code.
\begin{tabbing}
\dquad\={\it type}\dquad\={\it rev}\dquad\quad\={\it $($list i$)\ra($list i$)\ra$o.}\\
\>{\it rev}\quad{\it L1}\quad{\it L2}\quad$\pif$\quad{\it rev\_aux}\quad{\it L1}\quad{\it mnil}\quad{\it L2.}\\ \\
\>{\it type}\>{\it rev\_aux}\>{\it $($list i$)\ra($list i$)\ra($list i$)\ra$o.}\\
\>{\it rev\_aux}\quad{\it mnil}\quad{\it L2}\quad{\it L3.}\\
\>{\it rev\_aux}\quad{\it $($mcons X L1$)$}\quad{\it L2}\quad{\it L3}\quad$\pif$\\
\>\dquad\dquad\dquad{\it rev\_aux}\quad{\it L1}\quad{\it $($mcons X L2$)$}\quad{\it L3.}
\end{tabbing}
Testing in this case consisted of running {\it rev} 100,000 times on a
10 element list.

\paragraph*{Poly Linear Rev}
This program implements tail recursive reverse on polymorphic lists.
The predicates {\it rev} and {\it rev\_aux} have the polymorphic
types
\begin{tabbing}
\dquad{\it $($list A$)\ \ra\ ($list A$)\ \ra\ $o}\quad and\quad {\it $($list A$)\ \ra\ ($list A$)\ \ra\ ($list A$)\ \ra\ $o},
\end{tabbing}
and similar definitions to those in {\it Mono Linear Rev}. As in {\it
  Poly Naive Rev}, lists are represented in this example via user
  defined constructors. Testing consisted of running {\it rev} 100,000
  times on a  10 element list.

\paragraph*{Poly Naive Rev*}

This test case was like {\it Poly Naive Rev} except this time the
builtin representation of lists was used.

\paragraph*{Poly Linear Rev*}

This test case was like {\it Poly Linear Rev} except this time the
builtin representation of lists was used.

\paragraph*{Red Black Tree}

This program implements a polymorphic version of red-black trees. A
kind {\it btreety} of arity one is defined to categorize the
family of the trees. A type {\it color} with the two constants {\it
  red} and {\it black} is also defined. The leafs and nodes in a tree
are encoded by
constants {\it empty} and {\it node} of types
\begin{tabbing}
\dquad{\it btreety A}\quad and\\
\dquad{\it color $\ra$ A $\ra\ ($btreety A$)\ \ra\ ($btreety A$)\ \ra\ ($btreety A$)$}.
\end{tabbing}
The arguments provided to {\it node} represent the color, the left
subtree and the right subtree.
Predicates {\it add} and {\it memb} are defined to implement the
insertion and search operations respectively. Their types are
declared as
\begin{tabbing}
\dquad{\it A $\ \ra\ ($btreety A$)\ \ra\ ($btreety A$)\ \ra\ $o}\quad and \\
\dquad{\it A $\ \ra\ ($btreety A$)\ \ra\ $o}.
\end{tabbing}
The arguments of {\it add} correspond to the value to be inserted,
the original tree and the tree after insertion, respectively. The
predicate {\it memb} takes as its arguments a value and a tree that is
to be searched for this value. The testing consisted of creating a
tree of 1500 integer values and then searching for each of the values
in the tree.

\paragraph*{First-order Copy}
In this test, the program in Figure~\ref{fig:copy} for copying
$\lambda$-terms was used. However, the invocation of {\it copy} were
all restricted to first-order structures, \ie, those constructed
from only the constants {\it a} and {\it app}. Testing in this case
consisted of repeating 100,000 times the solution of the query
{\it $($copy t R$)$}, where $t$ is a first-order term of depth 4.

\begin{table}
\begin{tabular}{|l|c|c|c|} \hline\hline
{\it } & {\it\ Teyjus version 1\ } & {\it\ Teyjus version 2\ } & {\it\ Degradation\ } \\ \hline
{\it Mono Naive Rev}       & 1.51 secs  & 2.27 secs  & 50.3\%   \\
{\it Poly Naive Rev}       & 1.81 secs  & 2.80 secs  & 54.7\%   \\
{\it Mono Linear Rev}      & 1.18 secs  & 1.81 secs  & 53.4\%   \\
{\it Poly Linear Rev}      & 1.47 secs  & 2.24 secs  & 52.3\%   \\
{\it Red Black Tree}       & 2.7  secs  & 4.14 secs  & 53.3\%   \\
{\it First-order copy}     & 1.11 secs  & 1.73 secs  & 55.9\%   \\
{\it Poly Naive Rev*}      & 1.30 secs  & 1.65 secs  & 26.9\%   \\
{\it Poly Linear Rev*}     & 1.05 secs  & 1.31 secs  & 24.8\%   \\ \hline \hline
\end{tabular}
\caption{Timing comparisons on first-order programs.}\label{table:first_order}
\end{table}

Table~\ref{table:first_order} presents the results of running the test
cases described with
\begin{tabbing}
\dquad\quad{\it Teyjus Version 1 (v 1.0-b32)} and \\
\dquad\quad{\it Teyjus Version 2 (v 2.0-b2)} without type optimizations
\end{tabbing}
on a 2.6GHz 32-bit i686 processor. The numbers in the middle two
columns of the table represent the CPU time taken by the execution
of the programs. The last column of numbers denote
the performance difference between the two versions of systems,
which are calculated by the following formula.
\[\frac{execution\ time\ in\ Teyjus\ Version\ 2\ -\ execution\ time\ in\ Teyjus\ Version\ 1}{execution\ time\ in\ Teyjus\ Version\ 1}\]

The first six rows of the table indicate a fairly consistent
degradation arising out of the low-level representation used for terms
in the newer {\it Teyjus} system: averaged across these examples, the
degradation is about 53.3\%. The degradation is substantially less for
the last two cases. This result actually accords with
expectations. The builtin constructors {\it ::} and {\it nil} are
treated in a special way in our implementation model. This treatment
builds in the type optimizations for these constructors in a way that
is infeasible to turn off. Thus, in these cases the actual degradation
due to the unoptimized low-level representation of terms is partially
offset by improvements in the way types are handled. In interpreting
the results of this section, therefore, we shall disregard the data
from the last two rows in Figure~\ref{table:first_order}.

\section{Impact of Type Optimizations}\label{sec:expr_type}
As discussed in Chapter~\ref{chp:types}, there are two ways in which
the type associations that persist into execution are reduced in {\it
  Teyjus Version 2}.
First, the list of types associated with each constant occurring in
terms is reduced by eliminating instantiations for variables that
appear in the target type of the constant. Second, an analysis is
carried out over clause definitions to identify those variables in the
type of the predicates they define that have no effect on runtime
computations; it is redundant to carry along bindings for these
variables and hence these are eliminated.

A measurement of the impact of the two different levels of
types-related optimizations was conducted by turning on and off the
procedures in the compiler that effect the optimizations.
One set of programs over which testing might then be done consists of
those that are genuinely polymorphic in nature.
The test cases {\it  Poly Naive Rev}, {\it Poly
Linear Rev} and {\it Red Black Tree} introduced in the previous
section can be used as examples of this class. Another set of programs
that would be useful to test would be higher-order ones that represent
typical applications of $\lp$. The following programs were included as
representative of this class.

\paragraph*{Typeinf}
This program infers principal type schemes for ML-like
programs~\cite{Liang97let-polymorphismand}. Inside it, the
representation of the object-level types treats quantification
explicitly and utilizes abstractions to capture the binding effect.
A type inference algorithm similar to that in~\cite{DM82POPL} was
used, and the computation is specified in the $\Ll$-style.

\paragraph*{Hcinterp}
This program implements an interpreter for a language based on
first-order Horn clauses~\cite{NM98Handbook}. The declarations in
Figure~\ref{fig:prenex_formula} describe a signature for representing
such formulas. A predicate {\it interp} of type {\it form $\ \ra \
  $form$\ \ra \ $o} is defined for determining whether a given goal
formula is derivable from a conjunction of definite clauses. This
program needs higher-order features because object-level
quantification is encoded within it through abstractions. An
interesting aspect of this program in that it does not statically
fit within the higher-order pattern fragment. However, the standard
usage of this program ensures that it is dynamically in this fragment,
\ie, it is only ever necessary to solve higher-order pattern
unification problems during computation.

Polymorphic lists are used in the two higher-order programs. To focus
attention on the benefits that might be obtained from the type
optimizations, we have replaced the use of the system defined
constructors for representing these lists with the user defined
constructors {\it pcons} and {\it pnil} introduced in the previous
section.

\begin{table}
\begin{tabular}{|l|c|c|c|c|c|} \hline\hline
       & \multicolumn{5}{c|}{\it Teyjus version 2 (v 2.0-b2)} \\ \cline{2-6}
       & {\it\ \ none\ \ } & \multicolumn{2}{c|}{\it top-level } & \multicolumn{2}{c|}{\it top-level and clauses} \\ \hline
{\it Poly Naive Rev}   & 2.80 secs & 2.30 secs & $\ \ $ 17.9\% $\ \ $& 2.27 secs & 18.9\% \\
{\it Poly Linear Rev}  & 2.24 secs & 1.84 secs & 17.9\% & 1.81 secs & 19.2\% \\
{\it Red Black Tree}   & 4.14 secs & 3.80 secs & 8.2\%  & 3.78 secs & 8.7\% \\
{\it Typeinf}          & 1.27 secs & 1.20 secs & 5.5\%  & 1.20 secs & 5.5\% \\
{\it Hcinterp}         & 2.38 secs & 2.14 secs & 10.1\% & 2.14 secs & 10.1\% \\ \hline \hline
\end{tabular}
\caption{Timing comparison on type optimizations.}\label{table:type}
\end{table}

The results of our experiments are present in
Table~\ref{table:type}. The columns with tags {\it none}, {\it
top-level} and {\it top-level and clauses} denote the type
optimization levels as no type reduction, top-level constant type
reduction only, and reductions for both top-level constants and
predicate definitions respectively. The numbers of seconds in the table
correspond to the execution time of programs obtained with different
levels of type optimizations. The data for {\it Poly Naive Rev} and
{\it Poly Linear Rev} are collected from 100,000 invocations of {\it
rev} on a 10 element list of type {\it $($list i$)$}.  In the
case of {\it Red Black Tree}, the times that are measured are for
creating a tree with 1,500 integer elements and searching for each
element subsequently. The numbers in the 4th and 6th
columns indicate the percentage improvement resulting from the
different levels of type optimizations against a base that does not
use any of the optimizations. From the presented data, it
can be observed that type optimizations, especially that for
top-level constants, have a noticeable impact on first-order
polymorphic programs. The improvements in the case of the higher-order
programs is not so marked. This observation also accords with
intuitions. Many $\lp$ programs that use higher-order features
typically do so over monomorphic representations of objects, using
polymorphism only in utility predicates and data structures such as
those implementing lists. Type optimizations provide benefits only in
those situations where there is genuine use of polymorphism.

\section{Impact of Higher-Order Pattern Unification}\label{sec:expr_unif}
We now turn to measuring the effect of orienting the processing model
around higher-order pattern unification rather than using Huet's
general procedure. The testing in this context consists of comparing
the execution times of {\it Teyjus Version 1} and {\it Teyjus Version
  2} on a collection of typical $\lp$ programs. The specific programs
in our suite consisted of {\it Typeinf} and {\it Hcinterp} described
in the previous section and the following additional ones.

\paragraph*{Prenex}
This program implements a transformation from arbitrary formulas in a
first-order logic into ones that are in prenex normal
form. Abstractions in $\lambda$-terms are used to capture the binding
aspects of first-order quantifiers. The essential part of the program
is presented in Figure~\ref{fig:prenex}.

\paragraph*{Compiler}
This program implements a compiler for a small imperative language
with object-oriented features~\cite{Liang02compilerconstruction},
including a bottom up parser, a continuation passing-style
intermediate language, and generation of native byte code.

\paragraph*{Hcsyntax}
Relative to the signature specified in Figure~\ref{fig:prenex_term},
this program defines the predicates {\it goal} and {\it def\_clause}
of type {\it form $\ \ra\ $o} that serve to recognize formulas whose
syntax adhere to that of goal formulas and definite clauses in the
setting of first-order Horn clauses.

\paragraph*{Tailrec}
This program describes the encoding of a simple functional programming
language and implements a recognizer of tail recursive functions of
arbitrary arity~\cite{NM98Handbook}. The concept of scope embodied in
the object level language is explicitly encoded by abstractions, and
augment and generic goals are used to realize recursion over such
structure.

All the programs in this test suite except for {\it Hcinterp} can be
viewed as representatives of the $\Ll$-style programming. With regard
to the usage of types, the following observations can be made. The
examples {\it Prenex}, {\it Hcsyntax} and {\it Tailrec} only use
monomorphic types. Polymorphism is present in {\it Typeinf},
{\it Compiler} and {\it Hcinterp}, but as remarked in the previous
section, such usage is only relevant to the encoding of lists as
auxiliary data structures and is incidental to the essential
computation carried out by these programs. In this set of tests, we
have reverted to the use of built-in representations of lists rather
than using user defined constructors.

\begin{table}
\begin{tabular}{|l|c|c|c|c|} \hline\hline
{\it } & {\it\ Teyjus version 1\ } & \multicolumn{2}{c|}{\it\ Teyjus version 2\ } & {\it Improvement} \\ \hline
{\it Prenex}   & 3.71 secs &$\ $ 1.77 secs$\ $ &$\ $ 1.157 secs$\ $ & 68.8\%  \\
{\it Typeinf}  & 2.53 secs &$\ $ 1.16 secs$\ $ &$\ $ 0.758 secs$\ $ & 70.0\%  \\
{\it Compiler} & 2.05 secs &$\ $ 2.71 secs$\ $ &$\ $ 1.771 secs$\ $ & 13.8\%  \\
{\it Hcinterp} & 1.58 secs &$\ $ 2.14 secs$\ $ &$\ $ 1.399 secs$\ $ & 11.5\%  \\
{\it Hcsyntax} & 1.11 secs &$\ $ 1.75 secs$\ $ &$\ $ 1.144 secs$\ $ &  3.0\%  \\
{\it Tailrec}  & 1.90 secs &$\ $ 2.78 secs$\ $ &$\ $ 1.817 secs$\ $ &  4.3\%  \\ \hline \hline
\end{tabular}
\caption{Timing comparisons on $\Ll$ programs.}\label{table:hopu}
\end{table}
The results of this set of experiments are present in
Table~\ref{table:hopu}. The numbers of seconds appearing in the 2nd
and 3rd columns are the actual times taken by the execution of the
programs on the two versions of systems respectively. The numbers
appearing in the 4th column are a ``normalized'' execution time on
{\it Teyjus Version 2} obtained by correcting for the hypothesized
degradation arising from our choice of low-level term representation;
the normalization amounts to dividing the actual execution time on
{\it Teyjus version 2} by the factor $(1+53.3\%)$. The percentages in
the last column of the table corresponds to the improvement brought
about by the new system after the term encoding noise is factored out. The
calculation is carried out by the following formula.
\[\frac{normalized\ execution\ time\ in\ Teyjus\ v2\  -\ execution\ time\ in\ Teyjus\ v1}{execution\ time\ in\ Teyjus\ v1}\]

Performance improvements of varying degrees in
the different test cases can be seen to result from
using {\it Teyjus Version 2} . The execution time is substantially
reduced in the case of the first two programs. These programs use
higher-order pattern unification significantly and polymorphic typing
is not used in the first and only sparingly in the second.
Thus the better performance is attributable in these cases mostly to
the higher-order pattern unification employed in the interpretive
unification process of the emulator. In the {\it Compiler} example, a
significant part of the computation is not higher-order although there
are also parts that use $\lambda$-terms and unification in a
non-trivial way. Based on the earlier studies, we anticipate that type
optimizations contribute to about 5\%-6\% with the rest of the
improvement coming from the changed treatment of higher-order
unification. The {\it Hcinterp} program uses $\lambda$-terms and the
syntax here does not even adhere to the higher-order pattern
restriction. However, by the time unification is considered in this
case, most of the terms have, in fact, become first-order in
nature. Following the discussion in the previous section, it can also
be noticed that the improvement in this case is almost entirely
attributable to the type optimizations.
There is virtually no change in the performance observed over the last
two programs. This is also understandable. These programs embody only
an analysis of the objects they work over---first-order formulas and
functional programs in the respective cases. The L$_\lambda$ style of
programming results in the use of only first-order unification in such
analysis, higher-order pattern unification playing a role only when a
synthesis of new structure is also involved.

A question that is interesting to analyze is what particular
characteristics of unification problems in the higher-order pattern
fragment might cause a behavior difference between Huet's procedure
and a more targetted unification algorithm.
Our hypothesis, based on looking at the kinds of disagreement
pairs that actually participate in the interpretive unification
process during the execution of {\it Prenex} and {\it Typeinf}, is
that a
significant contributor to this difference is the presence during
unification of disagreement pairs of the form
\begin{tabbing}
\dquad\dquad$\dg{c_i}{(H\ c_1\ ...\ c_n)}$,
\end{tabbing}
where $H$ is a logic variable, $c_1,...,c_n$ are distinct constants
with higher universe index than $H$ and $i$ is some number between $1$
and $n$. Given such a pair, Huet's unification procedure attempts to
solve it by somewhat blindly considering bindings for $H$ of the
form $\lambda(n, \#j)$, for all $j$ such that $1\leq j\leq n$. This
gives rise to a (admittedly shallow) branching whose width in a
depth-first search setting is controlled by the particular value of
$i$, assuming that we stop the search at the first point of success.
On the other side, higher-order pattern unification treats such pairs
differently, generating the right substitution deterministically by
immediately trying to match $c_i$ to one of the constants in
$c_1,\ldots,c_n$.

To try and validate our hypothesis, we conducted an experiment using
the {\it copy} example.
The queries we used in this context were of the form {\it copy t
  Result}, where $t$ is a term with the structure
\begin{tabbing}
\dquad{\it abs $x_1\plam$ ... abs $x_n\plam$ (app $x_1$ (app $x_1$ (app $x_1$ (app $x_1$ (app $x_1$ $x_1$)))))}.
\end{tabbing}
By setting the arguments of {\it app} to $x_n$, the disagreement
pairs that are generated take the form  $\dg{c_n}{(H\ c_1\ ...\
  c_n)}$. The way substitutions are considered in {\it Teyjus Version
  1}, $(n-1)$ bindings are attempted for $H$ before the ``correct''
one for such a pair is actually found.

\begin{table}
\begin{tabular}{|c|c|c|c|c|} \hline\hline
{\it Number of abstractions } & {\it Teyjus version 1 } & \multicolumn{2}{c|}{\it\ Teyjus version 2\ } & {\it Improvement}  \\ \hline
{\it 1}        & 0.06   secs    & 0.09 secs  & 0.059 secs  & 1.7\%  \\
{\it 5}        & 0.44   secs    & 0.44 secs  & 0.287 secs  & 34.6\% \\
{\it 10}       & 0.99   secs    & 0.87 secs  & 0.569 secs  & 42.6\% \\
{\it 15}       & 1.77   secs    & 1.45 secs  & 0.945 secs  & 46.5\%  \\ \hline \hline
\end{tabular}
\caption{Effect of searching in pattern unification problems.}\label{table:hopu_copy}
\end{table}

Table~\ref{table:hopu_copy} presents the results obtained these
experiments. Execution times shown in this table result from
5,000 invocations of the given queries on the two
systems. The numbers in the 4th column are the
normalized execution times on {\it Teyjus Version 2}. The last column
denotes the performance difference obtained from viewing the
execution time on {\it Teyjus Version 1} as the basis of comparison.
An improvement that is linear to the number of abstractions can be
observed in this case.

\begin{table}
\begin{tabular}{|c|c|c|c|c|} \hline\hline
{\it Number of abstractions } & {\it Teyjus version 1 } & \multicolumn{2}{c|}{\it\ Teyjus version 2\ } & {\it Improvement}  \\ \hline
{\it 1}        & 0.06   secs    & 0.09 secs  & 0.059 secs  & 1.7\%  \\
{\it 5}        & 0.38   secs    & 0.50 secs  & 0.327 secs  & 13.9\% \\
{\it 10}       & 0.72   secs    & 0.95 secs  & 0.621 secs  & 13.7\% \\
{\it 15}       & 1.17   secs    & 1.56 secs  & 1.020 secs  & 12.9\%  \\ \hline \hline
\end{tabular}
\caption{Narrowing the effect of search in pattern unification.}\label{table:abs_copy}
\end{table}

The differences observed above could, of course, be the result of
other factors that we might have somehow overlooked in our
analysis. To try and eliminate this possibility, we conducted another
set of experiments, ones in which the pairs generated were such that
the very first substitution considered for $H$ in the {\it Teyjus
  Version 1} setting would be the right choice. Specifically, we once
again tried queries of the form {\it copy t Result}, but this time
where $t$ had the structure
\begin{tabbing}
\dquad{\it abs $x_1\plam$ ... abs $x_n\plam$ (app $x_1$ (app $x_1$ (app $x_1$ (app $x_1$ (app $x_1$ $x_1$)))))}.
\end{tabbing}
By always using the bound variable $x_1$ as the arguments of
{\it app}, the disagreement pairs generated are of the form
$\dg{c_1}{(H\ c_1\ ...\ c_n)}$. The first substitution generated for
$H$ in {\it Teyjus Version 1} succeeds for such pairs. We would
therefore expect much smaller differences with such queries.
Table~\ref{table:abs_copy} presents the results obtained from the new
experiment; execution time is measured again for 5,000
invocations of the given queries with the two versions of systems and
the different columns have the same explanations as before. The
figures in this table show much smaller differences, thereby
conforming with our expectations. Combined
with the earlier results, our hypothesis that a specific branching
behavior contributes significantly to the differences between the two
versions of the {\it Teyjus} system appears confirmed.

Before concluding this section, it is useful to understand that while
the observed responses of the two versions of the {\it Teyjus} system
agree on most practical programs and queries, they also sometimes
differ. When restricted to the L$_\lambda$ fragment of $\lp$ it is
sometimes possible that {\it Teyjus Version 1} will produce an answer
conditioned on the solutions to a remaining collection of
flexible-flexible disagreement pairs (that are known to have at least
one solution), whereas {\it Teyjus Version 2} will solve these pairs
completely. In the other direction, there are examples of
programs outside the L$_\lambda$ fragment on which {\it Teyjus Version
  1} will provide complete answers whereas {\it Teyjus Version 2} will
stop at a point short of this. As an example of this latter
kind, consider the following program
defining the predicate {\it mapfun} of type {\it $($list i$)\ \ra\
(i\ \ra\ i)\ \ra\ ($list i$)\ \ra\ $o} for some sort $i$:
\begin{tabbing}
\dquad{\it mapfun nil F nil}.\\
\dquad{\it mapfun $($X :: L1$)$ F $(($F X$)$ :: L2$)$}\quad $\pif$ {\it mapfun L1 F L2}.
\end{tabbing}
Intuitively, the predicate {\it mapfun} maps the elements in the first list
argument to those in the third by applying the function given by the
second argument.
Let {\it g} and {\it a} be constants of types {\it i $\ra$ i} and
{\it i} respectively. The disagreement pair $\dg{(F\
a)}{(g\ a)}$ that is generated in solving the query
\begin{tabbing}
\dquad{\it ?- mapfun $($a :: nil$)$ F $(($g a$)$ :: nil $)$}
\end{tabbing}
escapes the $\Ll$ subset and hence is not solved in {\it Teyjus
  Version 2}; instead it is simply produced as a remaining pair
  at the end of the computation. However, this disagreement
  pair can be successfully solved by Huet's procedure, and so, when the
  same query is provided to {\it Teyjus Version 1}, it will succeed
  with the two answer substitutions $\dg{F}{\lambdax{x}g\ x}$ and
  $\dg{F}{\lambdax{x}g\ a}$.

\section{A Summary of the Assessments}\label{sec:expr_conclusion}
We conclude this chapter by summarizing and consolidating the various
observations contained in it concerning our design ideas and the
specific realization of these in {\it Version 2} of the {\it Teyjus}
system.

One major characteristic of the new version of the {\it Teyjus} system
is its choice of low-level encoding of terms. The way we have chosen
to do this has meant a degradation in speed of about 50\%. While we
have not measured this explicitly, it is likely that space usage
is also impacted by this choice: hand-coded term representations are
bound to be significantly more compact than ones generated by the C
compiler based on structure declarations. One counter to these
drawbacks is that by letting the real code be free of low-level
decisions and hacking tricks, we have made it much more transparent,
modular and error-free. A further point to note is that special
low-level treatments can still be built in once an architecture has
been selected by changing a particular module that deals with this
issue in our implementation. A final point to note is that the way we
have dealt with this issue leads naturally to an extremely portable
system. We note in this context that such portability can also have an
important impact on the ``speed of execution'' by allowing us to use
newer and faster architectures to run $\lp$ programs. As a
specific example, recall that {\it Teyjus Version 2}, unlike {\it
  Teyjus Version 1}, can be built on 64 bit machines as well and not
just on 32 bit ones.
Table~\ref{table:portability} presents some data that is relevant in
this context. In particular, it shows the execution times for a set of
queries made against the {\it Prenex}, {\it Typeinf} and {\it
  Compiler} programs when running {\it Teyjus Version 2} on a
2.6GHZ 32-bit i686 and a 2.6GHZ 64-bits x86 processor. The performance
is noticeably better on the 64 bit architecture.

\begin{table}
\begin{tabular}{|l|c|c|c|} \hline\hline
{\it } &{\it\ \ \ \ \ \ \ 2.6GHZ 32-bit i686\ \ \ \ \ \ \ } &{\it\ \ \ \ \ \ \  2.6GHZ 64-bits x86\ \ \ \ \ \ \ } \\ \hline
{\it Prenex}                 & 2.71 secs                & 1.25 secs \\
{\it Typeinf}                & 1.16 secs                & 0.78 secs \\
{\it Compiler}               & 2.71 secs                & 1.66 secs \\
 \hline \hline
\end{tabular}
\caption{Comparing {\it Teyjus version 2} on   different
  architectures.}\label{table:portability}
\end{table}

The second kind of conclusion concerns the benefit of using
higher-order pattern unification. There are improvements from this
that take two forms. First,
this algorithm allows an efficient runtime time type processing
scheme that results in 5\% to 18\% speedups in the execution
times for a collection of first-order and practical $\Ll$ programs
that we tested. A further
observation is that the two kinds of type optimizations utilized in
our compiler do not contribute evenly to the overall performance
improvements. In fact, most of the acceleration results from the
reduction in type annotations maintained with constants; the
improvements from reductions in type associations with predicate
definitions are minor, especially for practically relevant $\lp$
applications. The second kind of advantage resulting from using
higher-order pattern unification concerns the reduction in search.
The improvement from this is large especially for $\lp$ programs used
in the intended  meta-programming tasks. At a more detailed level, our
analysis has also exposed the causes for such an improvement in the
treatment of search.

In addition to the impact on performance, orienting the implementation
around a treatment of only higher-order pattern unification has the
effect of considerably simplifying the structure of the
system. Although not directly quantifiable, the benefits from this
have been enormous. The instruction set for our abstract machine,
especially the part included for treating types, is much
simplified. The uniform nature of these instructions now makes it
possible to consider compiling other languages similar to $\lp$ to
them. The choice with regard to unification also eliminates branching
in its treatment, thereby also enormously simplifying the abstract
machine. The impact of this aspect should not be underestimated. The
need to deal with a more complex unification procedure in an efficient
fashion has made the code for {\it Teyjus Version 1} extremely
complicated and, hence, error-prone and inscrutable. By contrast, we
believe that even the realization of the abstract machine in {\it
  Teyjus Version 2} is quite penetrable and easy to maintain and
modify.

\chapter{Conclusion}\label{chp:conclusion}
In this thesis, we have considered an abstract machine and  compilation
based realization of the $\lp$ language that is oriented around
higher-order pattern unification.
We have not limited the syntax of the language in order to use this
restricted form of unification.
Rather, our approach has been to use the restriction dynamically:
while being prepared for arbitrary unification problems, an
implementation based on our ideas will solve completely only problems
in the higher-order pattern class, leaving any other problems as
constraints that are either to be solved later if subsequent
substitutions put them into the restricted class or to be reported to
the user as qualifications on answer substitution.
This approach is obviously theoretically limited in comparison with
one that uses Huet's procedure for the full class of unification
problems in that it could result in uninformative answers being
provided to the user in certain cases; we observed an example of this
kind in Section~\ref{sec:expr_unif}. However, our approach is
practically well-motivated: an empirical study of a large collection
of real programs in a $\lambda$Prolog-like setting has shown that
virtually all unification problems that are encountered during
computation are either in the higher-order pattern or in the even simpler
first-order class~\cite{MP92}. Within this context, the unification
algorithm that we use is capable of solving flexible-flexible
disagreement pairs and hence has the advantage sometimes of providing
more complete answers. From an implementation perspective, using the
restricted algorithm has the benefits of simplifying the processing
model by eliminating branching in search and greatly reducing the
runtime role of types.

At a concrete level, this thesis has developed an actual abstract
machine and compilation techniques to complement the processing model
described above. The structure that we have designed has several novel
components. First, it uses a representation of $\lambda$-terms based
on an explicit substitution calculus and it includes a reduction
procedure for these terms that is optimized to the particular context
of a higher-order logic programming language. Second, it seamlessly
integrates an interpretive treatment of higher-order unification
problems with a compilation based treatment of first-order unification
that is driven by the terms that appear in the heads of
clauses. Finally, it incorporates static analysis techniques to reduce
even further the runtime presence of types.

This thesis has also provided an actual implementation of $\lp$ based
on the design that it has proposed. This system, called {\it Teyjus
  Version 2} also has several interesting ideas. The two major
requirements that have driven its development are portability and
an openness in structure that can be exploited in extending its
capabilities and in experimenting with different low-level design
choices. These foci have led to implementation challenges that have
also been addressed. To free the implementation from architecture
specific decisions, we have pushed layout choices for terms to the C
compiler, making use of a broad understanding of such compilers to
obtain a tradeoff between efficiency and generality. To make the code
structure penetrable, we have used a genuinely high-level
language---Ocaml in this case---wherever possible in the
implementation. Since it is also imperative to use a low-level
language (typically C) for efficiency reasons in certain parts
of the system, we have had to
deal with the issue of interoperability between implementation
languages across a broad interface. An especially
interesting aspect of the code that we have developed is the manner in
which we have been able to realize the sharing of information about
instruction and general machine structure between the two languages
without tedious and error-prone replication in the two settings.

A final contribution of this thesis has been the evaluation of our
design ideas and a general understanding of the costly aspects of
higher-order unification. This part of our work has consisted of
instrumenting the new implementation and an earlier one that utilizes
Huet's original procedure for higher-order unification and of using
these two systems in a series of experiments over a relevant
collection of $\lp$ programs.

There have been four previous implementations of $\lp$ in
addition to {\it Version 1} of the {\it Teyjus} system that is
discussed in this thesis. Three of these have been interpreter based
and have used a high-level language exclusively in the realization:
specifically, in
{\it Prolog}~\cite{MN88}, {\it Lisp}~\cite{EP89} and {\it
  SML}~\cite{Elliott91asemi-functional, WM97}.
None of these systems considered in any detail the special issues that
arise in a low-level treatment of the higher-order aspects of $\lp$.
The compilation based implementations have been the
{\it Teyjus Version 1} discussed here and {\it
  Prolog/Mali}~\cite{brisset:compilation:inria:93}.
The {\it Prolog/Mali} system achieves compilation indirectly by first
translating $\lp$ programs into C code and then compiling the
resulting C code. The translation process utilizes a memory management
system called {\it Mali} that has been developed especially for logic
programming languages: in particular, translation is realized in the
form of calls to functions supported by this system. A more detailed
comparison of the treatment of the higher-order aspects to $\lp$
between the {\it Prolog/Mali} system and those in the {\it Teyjus} family
can be found in~\cite{N03treatment}.

The work in this thesis can be extended in several ways. One
interesting direction to pursue is that of incorporating a treatment of
particular cases of higher-order pattern unification into the
compilation structure rather than pushing this off entirely to the
interpretive phase. An example of where such compilation might be
useful is a situation that we discussed when analyzing the test programs
{\it Typeinf} and {\it Prenex} in Section~\ref{sec:expr_unif}. Here we
observed that a common form for disagreement pairs is
\begin{tabbing}
\dquad\dquad$\dg{t}{(H\ c_1\ ...\ c_n)}$,
\end{tabbing}
where $t$ is a first-order term, and $(H\ c_1\ ...\ c_n)$ is a term in
which $H$ is a logic variable and $c_1,...,c_n$  are distinct
constants with higher universe indexes than that of $H$.
The term $t$ is often obtained in these cases from
one of the arguments of the clause head. Compilation can therefore
utilize the structure of $t$ that is statically
available. For example, the instruction
\begin{tabbing}
\dquad\dquad {\it get\_structure\ $A_i$,\ $f$,\ $n$},
\end{tabbing}
can be enhanced so that when the dereferenced result of the term
given by
$A_i$ is actually a flexible higher-order pattern term,
the execution of the following instructions can be carried out in
a ``BND'' mode and geared towards realizing the relevant parts of the
computation described in Figure~\ref{fig:hopu_bnd}.
It can be observed from the transformation rules of {\it bnd}
that the argument list $[c_1,...,c_n]$ then has to be
carried across the instructions following the current
{\it get\_structure}.
To take a concrete example, suppose $t$ is of form $(f\ X)$,
where $f$ is a constant and $X$ is a subsequent occurrence of
a variable universally quantified at the clause head.
In the immediately following instruction {\it unify\_value}
corresponding to $X$, the list $[c_1,...,c_n]$ has to be
input to an interpretive {\it bnd} process.

This kind of passing on of the argument list of the
dynamic term to later instructions is not one that is necessary in a
first-order setting and hence has not been considered in
WAM-style compilation models. Two sorts of attempts were made
during the design of our abstract machine for realizing
this requirement, but neither of them
led to a solution that we considered satisfactory. The
unsuccessful attempts are nevertheless discussed below for
the purpose of illustrating the problems that were identified.

The first way of solving the problem that we considered is to set one
of the data registers $Ai$
to refer to an argument vector when necessary and to use this register
as an explicit argument to the subsequent instructions.
Taking the example $(f\ X)$, then we can have instructions as the following:
\begin{tabbing}
\dquad\dquad {\it get\_structure}\quad\={\it $A_1$,}\quad{\it f,}\quad{\it 1,}\quad{\it $A_{255}$}\\
\dquad\dquad {\it unify\_value}\>{\it $A_2$,}\quad{\it $A_{255}$}
\end{tabbing}
where {\it $A_{255}$} is the register holding the argument list.
However, this solution has a problem in that it adds more work to
instructions that are also used for simple first-order
unification. This form of unification is assumed to
occur much more frequently and hence this approach could
adversely affect the overall execution time.

The second method we have attempted is to use a special register,
for example, the register {\it ArgVector}, to refer to the argument vector.
This register can then be set in the execution of {\it get\_structure},
to be checked by the following instructions when necessary.
However, a closer examination on this solution reveals that it actually
requires the term $t$ from the clause head to be processed in a
depth-first manner, whereas the processing order of head unifications
underlying WAM instructions is in fact breath-first.
This can be illustrated by the following example. Suppose the head
of the clause that is to be compiled is of form
\begin{tabbing}
\dquad\dquad {\it foo ... (f (f X)) (g (g Y))},
\end{tabbing}
where $f$ and $g$ are top-level constants, and $X$ and $Y$ are second
or later occurrences of
variables universally quantified in the front of the clause.
The instructions generated in our
implementation take the following structure:
\begin{tabbing}
\dquad\dquad\={\it foo:}\= \dquad ... \\
\>{\it L1:}\quad\={\it get\_structure}\quad\={\it $A_1$,}\quad{\it f,}\quad{\it 1}\\
\>\>{\it unify\_variable}\>{\it $A_3$}\\
\>{\it L2:} \>{\it get\_structure}\>{\it $A_2$,}\quad{\it g,}\quad{\it 1} \\
\>\>{\it unify\_variable}\>{\it $A_4$}\\
\>{\it L3:} \>{\it get\_structure}\>{\it $A_3$,}\quad{\it f,}\quad{\it 1}\\
\>\>{\it unify\_value}\>{\it X}\\
\>{\it L4:} \>{\it get\_structure}\>{\it $A_4$,}\quad{\it g,}\quad{\it 1}\\
\>\>{\it unify\_value}\>{\it Y}
\end{tabbing}
Now suppose the goal to be solved
actually takes the form
\begin{tabbing}
\dquad\dquad {\it foo ... (G c1 ... cn) F},
\end{tabbing}
and further, assume the instruction {\it get\_structure} is enhanced
to deal with higher-order patterns in a way described above. Then the
execution of this
instruction at label {\it L1} sets the register {\it ArgVector} to refer to the argument
list $[c_1,...,c_n]$, which is assumed to be used by the
{\it get\_structure} and {\it unify\_value} instructions following label
{\it L3}. However, it can be observed that the execution of
{\it get\_structure} at label {\it L2} overwrites {\it ArgVector} to an
empty list.

A way to overcome this problem is to add a segment of instructions
that are only executed in the ``BND'' mode. For example, when the argument
{\it (f (f X))} of {\it foo} is considered in isolation, we can have
the following instructions generated.
\begin{tabbing}
\dquad\dquad\={\it foo:}\= \dquad ... \kill
\>{\it L1:}\quad\={\it get\_structure}\quad\={\it $A_1$,}\quad{\it f,}\quad{\it 1,}\quad{\it L5}\\
\>\>{\it unify\_variable}\>{\it $A_2$}\\
\>{\it L2:} \>{\it get\_structure}\>{\it $A_2$,}\quad{\it f,}\quad{\it 1,}\quad{\it L6} \\
\>\>{\it unify\_value}\>{\it X}\\
\>\>{\it goto}\>{\it END} \\
\\
\>{\it L5:} \>{\it unify\_variable}\>{\it $A_2$} \\
\>\>{\it bnd}\>{\it $A_2$,}\quad{\it f,}\quad{\it 1} \\
\>\>{\it unify\_value}\>{\it X} \\
\>\>{\it goto}\>{\it END} \\
\\
\>{\it L6:} \>{\it unify\_value}\>{\it X} \\
\>\>{\it goto}\>{\it END} \\
\\
\>{\it END:}\dquad ...
\end{tabbing}
In the code above, assume that the additional label argument
 to {\it get\_structure} corresponds to the start of the instruction
 sequence  that must be executed in the situation when the dynamic
 term is of the flexible higher-order pattern form
 discussed. Further, assume {\it bnd} and {\it goto} are
two new instructions. The former carries out the corresponding binding
actions in the rigid-flexible case specified in Figure~\ref{fig:hopu_bnd},
and the latter is a simple jump to the given address.

The problem with this method is obvious: viewing the entire clause head
as an application, the size of the instructions is exploded exponentially
to the total number of applications contained within it. The compilation result
is not satisfactory even for our original {\it foo} example. For
example, here we would get the rather long sequence shown below:
\begin{tabbing}
\dquad\dquad\={\it foo:}\= \dquad ... \\
\>{\it L1:}\quad\={\it get\_structure}\quad\={\it $A_1$,}\quad{\it f,}\quad{\it 1,}\quad{\it L5}\\
\>\>{\it unify\_variable}\>{\it $A_3$}\\
\>{\it L2:} \>{\it get\_structure}\>{\it $A_2$,}\quad{\it g,}\quad{\it 1,}\quad{\it L6} \\
\>\>{\it unify\_variable}\>{\it $A_4$}\\
\>{\it L3:} \>{\it get\_structure}\>{\it $A_3$,}\quad{\it f,}\quad{\it 1,}\quad{\it L7}\\
\>\>{\it unify\_value}\>{\it X}\\
\>{\it L4:} \>{\it get\_structure}\>{\it $A_4$,}\quad{\it g,}\quad{\it 1,}\quad{\it L8}\\
\>\>{\it unify\_value}\>{\it Y} \\
\>\>{\it goto}\>{\it END} \\
\\
\>{\it L5:} \>{\it unify\_variable}\>{\it $A_3$} \\
\>\>{\it bnd}\>{\it $A_3$,}\quad{\it f,}\quad{\it 1} \\
\>\>{\it unify\_value}\>{\it X} \\
\>\>{\it get\_structure}\>{\it $A_2$,}\quad{\it g,}\quad{\it 1,}\quad{\it L6} \\
\>\>{\it unify\_variable}\>{\it $A_4$}\\
\>\>{\it get\_structure}\>{\it $A_3$,}\quad{\it f,}\quad{\it 1,}\quad{\it L7}\\
\>\>{\it unify\_value}\>{\it X}\\
\>\>{\it get\_structure}\>{\it $A_4$,}\quad{\it g,}\quad{\it 1,}\quad{\it L8}\\
\>\>{\it unify\_value}\>{\it Y} \\
\>\>{\it goto}\>{\it END} \\
\\
\>{\it L6:}\dquad ...
\end{tabbing}

In the future research, the feasibilities of the methods proposed
above can be further explored. A closer study can be conducted
of the actual impact of each of them on actual $\lp$ programs.
Practical adjustments are also possible based on an empirical
assessment. For instance, the second method can be controlled in a
way such that it is only performed on the top-level structures of
the arguments of the clause head.

Another possible extension to the work in this thesis
is the reduction of the so-called {\it occurs-check} in unification.
In first-order unification, this check corresponds to examining the
structure of the term $t$ to ensure it does not contain occurrences of
the logic variable $X$ at the time when an attempt is made to bind $X$
to $t$.
This check is generalized in the context of
the higher-order pattern unification. It can be observed from
Figure~\ref{fig:hopu_mksubst} and Figure~\ref{fig:hopu_bnd} that
occurs-check is needed in unifying a pair $\dg{X}{t}$ for the
following three reasons.
\begin{enumerate}
\item The logic variable $X$ could occur in $t$, where non-unifiability
should be detected.
\item The term $t$ could contain a rigid sub-term with its head being a
constant $c$ such that $c$ resides in a universe higher than that of $X$,
which leads to non-unifiability.
\item The term $t$ could contain a flexible sub-term $(Y\ c_1\ ...\ c_n)$,
such that $Y$ resides in a universe that is lower than $X$, and the universe
levels of some constants $c_i$ in its argument list are higher than that of
$X$. In this situation, the (implicit) raising of $X$ introduces a list of
arguments which could be pruned against the arguments of $Y$.
\end{enumerate}

The performance of occurs-check is generally viewed as expensive to
execution, since otherwise, the solution of the pair $\dg{X}{t}$ can
be realized as simply binding $X$ to $t$ without any traversal
over the structure of $t$.
In the conventional implementations of {\it Prolog} and the
{\it Prolog/Mali} implementation
of $\lp$, the occurs-check is left out entirely.
In the {\it Teyjus} family of implementations, the occurs-check is performed
in the following way. A register {\it VAR} ({\it TY\_VAR} for the
first-order occurs-check on types) is used to record the variable
(type variable) for which a binding is being calculated, and is checked
against the structures of the term (type) that constitute the other
element of the
disagreement pair during the interpretive unification process.
These registers are also set in the executions of the instructions
{\it get\_structure} ({\it get\_typed\_structure}) and
{\it get\_type\_structure},
when the incoming term or type is a variable or a type variable, in which case
the computation starts to construct a first-order application (type
structure) as the binding for it, to communicate the (type) variable whose
occurrence should be checked in the interpretive unification invoked by
the following {\it unify\_value} or {\it unify\_type\_value}
corresponding to the arguments of the enclosing first-order
application or type structure.

Optimizations that are targeted towards avoiding unnecessary
occurs-check could be
significant to the performance of the implementations of our language.
In fact, one such optimization is already present in our
compilation model. This optimization happens in the compilation of
the pair $\dg{X}{t}$, where $X$ is a the first occurrence of a variable
that is universally quantified at the clause head.
In this situation, it can be observed that none of the three cases
requiring occurs-check described above can actually happen.
In particular, a new logic variable, say $X'$, with the current universe
level is introduced to replace $X$ when the clause definition is
selected to solve an incoming goal. Since $X$ is in its first occurrence,
it is impossible for $X'$  to be contained
by any other terms. Next, the universe index of $X'$ is already the largest
one in the current computation context, so that the possibility for the
second situation to occur is also eliminated.
Finally, the rasing of $X'$ against any flexible
$\Ll$-subterm $(Y\ b_1\ ...\ b_n)$ contained by $t$ results in
an argument list for $X'$ in which all the constants in $[b_1,...,b_n]$ are
contained, since $X'$ has the highest universe index, and consequently
nothing can be pruned in this argument list against $[b_1,...,b_n]$.
For these reasons, $X'$ can be immediately bound to $t$. Such a
special treatment of
bindings without occurs-check is in fact captured by the instructions
\begin{tabbing}
\dquad\dquad{\it get\_variable Ai, Aj}\dquad and\dquad {\it unify\_variable Ai},
\end{tabbing}
the execution of which simply copy the content of {\it Aj}
({\it S} for the latter) into
the register {\it Ai}. A similar optimization also exists for compiled type
unification through the usage of
\begin{tabbing}
\dquad\dquad{\it get\_type\_variable Ai, Aj}\dquad and\dquad
{\it unify\_type\_variable Ai}.
\end{tabbing}

Research in~\cite{Pientka06} and~\cite{Pientka03optimizinghigher-order}
 proposes an optimization, called {\it linearization}, for minimizing
 occurs-check similar to that in our compilation model in handling
 higher-order pattern unification within a dependently typed
 $\lambda$-calculus~\cite{Pfenning99systemdescription}.
When adopted into our context, this approach suggests a pre-processing in
compilation to translate the clause definitions into a form that any
subsequent variable occurrence in a clause head is replaced by a new
variable in its first use, with additional unifications over the new
variable with the one by which it is replaced inserted
into the beginning of the clause body.
For instance,
suppose a clause
under consideration is of form
\begin{tabbing}
\dquad\dquad {\it foo X (f X) $t$} \dquad $\pif$\dquad {\it $<$goal$>$},
\end{tabbing}
where $t$ is some arbitrary argument.
The linearization result becomes a clause
\begin{tabbing}
\dquad\dquad{\it foo X (f Z) $t$} \dquad $\pif$\dquad {\it X = Z, $<$goal$>$.}
\end{tabbing}
Within the computation context considered by~\cite{Pientka06}
and~\cite{Pientka03optimizinghigher-order},
where no compilation on unification is considered,
this approach has significant effect since after the
linearization, the bindings from a variable to a term required in the matching
of a clause head can be simply performed without occurs-check during their
interpretive computation. However, in our context, this approach actually
has almost the same effect as our special treatment on the first-occurrence of
variables described above except that computations requiring occurs-check is
further delayed till the end of the processing of the clause head.
The usefulness of this delay is arguable. Considering the example
above, suppose the argument $t$ in the clause is a constant $c$ and further
the query has the form
\begin{tabbing}
\dquad\dquad{\it ?- foo W (f (g W)) d}.
\end{tabbing}
where $d$ is a constant different from $c$. The delay of the unification
over $\dg{W}{(g\ W)}$ is beneficial here since failure will be simply recognized
from the inequality between $c$ and $d$. However, in another case,
suppose the third argument of the clause and the query are of the form
$(f\ (f\ (f\ (g\ c))))$ and $(f\ (f\ (f\ (g\ d))))$ respectively, where the
non-matching constants $c$ and $d$ are embedded deeply inside, the eager
calculation over $\dg{W}{(g\ W)}$ becomes more efficient than
actually carrying out the unification on
\begin{tabbing}
\dquad\dquad$(f\ (f\ (f\ (g\ c))))$ \dquad and\dquad $(f\ (f\ (f\ (g\ d))))$.
\end{tabbing}

A more useful solution to this problem that can be considered
is to build a mechanism to dynamically detect the absence of the three
situations requiring occurs-check described before, and perform the simple
binding when it is the case.
For example, compound terms can be
attributed with the maximum universe index of the constants contained inside,
and an additional attribute can be associated with logic variables to
indicate whether they are in their first occurrence. Such attributes should
be maintained by the unification and normalization processes for
them to have any practical value.
A specific approach of this sort is to be investigated.

In addition to improving our abstract machine and processing
structure, enhancements can also be made to the system that has been
implemented.
For example, compilation treatment can be considered for handling
the top-level queries in our system. In the absence of such
compilation, queries are restricted to not containing augment goals. A
compiled treatment would allow us to lift this restriction.
Second, the explicit treatment on the disjunctive goals by the
abstract machine discussed in Section~\ref{sec:misc} could also be
beneficial to the performance of the system. Finally, a garbage
collector for the emulator is also an important enhancement
to our system. The construction of such a garbage collector is, in
fact, currently under investigation.

Many of the implementation ideas developed in this thesis seem not
to be limited to $\lp$ and should be of use within the broader
framework of implementing higher-order features in logic programming
and reasoning systems. Specifically, these ideas may be applicable
in the context of logic programming within a dependently typed
$\lambda$-calculus~\cite{Pfenning99systemdescription}, and of
meta-theory based reasoning about computational
systems~\cite{BGMNT07, gacek08ijcar}. These kinds of systems seem to
be of growing importance within the specification and verification
realm. It would be of interest, therefore, to investigate the actual
applications of our ideas in these more general settings.

\bibliography{root}
\bibliographystyle{plain}
\end{document}